\documentclass[12pt,reqno]{amsart}
\allowdisplaybreaks[4]
\usepackage{graphicx} 
\usepackage{amssymb}
\usepackage{amsmath}
\usepackage{cite}
\usepackage{subfigure}
\usepackage{graphicx}
\usepackage{epstopdf}
\usepackage{color}

\newtheorem{corollary}{Corollary}
\newtheorem{lemma}{Lemma}
\newtheorem{remark}{Remark}
\newtheorem{theorem}{Theorem}

\begin{document}
\title[Rogue wave]{Darboux transformation of the second-type derivative nonlinear Schr\"odinger equation}
\author{Yongshuai Zhang
$^{1}$, Lijuan Guo
$^{1}$, Jingsong He$^{1}$
 and Zixiang Zhou
$^{2}$}
\thanks{$^*$ Corresponding author: hejingsong@nbu.edu.cn, jshe@ustc.edu.cn}
\dedicatory {$^1$ Department of Mathematics, Ningbo University, Ningbo, Zhejiang 315211, P.\ R.\ China\\
$^2$ School of Mathematical Sciences, Fudan University, Shanghai 200433,
  P. R. China.}
\begin{abstract}
The second-type derivative nonlinear Schr\"odinger (DNLSII) equation
was introduced as an integrable model in 1979. Very recently, the
DNLSII equation has been shown by an experiment to be a model
of the evolution of optical pulses involving self-steepening
without concomitant self-phase-modulation. In this paper the $n$-fold
  Darboux transformation (DT)  $T_n$ of the coupled DNLSII
equations is constructed in terms of determinants. Comparing with
the usual DT of the soliton equations, this kind of DT is unusual
because $T_n$ includes complicated integrals of seed
solutions in the process of iteration. By a tedious analysis, these
integrals are eliminated in $T_n$ except the integral of the
seed solution.  Moreover, this $T_n$ is reduced to the DT of the
DNLSII equation under a reduction condition.  As applications of
$T_n$, the explicit expressions of soliton, rational soliton,
breather, rogue wave and multi-rogue wave solutions for the DNLSII
equation are displayed.
\end{abstract}

\maketitle \vspace{-0.5cm}




\section{Introduction}
In 1979, Chen-Lee-Liu (CLL) linearized the nonlinear Hamiltonian
systems and identified it with the time part of the Lax equation in
order to test the integrability of nonlinear Hamiltonian systems by
inverse scattering method \cite{PS20490}. Moreover, several
integrable equations and non-integrable equations were introduced.
Among them, there is a derivative-type nonlinear equation
\begin{equation}
  \label{cll}
  {\rm i}q_t+q_{xx}+{\rm i}|q|^2q_x=0,
\end{equation}
which is integrable and is usually called Chen-Lee-Liu(CLL) equation or the
second-type derivative nonlinear Schr\"odinger equation (DNLSII).
Via the Hirota method, the explicit formula of $n$-soliton solution
of the DNLSII equation was obtained \cite{JPSJ49813, JPSJ641519}. The
Lax pair for the DNLSII equation was generalized into matrix-type and
new integrable coupled derivative nonlinear Schr\"odinger equations
were obtained \cite{PLA25753}. The Cauchy problem of the
semi-classical DNLSII equation  was studied \cite{jhlee} under a
rapidly oscillating initial profile. Except the DNLSII equation, there exists
another type of DNLS equation, i.e. the first type derivative nonlinear
Schr\"odinger equation (DNLSI) \cite{JMP19798}
\begin{equation}
   \label{kn}
  {\rm i}\tilde{q}_t+\tilde{q}_{xx}+{\rm
  i}(|\tilde{q}|^2\tilde{q})_x=0,
\end{equation}
which is connected with the DNLSII equation by a simple gauge transformation
\cite{JMP253433,JPSJ52394} of form
$\tilde{q}=q\exp\left(-\frac{\rm
i}{2}\int^x|q|^2\mathrm{d}x\right)$, with a trivial
but very important  observation $ |q|^2=|\tilde{q}|^2$, the
inverse gauge  transformation is given by
\begin{equation}\label{gt}
q=\tilde{q}\exp\left(\frac{\rm i}{2}\int^x|
\tilde{q}|^2\mathrm{d}x\right).
\end{equation}
By comparing with the intensive investigations on the DNLSI equation from the
point of views of mathematics and physics, e.g. see
\cite{ablowitzbook,rogister,molhus,mio, johnson,tzoar,anderson,buti}
and references therein, the DNLSII equation has received little
attention as mentioned above in the past decades, to the best of our
knowledge. Perhaps, there are two reasons for the opposite fate of
the two types of the DNLS equation.
\begin{itemize}
\item  The lack of physical relevance of the DNLSII equation
 for a very long time. In contrast to the direct applications  in
plasmas \cite{rogister,molhus,mio,buti}, water wave \cite{johnson}
and nonlinear optics \cite{tzoar,anderson} of the DNLSI equation, the DNLSII equation
is only known to model the propagation of the  self-steepening
optical pulses without self-phase modulation \cite{PRA76021802} in
2007.\\
\item  The doubt of necessity of the study of explicit solutions $q$ for
 the DNLSII equation. This well-known doubt is originated from that the modulus of $q$,
a solution of the DNLSII equation,  can be calculated from a known solution
$\tilde{q}$ of the DNLSI equation through the gauge transformation
eq.(\ref{gt}). The latter has been solved exactly by various methods
such as the inverse scattering transformation, bilinear method and
Darboux transformation (DT) \cite{JMP19798,JPSJ461008,JPA4013607,
PRE69066604,JPSJ68355,JPA361931,JPSJ641519,JPA44305203,guolingliudnlsI}.
\end{itemize}

However, the above two reasons are fading fast at the present time, and
then from now on the DNLSII equation deserves further studies. Physically,
the self-steepening of light pulses, originating from its propagation in
a medium with an intensity dependent index of refraction, was first
introduced in \cite{PR164312}, and later was observed in optical
pulses with possible shock formation \cite{PRL31422}.  The isolated
self-steepening is less common in nature and is often studied in
conjunction with self-phase-modulation \cite{tzoar,anderson}. In
other words, it is very difficult to observe self-steepening alone
although a controllable self-steepening has been suggested in theory
and demonstrated in experiment in 2006 through the wave vector
mismatch \cite{PRL97073903}. Soon after, at 2007, J. Moses, B. A.
Malomed and F. W. Wise \cite{PRA76021802} have reported an
experimental manifestation in optical pulses propagation involving
self-steepening without self-phase-modulation. This experiment
provides the first physical realization of the DNLSII equation, and further
shows that the DNLSII equation, like the nonlinear Schr\"odinger(NLS) equation
and DNLSI equation, is also a real and important
physical model. Although the gauge transformation \eqref{gt} gives the expression of $q$ with an integral from a known solution $\tilde{q}$, this integral makes the process of getting multi-soliton solutions very difficult. Therefore, it is an interesting problem to seek new solutions,
e.g. breather and rogue wave, of the DNLSII equation which are not generated
by gauge transformation (\ref{gt}) from corresponding known
solutions of the DNLSI equation.

The rogue waves, appearing from nowhere and disappearing without a
trace \cite{PLA373675}, have drawn much attention both from
theoretical predictions and experimental observations. Theoretically, a
possible  generating mechanisms for the rogue waves of the NLS
equation is introduced as follows \cite{arxivhe}:  the progressive fusion and fission of $n$ degenerate
breathers associated with a critical eigenvalue $\lambda_0$ creates
an order-$n$ rogue wave. The value $\lambda_0$ is a zero point of an
eigenfunction of the Lax pair of the NLS equation and it corresponds
to the limit that the period of the breather tends to infinity.
After the debate for several decades, the rogue wave has been observed
experimentally very recently in nonlinear fiber optics, in
water-wave tanks,  and in plasmas \cite{nature2007,NP6790,PRL106204502,PRL107255005}.

In addition to above important experiment works, the analytical and simple form of the rogue waves for the integrable partial differential equations(PDEs) is one of central topics in recent years. Up to now the DT becomes one of the cornerstones in the approach to construct analytical forms of the rogue waves.  More generally, the Darboux transformation, as a known world-wide method, provides a systematic way to study  the explicit solutions  of the integrable systems after the pioneering works in 1979 \cite{Matveev_1973_217,Matveev_1973_219,Matveev_1973_425,Matveev_1973_503}, which had extended a simple lemma by Darboux(1882) dealing with 1--dimensional scalar Schr\"{o}dinger equation  to the hierarchies of linear and nonlinear PDE's and their non-commutative, differential--difference and difference--difference (i.e. lattice) version. Since long time (from 1983 and until 2010) the study of the connections between the analytical
 rogue waves and integrable PDEs was limited by the very special solutions including the genuine Peregrine breather discovered by Peregrine himself in 1983 \cite{Peregrine1983} and its two next higher analogues: $P_2$ breather and $P_3$ breather discovered by Akhmediev, Eleonsky and Kulagin in 1985 \cite{Akhmediev_1985_894}  and $P_3$ breathers found by Akhmediev and his coworkers in 2009 \cite{Akhmediev_2009_26601}.
These solutions were describing the simplest rogue events in the fibers optics and hydrodynamics. At the same period Ankiewicz, Akhmediev and Clarkson(2010) produced a conjecture about the existence of the hierarchy of similar isolated solutions labeled by the integer $n$ reaching at a single point of space time the maximum of its absolute value equal
$2n+1$ provided that the asymptotic magnitude of the solution equals 1 \cite{Ankiewicz_2010_122002}.
In 2011, they succeeded to partially justify this conjecture only partially, by calculating explicitly the initial data for $P_4$ and $P_5$ by use of the DT method \cite{Kedziora_2011_56611}, but not the complete solutions. Similar solutions were found
at the same period for many other physical systems relevant to various physical applications by different authors \cite{akhmediev2010prehiota,PRE85026601,PRE86026606,hexukp2012fleq}.

In 2010,  Matveev and his coauthors  have significantly developed the technique of obtaining multi-rogue wave (MRW) solutions and
presented explicit forms for these higher order structures via explicit Wronskian formulas at the first time\cite{Dubard_2010_247}, and then
given a series papers on the properties and compact forms of MRW \cite{dubardthesis2010, matveev2011nlsKPI,Gaillard_2011_435204,Gaillard_2013_13504}. It was  conjectured in  \cite{dubardthesis2010}(see page 30)  that, for an order-$n$ rogue wave under the fundamental pattern of the NLS,
 the number of these maxima in space-time is $n(n+1)-1$ with one ``super peak'' of the height $2n+1$ surrounded by
$n(n+1)-2$ gradually decreasing peaks in two sides.  They further have conjectured that there exist other pattern  of an order-$n$ rogue wave possessing  $n(n + 1)/2$ uniform peaks \cite{dubardthesis2010, matveev2011nlsKPI,Gaillard_2011_435204}.  In \cite{matveevkpI2013} it was also shown that how to construct the initial data for KP--I
equation in order to produce the extremal rogue wave events caused by the collision of the simple rogue
waves. The attached movies gave an understanding of various scenario
of collisions of the simple rogue waves in the KP--I model, leading
in particular to short--living $2+1$ dimensional extremal rogue waves on a shallow water.
Later, many authors studied different models using Hirota method, Darboux transformation approach, and  thus the whole research of analytical
 theory for rogue wave is  growing very rapidly \cite{PRE85026601,PRE86026606,hexukp2012fleq,PRE86036604,Mu_2012_84001}.

These existed results inspire us to seek breather and rogue wave
solutions of the DNLSII equation without using the DNLSI equation. We
can further show the independent value of this physical model and
stimulate new way to observe the rogue waves by using
self-steepening effect in optical system. Meanwhile, a rogue wave
solution of the DNLSII equation may be used conveniently to observe the
self-steepening effect alone because of its high-amplitude and
localized property in $(x,t)$ plane. To this end, we are going to
construct the DT and its determinant
representation of the DNLSII equation, and then use it to generate the
intended solutions, i.e. soliton, breather and rogue wave. Although
there are several well-known results of DT for the NLS equation
\cite{Salle_1982_227,GN,teacherli,Gu1,Matveev,Gu2,he2006} and two other types of the DNLS
equation \cite{JPA44305203,xuhe2012jmpdnlsIII}, the DT of the DNLSII equation
is non-trivial because it includes an complicated integral of
seed solution, as we shall show later. This kind of DT for soliton
equations has not been reported so far as the authors know.

The present paper is organized as follows. In section 2, we give
a detailed derivation of the determinant representation $T_n$ of the
 $n$-fold DT for the coupled DNLSII equations. Particularly, the first
order DT---$T_1$ depends on an integral of the seed solution, which is
different from the usual DT. Through a cumbersome
analysis, we eliminate this kind of  integral depending on  higher order
potentials and simplify the expression of the $T_n$. The reduction of the
$T_n$ to the DNLSII equation is also obtained in this section.
In section 3, the soliton, rational soliton, breather, rogue wave and
multi-rogue wave solutions are given by applying the multi-fold DT of the
DNLSII equation. The relationship between solutions of the DNLSII equation
and the DNLSI equation is discussed. Finally, conclusions and discussions
are given in section 4.

\section{Darboux transformation}
In this section, we will construct the DT of the coupled DNLSII equations
\begin{subequations}\label{clleq}
\begin{equation}\label{qe}
{\rm i}q_t+q_{xx}+{\rm i}qrq_x=0,
\end{equation}
\begin{equation}\label{re}
{\rm i}r_t-r_{xx}+{\rm i}rqr_x=0.
\end{equation}
\end{subequations}
This system is the compatibility condition of the following Lax pair \cite{JPSJ52394}:
\begin{equation}\label{laxpair}
\left\{
\begin{aligned}
\Phi_x&=U\Phi=({\rm -i }\lambda^2+\frac{\rm i}{4}qr)J\Phi+\lambda Q\Phi,\\
\Phi_t&=V\Phi=(-2{\rm i}\lambda^4+{\rm i}qr\lambda^2+\frac{1}{4}(qr_x-rq_x)-\frac{\rm i}{8}q^2r^2)J\Phi+2\lambda^3Q\Phi+\lambda W\Phi,
\end{aligned}
\right.
\end{equation}
with
\begin{equation}
\begin{aligned}\nonumber
\Phi(x,t,\lambda)=\left(\begin{matrix} f(x,t,\lambda)\\ g(x,t,\lambda) \end{matrix} \right), \quad J=\left(\begin{matrix} 1 &0\\ 0 &-1 \end{matrix}\right), \quad
Q=\left(\begin{matrix} 0 &q\\ -r &0 \end{matrix}\right), \quad
W=\left(\begin{matrix} 0 &{\rm i}q_x-\frac{1}{2}q^2r\\ {\rm i}r_x+\frac{1}{2}r^2q &0\end{matrix}\right).
\end{aligned}
\end{equation}
Under a reduction condition $r=q^*$, two equations \eqref{clleq} are reduced to the DNLSII
equation.\\


{\noindent{\bf 2.1 The first order Darboux transformation}}\\

In general, we suppose
\begin{equation}
  \label{DT1}
  T(\lambda)=
  \left(\begin{array}{cc}
  a_1 &b_1\\
  c_1 &d_1
  \end{array}\right)\lambda+\left(\begin{array}{cc}
  a_0 &b_0\\
  c_0 &d_0
  \end{array}\right)
\end{equation}
is the first order DT of the coupled DNLSII equations,
i.e. there exist $U^{[1]}$ and $V^{[1]}$ possessing the same form as
$U$ and $V$ with $q$ and $r$ replaced by the new potentials $q^{[1]}$
and $r^{[1]}$, such that $\Phi^{[1]}=T\Phi$ satisfies
\begin{equation} \label{newlaxeqs}
\Phi^{[1]}_x=U^{[1]}\Phi^{[1]},\quad \Phi^{[1]}_t=V^{[1]}\Phi^{[1]}.
\end{equation}
As given in \cite{Matveev,Gu2}, $T(\lambda)$ must satisfy following conditions
\begin{equation}\label{Txt}
\begin{aligned}
  T_x+TU=U^{[1]}T,\quad T_t+TV=V^{[1]}T
\end{aligned}
\end{equation}
according to eq.(\ref{newlaxeqs}). Here $a_0, b_0, c_0, d_0, a_1,
b_1, c_1, d_1$ are the functions of $x$ and $t$ need to be
determined. Comparing the coefficients of $\lambda^j, (j=3,2,1,0)$  in the $x$-part of \eqref{Txt}, we obtain


\begin{subequations}\label{xpart}
\noindent $\lambda^3$:
\begin{equation}\label{xparta}
  b_1=0,\quad c_1=0,
\end{equation}
$\lambda^2$:
\begin{equation}\label{xpartb}
a_1q+2{\rm i} b_0-q^{[1]}d_1=0, \quad
r^{[1]}a_1-2{\rm i}c_0-d_1r=0,
\end{equation}
$\lambda^1$:
\begin{equation}\label{xpartc}
\begin{aligned}
  \frac{1}{4}{\rm i}a_1qr+a_{1x}-\frac{1}{4}{\rm i}a_1q^{[1]}r^{[1]}-b_0r-q^{[1]}c_0&=0,
  \quad a_0q-q^{[1]}d_0=0,\\
  \frac{1}{4}{\rm i}d_1q^{[1]}r^{[1]}+d_{1x}-\frac{1}{4}{\rm i}d_1qr+c_0q+r^{[1]}b_0&=0,
  \quad r^{[1]}a_0-d_0r=0,
\end{aligned}
\end{equation}
$\lambda^0$:
\begin{equation}\label{xpartd}
\begin{aligned}
  -\frac{1}{4}{\rm i}a_0q^{[1]}r^{[1]}+a_{0x}+\frac{1}{4}{\rm i}a_0qr&=0,
  \quad -\frac{1}{4}{\rm i}b_0qr+b_{0x}-\frac{1}{4}{\rm i}b_0q^{[1]}r^{[1]}=0,\\
  \frac{1}{4}{\rm i}c_0qr+c_{0x}+\frac{1}{4}{\rm i}c_0q^{[1]}r^{[1]}&=0,\quad
  -\frac{1}{4}{\rm i}d_0qr+d_{0x}+\frac{1}{4}{\rm i}d_0q^{[1]}r^{[1]}=0.
\end{aligned}
\end{equation}
\end{subequations}
Similarly, the $t$-part of \eqref{Txt} leads to

\begin{subequations}\label{tpart}
  \noindent $\lambda^4:$
  \begin{equation}\label{tparta}
    -2q^{[1]}d_1+2a_1q+4{\rm i}b_0=0,\quad -4{\rm i}c_0-2d_1r+2r^{[1]}a_1=0,
  \end{equation}
  $\lambda^3:$
  \begin{equation}\label{tpartb}
  \begin{aligned}
    -q^{[1]}r^{[1]}a_1{\rm i}-2b_0r-2q^{[1]}c_0+a_1qr{\rm i}&=0,&\quad 2a_0q-2q^{[1]}d_0=0,\\
    2r^{[1]}b_0+q^{[1]}r^{[1]}d_1{\rm i}+2c_0q-d_1qr{\rm i}&=0,&\quad -2d_0r+2r^{[1]}a_0=0,
    \end{aligned}
  \end{equation}
  $\lambda^2$:
  \begin{equation}\label{tpartc}
    \begin{aligned}
      a_0qr{\rm i}-q^{[1]}r^{[1]}a_0{\rm i}=0,\qquad  -d_0qr{\rm i}+q^{[1]}r^{[1]}d_0{\rm i}&=0,\\
      -q^{[1]}r^{[1]}b_0{\rm i}-\frac{1}{2}a_1q^2r-q^{[1]}_{x}d_1{\rm i}+\frac{1}{2}{q^{[1]}}^2r^{[1]}d_1-b_0qr{\rm i}+a_1q_x{\rm i}&=0,\\
      -{\rm i}r^{[1]}_{x}a_1-\frac{1}{2}q^{[1]}{r^{[1]}}^2a_1+\frac{1}{2}d_1qr^2+q^{[1]}r^{[1]}c_0{\rm i}+c_0qr{\rm i}+{\rm i}d_1r_{x}&=0,
    \end{aligned}
  \end{equation}
  $\lambda^1$:
  \begin{equation}\label{tpartd}
    \begin{aligned}
    &\frac{1}{2}{q^{[1]}}^2r^{[1]}d_0-q^{[1]}_{x}d_0{\rm i}-\frac{1}{2}a_0q^2r+a_0q_x{\rm i}=0,\quad
     -\frac{1}{2}q^{[1]}{r^{[1]}}^2a_0-r^{[1]}_{x}a_0{\rm i}+d_0r_x{\rm i}+\frac{1}{2}d_0qr^2=0,\\
    &\frac{1}{2}c_0{q^{[1]}}^2r^{[1]}-\frac{1}{4}q^{[1]}r^{[1]}_{x}a_1+\frac{1}{4}a_1qr_x+\frac{1}{2}b_0qr^2+\frac{1}{4}r^{[1]}q^{[1]}_{x}a_1
      -\frac{1}{8}{\rm i}a_1q^2r^2-c_0q^{[1]}_{x}{\rm i}-\frac{1}{4}a_1rq_x\\
      &+b_0r_x{\rm i}
      +\frac{1}{8}{\rm i}{q^{[1]}}^2{r^{[1]}}^2a_1+a_{1t}=0,\\
  &\frac{1}{4}d_1rq_x+\frac{1}{4}q^{[1]}r^{[1]}_{x}d_1+c_0q_x{\rm i}-\frac{1}{4}d_1qr_x+\frac{1}{8}{\rm i}d_1q^2r^2+d_{1t}-\frac{1}{2}c_0q^2r-b_0r^{[1]}_{x}{\rm i}-\frac{1}{8}{\rm i}{q^{[1]}}^2{r^{[1]}}^2d_1\\
  &-\frac{1}{4}r^{[1]}q^{[1]}_{x}d_1-\frac{1}{2}b_0q^{[1]}{r^{[1]}}^2=0,
    \end{aligned}
  \end{equation}
  $\lambda^0:$
  \begin{equation}\label{tparte}
    \begin{aligned}
     \frac{1}{4}r^{[1]}q_{x}a_0+a_{0t}+\frac{1}{4}a_0qr_x-\frac{1}{4}q^{[1]}r_{x}a_0+\frac{1}{8}{\rm i}{q^{[1]}}^2{r^{[1]}}^2a_0-\frac{1}{8}{\rm i}a_0q^2r^2-\frac{1}{4}a_0rq_x=0,\\
     b_{0t}+\frac{1}{4}b_0r^{[1]}q_{x}+\frac{1}{8}{\rm i}b_0{q^{[1]}}^2{r^{[1]}}^2+\frac{1}{8}{\rm i}b_0q^2r^2-\frac{1}{4}b_0qr_x+\frac{1}{4}b_0rq_x-\frac{1}{4}b_0q^{[1]}r_{x}=0,\\
     c_{0t}-\frac{1}{8}{\rm i}c_0q^2r^2+\frac{1}{4}c_0q^{[1]}r_{x}+\frac{1}{4}c_0qr_x-\frac{1}{4}c_0rq_x-\frac{1}{4}c_0r^{[1]}q_{x}
     -\frac{1}{8}{\rm i}c_0{q^{[1]}}^2{r^{[1]}}^2=0,\\
     \frac{1}{8}{\rm i}d_0q^2r^2-\frac{1}{4}d_0qr_x+\frac{1}{4}d_0rq_x+\frac{1}{4}q^{[1]}r_{x}d_0-\frac{1}{8}{\rm i}{q^{[1]}}^2{r^{[1]}}^2d_0-\frac{1}{4}r^{[1]}q^{[1]}_{x}d_0+d_{0t}=0.
    \end{aligned}
  \end{equation}
\end{subequations}

From \eqref{xpartc}, we have
$a_0=d_0=0$ in order to get  non-trivial transformation.
Therefore, suppose the first order DT is of form
form of
\begin{equation}
  \label{DT11}
  T_1(\lambda; \lambda_1)=\left(\begin{array}{cc}
  a_1 &0\\
  0  &d_1
  \end{array}\right)\lambda+\left(\begin{array}{cc}
  0  &b_0\\
  c_0&0
  \end{array}
  \right).
\end{equation}
The reason of the appearance of $\lambda_1$ in eq.(\ref{DT11}) is
that we shall use the eigenfunction associated with $\lambda_1$ to
parameterize the coefficients of $T_1$ with the help of
eqs.(\ref{xparta}) to (\ref{tparte}).

It is easy to get  $(a_1d_1)_x=0,\, (a_1d_1)_t=0,\, (b_0c_0)_x=0,
\mbox{ and } (b_0c_0)_t=0,$  from \eqref{xpartc},\,
\eqref{xpartd},\, \eqref{tpartd} and \eqref{tparte}.  So we  set $d_1=\dfrac{1}{a_1}$ and
$c_0=\dfrac{\lambda_1^2}{b_0}$, and then get
  \begin{equation}\label{a1b0}
    (a_1b_0)_x=\dfrac{1}{2}{\rm i}a_1b_0qr, \quad
     (a_1b_0)_t=-\dfrac{1}{4}a_1b_0({\rm i}q^2r^2-2qr_x+2rq_x).
  \end{equation}
\begin{lemma}
Let
\begin{equation}
  \label{HH}
  H=\exp\left(\int^{(x,\,t)}_{(x_0,\,t_0)}\frac{1}{2}\mathrm{i}
  qr\mathrm{d}x-\frac{1}{4}(\mathrm{i}q^2r^2-2qr_x+2rq_x)\mathrm{d}t\right)
\end{equation}
and
\begin{equation}\label{H}
  \dfrac{\partial H}{\partial x}=\dfrac{1}{2}{\rm i}Hqr, \quad
  \dfrac{\partial H}{\partial t}=-\dfrac{1}{4}H({\rm
  i}q^2r^2-2qr_x+2rq_x),
  \end{equation}
then $a_1b_0=\eta H$($\eta$ is an integral constant).
  \end{lemma}
  \begin{proof}
    By the coupled DNLSII equations \eqref{clleq},
    \begin{equation} \label{conservedlaw}
     (\dfrac{1}{2}{\rm i}qr)_t
     =-\dfrac{1}{4}({\rm i}q^2r^2-2qr_x+2rq_x)_x.
    \end{equation}
It is obvious  that \eqref{H} is true, and $a_1b_0=\eta H$ is a solution of \eqref{a1b0}.
\end{proof}

\begin{remark}\label{remark1}
In some special cases, the explicit expression of $H$ can be solved explicitly. For instance, if $q=r=0$,  $H$ is a constant. If $r=q^*=c\exp(\mathrm{i}(ax+bt))$, $H=\exp(-\frac{1}{4}\mathrm{i}c^2(4at+c^2t-2x))$.
\end{remark}

To determine the analytic expression of $\ T_1$,
introduce $n$  eigenfunctions $\Phi_i$ corresponding to $n$ distinct
 eigenvalues $\lambda_i$ as
\begin{equation}
\Phi_i=\left(\begin{array}{c} f_i\\ g_i\end{array}\right), \quad
f_i=f_i(x,t,\lambda_i),\quad g_i=g_i(x,t,\lambda_i),
 \quad i=1,2,...,n.
\end{equation}
Let $T_1(\lambda;\lambda_1)|_{\lambda=\lambda_1}\Phi_1=0$, that is
   \begin{equation}
     \nonumber
     a_1\lambda_1f_1+b_0g_1=0,\quad c_0f_1+d_1\lambda_1g_1=0.
   \end{equation}
By setting $\eta=-\lambda_1$ and using the identity $a_1b_0 =\eta H$ in lemma 1,
the coefficients of $T_1$ are parameterized
by eigenfunction of $\lambda_1$ as
   \begin{equation}\nonumber
     a_1=\sqrt{\dfrac{H}{h_1}},\quad d_1=\sqrt{\dfrac{h_1}{H}},\quad b_0=-\lambda_1\sqrt{Hh_1},\quad c_0=\dfrac{-\lambda_1}{\sqrt{Hh_1}},
   \end{equation}
   where $h_1={\dfrac{f_1}{g_1}}$.

\begin{theorem}\label{thm_DT1}
Suppose $\lambda_1\neq0$, and define  $H$ and $\Phi_1$ as above, then
  \begin{equation}\label{1stDTm}
    T_1=T_1(\lambda;\ \lambda_1)=\left(\begin{matrix}
     \sqrt{H}  & \\
         &\dfrac{1}{\sqrt{H}}
    \end{matrix}\right)
    \left(\begin{matrix}
      \dfrac{\lambda}{\sqrt{h_1}}   &-\lambda_1\sqrt{h_1}\\
      -\dfrac{\lambda_1}{\sqrt{h_1}} &\lambda\sqrt{{h_1}}
    \end{matrix}\right)
  \end{equation}
  is a matrix of the one-fold DT of the coupled DNLSII equations.
Correspondingly,  new potentials are given by
  \begin{equation}\label{newpotentials}
    q^{[1]}=\dfrac{H}{h_1}(q-2{\rm i}\lambda_1h_1),
    \quad
    r^{[1]}=\dfrac{h_1}{H}(r-2{\rm i}\dfrac{\lambda_1}{h_1}),
  \end{equation}
  and new eigenfunctions $\Phi_{i}^{[1]}$ associated with
$\lambda_i$  $(i\geq2)$ are given by
  \begin{equation}\label{neweigenfun}
  \Phi^{[1]}_i=\left(\begin{matrix}f^{[1]}_i \\g^{[1]}_i
  \end{matrix}\right)=\dfrac{1}{\sqrt{f_1g_1}}\left(\begin{matrix}
     \sqrt{H}  & \\
         &\dfrac{1}{\sqrt{H}}
    \end{matrix}\right)\left(\begin{matrix}
  \begin{vmatrix}
  \lambda_1f_1 &g_1\\
  \lambda_if_i &g_i
  \end{vmatrix}\\\\
  \begin{vmatrix}
  \lambda_1g_1 &f_1\\
  \lambda_ig_i &f_i
  \end{vmatrix}
  \end{matrix}\right).
  \end{equation}
\end{theorem}

\begin{proof}
By substituting $T=T_1$, $q^{[1]}=\dfrac{H}{h_1}(q-2{\rm
i}\lambda_1h_1)$ and
    $r^{[1]}=\dfrac{h_1}{H}(r-2{\rm i}\dfrac{\lambda_1}{h_1})$
into \eqref{Txt} and comparing the coefficients of $\lambda$, eqs.
 \eqref{xpart} and \eqref{tpart} are verified by using the coupled
DNLSII equations (\ref{clleq}) and the Lax pair (\ref{laxpair}).
\end{proof}

Usually $\Phi_1$ is called a generating function of $T_1$. Note that the specific
analytical expression of $H$  includes an
integral with respect to the seed solution $(q,\,r)$, which is different
from the usual DT of soliton equations. This feature may result in
difficulties in obtaining the Darboux matrix of the multi-fold DT and
the explicit solutions of the DNLSII equation. Therefore, the multi-fold DT of
the DNLSII equation needs to be studied further.



\vspace{2 ex}

{\noindent \bf{2.2 The second order Darboux transformation}}\\

In this section, we will illustrate the second order DT of the
coupled DNLSII equations using the iteration of $T_1$, and show how
to overcome the difficulty induced by the integral of the seed
solutions.

It is clear that the first order DT $T_1$, generated by
$\Phi_1$,  maps ($q, r, \Phi_i$) to ($q^{[1]}, r^{[1]},
\Phi_i^{[1]}$). In order to do DT again for this new system,
 set $h_2^{[1]}=\frac{f_2^{[1]}}{g_2^{[1]}}$, and
 define $H^{[1]}$ by \eqref{HH} where  $q$ and $r$ are replaced
 by  $q^{[1]}$ and $r^{[1]}$ respectively.
By iteration of $T_1$ once  based on
($q^{[1]},r^{[1]},\Phi_i^{[1]}$), a one-fold DT is given by
\begin{eqnarray}
T_1^{[1]}(\lambda;\ \,\lambda_2) &=\left(\begin{matrix}
     \sqrt{H^{[1]}}  & \\
         &\dfrac{1}{\sqrt{H^{[1]}}}
    \end{matrix}\right)
    \left(\begin{matrix}
      \dfrac{\lambda}{\sqrt{h_2^{[1]}}}   &-\lambda_2\sqrt{h_2^{[1]}}\\
      -\dfrac{\lambda_2}{\sqrt{h_2^{[1]}}} &\lambda\sqrt{{h_2^{[1]}}}
    \end{matrix}\right),
\end{eqnarray}
which maps ($q^{[1]},r^{[1]},\Phi_i^{[1]}$) to ($q^{[2]},r^{[2]},\Phi_i^{[2]}$) as $T_1$ does in Theorem 1. The generating function of $T_1^{[1]}(\lambda;\ \,
\lambda_2)$ is an eigenfunction $\Phi_2^{[1]}$ associated with $\lambda_2$, i.e.,
 $T_1^{[1]}(\lambda;\  \lambda_2)|_{\lambda=\lambda_2} \Phi_2^{[1]}=0$.
Furthermore, the second order DT $T_2 $ is defined by
\begin{equation}\label{2rdDT}
T_2=T_2(\lambda;\lambda_1,\lambda_2)=T_1^{[1]}(\lambda;\ \lambda_2)T_1(\lambda;\lambda_1)
\end{equation}
as a composition of $T_1$ and $T_1^{[1]}$. Note that the kernel of $T_2$ consists of
two eigenfunctions
\begin{equation}\label{eq22}
T_2|_{\lambda=\lambda_j}\Phi_j=0 \ \text{for}\ j=1,\ 2.
\end{equation}
By direct calculation, $T_2$ has a simple expression
\begin{equation*}
  T_2(\lambda;\ \lambda_1,\lambda_2)=\left(\begin{matrix}
    a_2^{[2]}&\\
    &d_2^{[2]}
  \end{matrix}\right)\lambda^2+\left(\begin{matrix}
    &b_1^{[2]}\\
    c_1^{[2]}&
  \end{matrix}\right)\lambda+\left(\begin{matrix}
    \lambda_1\lambda_2\sqrt{\dfrac{H^{[1]}h_2^{[1]}}{Hh_1}}& \\ &\lambda_1\lambda_2\sqrt{\dfrac{Hh_1}{H^{[1]}h_2^{[1]}}}
  \end{matrix}\right).\end{equation*}
Here
$a_2^{[2]},\,d_2^{[2]},\,b_1^{[2]},\,c_1^{[2]}$  are determined by four linear equations
\eqref{eq22} according to the Cramer's rule.
After the action of DT $T_2$ on the seed solution $(q,r,\Phi_i)$,
the second solution $(q^{[2]},\,r^{[2]})$ is

\begin{equation}
      \label{q2}
      \begin{aligned}
        q^{[2]}&={\dfrac{H^{[1]}}{h_1}}\left(\frac{\begin{vmatrix} \lambda_1g_1 &f_1\\  \lambda_2g_2  &f_2\end{vmatrix}}{\begin{vmatrix} \lambda_1f_1  &g_1\\ \lambda_2f_2  &g_2\end{vmatrix}}q-2{\rm i}\frac{\begin{vmatrix}\lambda_1^2f_1  &f_1\\  \lambda_2^2f_2  &f_2\end{vmatrix}}{\begin{vmatrix} \lambda_1f_1 &g_1\\  \lambda_2f_2 &g_2\end{vmatrix}}\right),\\
        r^{[2]}&={\dfrac{h_1}{H^{[1]}}}\left(\frac{\begin{vmatrix} \lambda_1f_1 &g_1\\  \lambda_2f_2  &g_2\end{vmatrix}}{\begin{vmatrix} \lambda_1g_1  &f_1\\ \lambda_2g_2  &f_2\end{vmatrix}}r-2{\rm i}\frac{\begin{vmatrix}\lambda_1^2g_1  &g_1\\  \lambda_2^2g_2  &g_2\end{vmatrix}}{\begin{vmatrix} \lambda_1g_1 &f_1\\  \lambda_2g_2 &f_2\end{vmatrix}}\right).
      \end{aligned}
    \end{equation}

In the above expression, the integral $\int q^{[1]}r^{[1]}\mathrm{d}x$ in $H^{[1]}$ is not
calculable explicitly in general. This is a crucial problem to get explicit form
of solutions $q^{[2]}$ and $r^{[2]}$ of the coupled DNLSII equations by DT.
We would fail to get higher order DT if we could not overcome this problem, because
the more complicated integral like $\int q^{[2]}r^{[2]}\mathrm{d}x$ will appear in it.
Fortunately, we find out that $\frac{H^{[1]}}{h_1}$ is a
constant (a general result will be displayed in the next subsection), so
let $\frac{H^{[1]}}{h_1}=1$ without loss of generality. Thus the two-fold Darboux matrix are as follows.
\begin{theorem}\label{thm_DT2}
 The two-fold DT of the coupled DNLSII equations is
  \begin{equation}
    \label{T2_1}
    T_2=T_2(\lambda;\lambda_1,\lambda_2)=Q_2^{-1}\left(\begin{matrix}
      (T_2)_{11} &(T_2)_{12}\\
      (T_2)_{21} &(T_2)_{22}
    \end{matrix}\right)
  \end{equation}
  with
  \begin{eqnarray*}
    Q_2=\sqrt{\begin{vmatrix}
      \lambda_1g_1 &f_1\\
      \lambda_2g_2 &f_2
    \end{vmatrix}\begin{vmatrix}
      \lambda_1f_1 &g_1\\
      \lambda_2f_2 &g_2
    \end{vmatrix}},\quad
    (T_2)_{11}&=\begin{vmatrix}
      \lambda^2              &0               &1\\
      \lambda_1^2f_1         &\lambda_1g_1    &f_1\\
      \lambda_2^2f_2         &\lambda_2g_2    &f_2\\
    \end{vmatrix},\quad
    (T_2)_{12}&=\begin{vmatrix}
      0                 &\lambda          &0\\
      \lambda_1^2f_1         &\lambda_1g_1    &f_1\\
      \lambda_2^2f_2         &\lambda_2g_2    &f_2\\
    \end{vmatrix},
    \end{eqnarray*}
    \begin{equation*}
    \begin{aligned}
    (T_2)_{21}&=\begin{vmatrix}
      0                 &\lambda          &0\\
      \lambda_1^2g_1         &\lambda_1f_1    &g_1\\
      \lambda_2^2g_2         &\lambda_2f_2    &g_2\\
    \end{vmatrix},\quad
    (T_2)_{22}&=\begin{vmatrix}
      \lambda^2              &0               &1\\
      \lambda_1^2g_1         &\lambda_1f_1    &g_1\\
      \lambda_2^2g_2         &\lambda_2f_2    &g_2\\
    \end{vmatrix}.
    \end{aligned}
    \end{equation*}
After the action of $T_2$, the transformed  eigenfunctions
    $\Phi_i^{[2]}=\left(\begin{matrix}
      f_i^{[2]}\\
      g_i^{[2]}
    \end{matrix}\right)$ corresponding to $\lambda_i$ $(i>2)$ are
    \begin{equation}
      \label{phi2_1}
      \Phi_i^{[2]}=T_2(\lambda;\lambda_1,\lambda_2)|_{\lambda=\lambda_i}\Phi_i=Q_2^{-1}\left(\begin{matrix}
      \begin{vmatrix}
         \lambda_1^2f_1 &\lambda_1g_1  &f_1\\
         \lambda_2^2f_2 &\lambda_2g_2  &f_2\\
         \lambda_i^2f_i &\lambda_ig_i  &f_i
       \end{vmatrix} &\begin{vmatrix}
         \lambda_1^2g_1 &\lambda_1f_1  &g_1\\
         \lambda_2^2g_2 &\lambda_2f_2  &g_2\\
         \lambda_i^2g_i &\lambda_if_i  &g_i
       \end{vmatrix}
      \end{matrix}\right)^T,
    \end{equation}
    and the second order solutions $(q^{[2]}, \ r^{[2]})$ for the coupled DNLSII equations
 are
    \begin{equation}
      \label{q2_1}
      \begin{aligned}
        q^{[2]}=\frac{\begin{vmatrix} \lambda_1g_1 &f_1\\  \lambda_2g_2  &f_2\end{vmatrix}}{\begin{vmatrix} \lambda_1f_1  &g_1\\ \lambda_2f_2  &g_2\end{vmatrix}}q-2{\rm i}\frac{\begin{vmatrix}\lambda_1^2f_1  &f_1\\  \lambda_2^2f_2  &f_2\end{vmatrix}}{\begin{vmatrix} \lambda_1f_1 &g_1\\  \lambda_2f_2 &g_2\end{vmatrix}},\quad
        r^{[2]}=\frac{\begin{vmatrix} \lambda_1f_1 &g_1\\  \lambda_2f_2  &g_2\end{vmatrix}}{\begin{vmatrix} \lambda_1g_1  &f_1\\ \lambda_2g_2  &f_2\end{vmatrix}}r-2{\rm i}\frac{\begin{vmatrix}\lambda_1^2g_1  &g_1\\  \lambda_2^2g_2  &g_2\end{vmatrix}}{\begin{vmatrix} \lambda_1g_1 &f_1\\  \lambda_2g_2 &f_2\end{vmatrix}}.
      \end{aligned}
    \end{equation}\\
\end{theorem}


{\noindent\bf{2.3 The $n$-fold   Darboux transformation}}\\

In this subsection, we study the determinant representation of
the $n$-fold   DT of the coupled DNLSII equations and overcome the
problem induced by the integral of the seed solution.

According to the analysis for $T_1$ and $T_1^{[1]}$, set a chain of DTs
as follows:
\begin{equation}\label{n-1iteration}
q\stackrel{T_1}{\rightarrow} q^{[1]}\stackrel{T_1^{[1]}}{\rightarrow} q^{[2]}
\stackrel{T_1^{[2]}}{\rightarrow} q^{[3]}\cdots\rightarrow\cdots q^{[j-1]}\stackrel{T_1^{[j-1]}}
{\rightarrow} q^{[j]}
\cdots\rightarrow\cdots q^{[n-1]}\stackrel{T_1^{[n-1]}}{\rightarrow} q^{[n]},
\end{equation}
with
\begin{eqnarray}
T_1^{[j-1]}(\lambda;\ \,\lambda_j) &=\left(\begin{matrix}
     \sqrt{H^{[j-1]}}  & \\
         &\dfrac{1}{\sqrt{H^{[j-1]}}}
    \end{matrix}\right)
    \left(\begin{matrix}
      \dfrac{\lambda}{\sqrt{h_j^{[j-1]}}}   &-\lambda_j\sqrt{h_j^{[j-1]}}\\
      -\dfrac{\lambda_j}{\sqrt{h_j^{[j-1]}}} &\lambda\sqrt{{h_j^{[j-1]}}}
    \end{matrix}\right)
\end{eqnarray}
which gives the $j$-fold DT in iteration.
Here $h_j^{[j-1]}=\dfrac{f_j^{[j-1]}}{g_j^{[j-1]} }$, $H^{[j-1]}$ is
defined by \eqref{HH} where  $q$ and $r$ are replaced
 by  $q^{[j-1]}$ and $r^{[j-1]}$ respectively. Note that $T_1^{[j-1]}$ is generated by
$\Phi_{j}^{[j-1]}$, $T_1^{[0]}=T_1, q^{[0]}=q,  r^{[0]}=r,\Phi_{j}^{[0]}=\Phi_j$.
The operator $T_1^{[j-1]}(\lambda;\ \,\lambda_j) $ defines map
 \begin{equation} \label{jthmap}
T_1^{[j-1]}:\left\{
\begin{array}{ll}
&q^{[j-1]} \rightarrow q^{[j]},\\
&r^{[j-1]} \rightarrow r^{[j]},\\
&\Phi_i^{[j-1]} \rightarrow \Phi_i^{[j]},
\end{array}\right.
i,j=1,2,\cdots n, \ n\geq 1,
\end{equation}
in the same way as $T_1$ does. The composition of all DTs in eq.(\ref{n-1iteration}) gives
 the $n$-fold   DT of the
coupled DNLSII equations
\begin{equation}\label{nmulitplication}
T_n= T_1^{[n-1]}\cdots  T_1^{[j-1]} \cdots T_1^{[1]} T_1.
\end{equation}
The kernel of $T_n$ consists of $\Phi_j(j=1,2,\cdots,n)$.\\

According to \eqref{DT11} and \eqref{nmulitplication}, the $n$-fold   DT $T_n$ should
 be in the form of
\begin{equation}\label{Tngeneral}
T_n=T(\lambda;\lambda_1,\lambda_2,...,\lambda_n)=\sum^{n}_{l=0}P_l\lambda^l,
\end{equation}
where
\begin{equation}
P_l\in D \mbox{ (if $l-n$ is even)},\quad P_l\in A \mbox{ (if $l-n$ is odd)},
\end{equation}
and
\begin{equation*}
\begin{aligned}
  D&=\left\{\left.\left(\begin{array}{cc}
  a &0\\
  0 &d
  \end{array}
  \right)\right|\mbox{$a,d$ are complex functions of $x$ and $t$}\right\},\\
  A&=\left\{\left.\left(\begin{array}{cc}
  0 &b\\
  c &0
  \end{array}
  \right)\right|\mbox{$b,c$ are complex functions of $x$ and $t$}\right\}.
\end{aligned}
\end{equation*}
In particular
\begin{equation*}
P_n=\left(\begin{array}{cc}
a_n^{[n]} &0\\
0  &d_n^{[n]}
\end{array}\right),\quad
P_{n-1}=\left(\begin{array}{cc}
0  &b_{n-1}^{[n]}\\
c_{n-1}^{[n]} &0
\end{array}
\right),
\end{equation*}
and
\begin{equation}
  P_0=\prod_{k=1}^{n}\left(\begin{matrix}
     &-\lambda_k\sqrt{H^{[k-1]}{h_k^{[k-1]}}}\\
    -\lambda_k\frac{1}{\sqrt{H^{[k-1]}{h_k^{[k-1]}}}} &
  \end{matrix}\right).
\end{equation}
Here $H^{[0]}=H,h_1^{[0]}=h_1$.
If $n$ is odd, then
\begin{equation*}
  P_0=\left(\begin{matrix}
     &b_0^{[n]}\\
    c_0^{[n]}&
  \end{matrix}\right)\in A,
\end{equation*}
with
\begin{eqnarray*}
  b_0^{[n]}&=&-\lambda_1\lambda_2\lambda_3\ldots\lambda_n\frac{\sqrt{{H}{h_1}H^{[2]}{h_3^{[2]}}H^{[4]}{h_5^{[4]}}\ldots H^{[n-1]}{h_n^{[n-1]}}}}{\sqrt{{H}^{[1]}h_2^{[1]}H^{[3]}h_4^{[3]}H^{[5]}h_6^{[5]}\ldots H^{[n-2]}h_{n-1}^{[n-2]}}},\\
  c_0^{[n]}&=&-\lambda_1\lambda_2\lambda_3\ldots\lambda_n\frac{\sqrt{{H}^{[1]}h_2^{[1]}H^{[3]}h_4^{[3]}H^{[5]}h_6^{[5]}\ldots H^{[n-2]}h_{n-1}^{[n-2]}}}{\sqrt{{H}{h_1}H^{[2]}{h_3^{[2]}}H^{[4]}{h_5^{[4]}}\ldots H^{[n-1]}{h_n^{[n-1]}}}}.
\end{eqnarray*}
If $n$ is even, then
\begin{equation*}
  P_0=\left(\begin{matrix}
     a_0^{[n]}&\\
    &d_0^{[n]}
  \end{matrix}\right)\in D,
\end{equation*}
with
\begin{eqnarray*}
  a_0^{[n]}&=&\lambda_1\lambda_2\lambda_3\ldots\lambda_n\frac{\sqrt{{H}^{[1]}h_2^{[1]}H^{[3]}h_4^{[3]}H^{[5]}h_6^{[5]}\ldots H^{[n-1]}h_{n}^{[n-1]}}}{\sqrt{{H}{h_1}H^{[2]}{h_3^{[2]}}H^{[4]}{h_5^{[4]}}\ldots H^{[n-2]}{h_{n-1}^{[n-2]}}}},\\
  d_0^{[n]}&=&\lambda_1\lambda_2\lambda_3\ldots\lambda_n\frac{\sqrt{{H}{h_1}H^{[2]}{h_3^{[2]}}H^{[4]}{h_5^{[4]}}\ldots H^{[n-2]}{h_{n-1}^{[n-2]}}}}{\sqrt{{H}^{[1]}h_2^{[1]}H^{[3]}h_4^{[3]}H^{[5]}h_6^{[5]}\ldots H^{[n-1]}h_{n}^{[n-1]}}}.
\end{eqnarray*}
As mentioned above, $T_n$ depends on complicated
integrals through $H^{[k]}(k=1,2,\cdots,n-1)$. These integrals lead
to a crucial problem to get explicit form of the $q^{[n]}$ and $r^{[n]}$.
Based on the observation of $H$ in lemma 1, the following theorem  provides a general
result to eliminate $H^{[j]}(j=0,1,2,\cdots,n-1)$ in order  to
avoid unfavorable integrals in $T_n$.


\begin{theorem}\label{thm_simplify}
If $i\geq k+1$, then
$\frac{g_i^{[k]}}{f_i^{[k]}}H^{[k+1]}$ is a constant.
\end{theorem}


\begin{proof}
  We need to prove
  $(\frac{g_i^{[k]}}{f_i^{[k]}}H^{[k+1]})_x=0$ and $(\frac{g_i^{[k]}}{f_i^{[k]}}H^{[k+1]})_
  t=0$. Namely,
  \begin{eqnarray}
    &&\dfrac{g_{ix}^{[k]}}{g_{i}^{[k]}}-\dfrac{f_{ix}^{[k]}}{f_{i}^{[k]}}+\dfrac{H_{x}^{[k+1]}}{H^{[k+1]}}=0\label{xproof},\\
    &&\dfrac{g_{it}^{[k]}}{g_{i}^{[k]}}-\dfrac{f_{it}^{[k]}}{f_{i}^{[k]}}+\dfrac{H_{t}^{[k+1]}}{H^{[k+1]}}=0\label{tproof}.
  \end{eqnarray}
According to Lax pair \eqref{laxpair},
  \begin{equation*}
  \begin{aligned}
   {f_{ix}^{[k]}}=\lambda_iq^{[k]}g_i^{[k]}-{\rm i}\lambda_i^2f_i^{[k]}+\dfrac{1}{4}{\rm i}q^{[k]}r^{[k]}f_i^{[k]},\\
   {g_{ix}^{[k]}}=-\lambda_ir^{[k]}f_i^{[k]}+{\rm i}\lambda_i^2g_i^{[k]}-\dfrac{1}{4}{\rm i}q^{[k]}r^{[k]}g_i^{[k]}.
   \end{aligned}
  \end{equation*}
Hence the sum of the first two terms in \eqref{xproof} is
  \begin{equation*}
    \dfrac{g_{ix}^{[k]}}{g_{i}^{[k]}}-\dfrac{f_{ix}^{[k]}}{f_{i}^{[k]}}=2{\rm i}\lambda_i^2-\dfrac{1}{2}{\rm i}q^{[k]}r^{[k]}-\lambda_iq^{[k]}\dfrac{g_{i}^{[k]}}{f_{i}^{[k]}}-\lambda_ir^{[k]}\dfrac{f_{i}^{[k]}}{g_{i}^{[k]}}.
  \end{equation*}
The third term of \eqref{xproof} can be expressed by
  \begin{equation*}
  \dfrac{H_{x}^{[k+1]}}{H^{[k+1]}}=\dfrac{1}{2}{\rm i}q^{[k+1]}r^{[k+1]}=-2{\rm i}\lambda_i^2+\dfrac{1}{2}{\rm i}q^{[k]}r^{[k]}+\lambda_iq^{[k]}\dfrac{g_i^{[k]}}{f_i^{[k]}}+\lambda_ir^{[k]}\dfrac{f_i^{[k]}}{g_i^{[k]}}.
  \end{equation*}
Here we have used \eqref{neweigenfun} and
 \begin{equation}\label{qrk+1}
    q^{[k+1]}={H^{[k]}}\dfrac{g_{i}^{[k]}}{f_{i}^{[k]}}(q^{[k]}-2{\rm i}\lambda_i\dfrac{f_{i}^{[k]}}{g_{i}^{[k]}}),
    \quad
    r^{[k+1]}=\dfrac{1}{H^{[k]}}\dfrac{f_{i}^{[k]}}{g_{i}^{[k]}}(r^{[k]}-2{\rm i}{\lambda_i}\dfrac{g_{i}^{[k]}}{f_{i}^{[k]}}),
  \end{equation}
which comes from \eqref{newpotentials} and  \eqref{jthmap}. Hence \eqref{xproof} is true.\\

On the other hand, according to the $t$-part of Lax pair \eqref{laxpair},
the first two terms of \eqref{tproof} are given by
  \begin{eqnarray}
    \dfrac{f_{it}^{[k]}}{f_i^{[k]}}=&-2{\rm i}\lambda_i^4+{\rm i}q^{[k]}r^{[k]}\lambda_i^2+\frac{1}{4}q^{[k]}r_x^{[k]}-\frac{1}{4}r^{[k]}q^{[k]}_x-\frac{1}{8}{\rm i}{q^{[k]}}^2{r^{[k]}}^2 \nonumber \\
    &+\left(2q^{[k]}\lambda_i^3-\frac{1}{2}{q^{[k]}}^2r^{[k]}\lambda_i+{\rm i}q_x^{[k]}\lambda_i\right)\dfrac{g_i^{[k]}}{f_i^{[k]}}
    , \label{2termtpart1}\\
    \dfrac{g_{it}^{[k]}}{g_i^{[k]}}=&2{\rm i}\lambda_i^4-{\rm i}q^{[k]}r^{[k]}\lambda_i^2-\frac{1}{4}q^{[k]}r_x^{[k]}+\frac{1}{4}r^{[k]}q^{[k]}_x+\frac{1}{8}{\rm i}{q^{[k]}}^2{r^{[k]}}^2 \nonumber \\
    &-\left(2r^{[k]}\lambda_i^3-\frac{1}{2}{q^{[k]}}{r^{[k]}}^2\lambda_i-{\rm i}r_x^{[k]}\lambda_i\right)\dfrac{f_i^{[k]}}{g_i^{[k]}}. \label{2termtpart2}
  \end{eqnarray}
The third term of  \eqref{tproof} is
  \begin{equation*}
    \frac{H_t^{[k+1]}}{H^{[k+1]}}=-\frac{1}{4}{\rm i}{q^{[k+1]}}^2{r^{[k+1]}}^2+\frac{1}{2}q^{[k+1]}r^{[k+1]}_x
    -\frac{1}{2}r^{[k+1]}q^{[k+1]}_x
  \end{equation*}
according to \eqref{H} and iteration \eqref{qrk+1}.
Moreover, \eqref{qrk+1} implies
\begin{equation*}
  \begin{aligned}
  q^{[k+1]}_x=&\frac{H^{[k]}g_i^{[k]}\left(2{\rm i}q^{[k]}f_i^{[k]}\lambda_i^2-g_i^{[k]}{q^{[k]}}^2\lambda_i
  +f_i^{[k]}\dfrac{\partial}{\partial{x}}q^{[k]}\right)}{{f_i^{[k]}}^2},\\
  r^{[k+1]}_x=&-\frac{f_i^{[k]}\left(2{\rm i}r^{[k]}g_i^{[k]}\lambda_i^2-f_i^{[k]}{r^{[k]}}^2\lambda_i
  -g_i^{[k]}\dfrac{\partial}{\partial{x}}r^{[k]}\right)}{H^{[k]}{g_i^{[k]}}^2}.
  \end{aligned}
  \end{equation*}
Substituting them into $\frac{H_t^{[k+1]}}{H^{[k+1]}}$ leads to
\begin{equation} \label{3rdtermtpart}
  \frac{H_t^{[k+1]}}{H^{[k+1]}}=\frac{\Gamma}{4f_i^{[k]}g_i^{[k]}}
\end{equation}
  with
  \begin{eqnarray*}
  &\Gamma=&-16{\rm i}f_i^{[k]}g_i^{[k]}\lambda_i^4+(8{f_i^{[k]}}^2r^{[k]}+8{g_i^{[k]}}^2q^{[k]})\lambda_i^3
  +8{\rm i}g_i^{[k]}f_i^{[k]}q^{[k]}r^{[k]}\lambda_i^2\\
  &&-(2{f_i^{[k]}}^2q^{[k]}{r^{[k]}}^2
  +4{\rm i}{f_i^{[k]}}^2r_x^{[k]}+2{g_i^{[k]}}^2{q^{[k]}}^2r^{[k]}
  -4{\rm i}{g_i^{[k]}}^2q_x^{[k]})\lambda_i\\
  &&-{\rm i}{g_i^{[k]}}{q^{[k]}}^2{r^{[k]}}^2f_i^{[k]}+2{g_i^{[k]}}{f_i^{[k]}}q^{[k]}
  r_x^{[k]}-2{g_i^{[k]}}{f_i^{[k]}}r^{[k]}q_x^{[k]}.
  \end{eqnarray*}
\eqref{tproof} is proved by substituting \eqref{2termtpart1}, \eqref{2termtpart2} and  \eqref{3rdtermtpart} with a cumbersome simplification.
\end{proof}


Set $\dfrac{H^{[k+1]}}{h^{[k]}_{k+1}}=1$ for $k=0,1,2,\cdots,n-2$ without loss of generality,
then
\begin{eqnarray*}
&&a_0^{[n]}= \lambda_1\lambda_2\lambda_3\ldots\lambda_n\sqrt{\frac{h_n^{[n-1]}}{H}},\quad
  d_0^{[n]}=\lambda_1\lambda_2\lambda_3\ldots\lambda_n\sqrt{\frac{H}{h_n^{[n-1]}}},\quad \mbox{if $n$ is even},\\
  &&b_0^{[n]}=-\lambda_1\lambda_2\lambda_3\ldots\lambda_n\sqrt{H h_n^{[n-1]}},\quad
  c_0^{[n]}=-\lambda_1\lambda_2\lambda_3\ldots\lambda_n\frac{1}{\sqrt{H h_n^{[n-1]}}},\quad \mbox{if $n$ is odd}.
\end{eqnarray*}
Namely,
\begin{equation} \label{P0withhn}
    P_0=\begin{cases}
    \left(\begin{matrix}
     \lambda_1\lambda_2\lambda_3\ldots\lambda_n\sqrt{\frac{h_n^{[n-1]}}{H}}&\\
   & \lambda_1\lambda_2\lambda_3\ldots\lambda_n\sqrt{\frac{H}{h_n^{[n-1]}}}
  \end{matrix}\right) &\mbox{if $n$ is even},\\\\
   \left(\begin{matrix}
     &-\lambda_1\lambda_2\lambda_3\ldots\lambda_n\sqrt{H h_n^{[n-1]}}\\
    -\lambda_1\lambda_2\lambda_3\ldots\lambda_n\frac{1}{\sqrt{H h_n^{[n-1]}}}&
  \end{matrix}\right) &\mbox{if $n$ is odd}.
  \end{cases}
\end{equation}
This expression of $P_0$ shows that many unfavorable integrals in $T_n$
are eliminated because of the disappearance of
$H^{[k]}(k=1,2,\cdots,n-1)$.  But we need to further simplify $P_0$ because there exist
hidden integrals in $h_n^{[n-1]}$.  Since $P_0 \in D$ if $n$ is even and $P_0 \in A$ if $n$ is odd, it needs to discuss expression of the $T_n$ for even $n$ and odd $n$ separately.
Here, we define a new kind of Vandermonde-type matrix to simplify the expression of the $n$-fold   DT and the $n$-fold potentials $(q^{[n]},\ r^{[n]})$.
\begin{itemize}
\item For $n=2k+1\ (k\in\mathbb{N})$
 \begin{equation}
  \label{van}
  \begin{aligned}
  W_n&=\left(\begin{matrix}
        \lambda_1^{n-1}f_1 &\lambda_1^{n-2}g_1 &\lambda_1^{n-3}f_1 &\lambda_1^{n-4}g_1 &\cdots &\lambda_1^2f_1 &\lambda_1g_1 &f_1\\
        \lambda_2^{n-1}f_2 &\lambda_2^{n-2}g_2 &\lambda_2^{n-3}f_2 &\lambda_2^{n-4}g_2 &\cdots &\lambda_2^2f_2 &\lambda_2g_2 &f_2\\
        \vdots &\vdots &\vdots &\vdots &\vdots &\vdots  &\vdots &\vdots\\
        \lambda_n^{n-1}f_n &\lambda_n^{n-2}g_n &\lambda_n^{n-3}f_n &\lambda_n^{n-4}g_n &\cdots &\lambda_n^2f_n &\lambda_ng_n &f_n\\
  \end{matrix}\right),\\
  \widehat{W_n}&=\left(\begin{matrix}
        \lambda_1^{n-1}g_1 &\lambda_1^{n-2}f_1 &\lambda_1^{n-3}g_1 &\lambda_1^{n-4}f_1 &\cdots &\lambda_1^2g_1 &\lambda_1f_1 &g_1\\
        \lambda_2^{n-1}g_2 &\lambda_2^{n-2}f_2 &\lambda_2^{n-3}g_2 &\lambda_2^{n-4}f_2 &\cdots &\lambda_2^2g_2 &\lambda_2f_2 &g_2\\
        \vdots &\vdots &\vdots &\vdots  &\vdots &\vdots  &\vdots &\vdots\\
        \lambda_n^{n-1}g_n &\lambda_n^{n-2}f_n &\lambda_n^{n-3}g_n &\lambda_n^{n-4}f_n &\cdots &\lambda_n^2g_n &\lambda_nf_n &g_n\\
  \end{matrix}\right),\\
  {M_n}&=\left(\begin{matrix}
        \lambda_1^{n}f_1 &\lambda_1^{n-2}f_1 &\lambda_1^{n-3}g_1 &\lambda_1^{n-4}f_1 &\cdots &\lambda_1^2g_1 &\lambda_1f_1 &g_1\\
        \lambda_2^{n}f_2 &\lambda_2^{n-2}f_2 &\lambda_2^{n-3}g_2 &\lambda_2^{n-4}f_2 &\cdots &\lambda_2^2g_2 &\lambda_2f_2 &g_2\\
        \vdots &\vdots &\vdots &\vdots &\vdots &\vdots  &\vdots &\vdots\\
        \lambda_n^{n}f_n &\lambda_n^{n-2}f_n &\lambda_n^{n-3}g_n &\lambda_n^{n-4}f_n &\cdots &\lambda_n^2g_n &\lambda_nf_n &g_n\\
  \end{matrix}\right),\\
  \widehat{M_n}&=\left(\begin{matrix}
        \lambda_1^{n}g_1 &\lambda_1^{n-2}g_1 &\lambda_1^{n-3}f_1 &\lambda_1^{n-4}g_1 &\cdots &\lambda_1^2f_1 &\lambda_1g_1 &f_1\\
        \lambda_2^{n}g_2 &\lambda_2^{n-2}g_2 &\lambda_2^{n-3}f_2 &\lambda_2^{n-4}g_2 &\cdots &\lambda_2^2f_2 &\lambda_2g_2 &f_2\\
        \vdots &\vdots &\vdots &\vdots &\vdots &\vdots &\vdots &\vdots\\
        \lambda_n^{n}g_n &\lambda_n^{n-2}g_n &\lambda_n^{n-3}f_n &\lambda_n^{n-4}g_n &\cdots &\lambda_n^2f_n &\lambda_ng_n &f_n\\
  \end{matrix}\right).
  \end{aligned}
\end{equation}
\item For $n=2k \ (k\in\mathbb{N}^*)$
\begin{equation}
  \label{van1}
  \begin{aligned}
    W_n&=\left(\begin{matrix}
        \lambda_1^{n-1}f_1 &\lambda_1^{n-2}g_1 &\lambda_1^{n-3}f_1 &\lambda_1^{n-4}g_1 &\cdots &\lambda_1^2g_1 &\lambda_1f_1 &g_1\\
        \lambda_2^{n-1}f_2 &\lambda_2^{n-2}g_2 &\lambda_2^{n-3}f_2 &\lambda_2^{n-4}g_2 &\cdots &\lambda_2^2g_2 &\lambda_2f_2 &g_2\\
        \vdots &\vdots &\vdots &\vdots &\vdots &\vdots &\vdots &\vdots\\
        \lambda_n^{n-1}f_n &\lambda_n^{n-2}g_n &\lambda_n^{n-3}f_n &\lambda_n^{n-4}g_n &\cdots &\lambda_n^2g_n &\lambda_nf_n &g_n\\
  \end{matrix}\right),\\
  \widehat{W_n}&=\left(\begin{matrix}
        \lambda_1^{n-1}g_1 &\lambda_1^{n-2}f_1 &\lambda_1^{n-3}g_1 &\lambda_1^{n-4}f_1 &\cdots &\lambda_1^2f_1 &\lambda_1g_1 &f_1\\
        \lambda_2^{n-1}g_2 &\lambda_2^{n-2}f_2 &\lambda_2^{n-3}g_2 &\lambda_2^{n-4}f_2 &\cdots &\lambda_2^2f_2 &\lambda_2g_2 &f_2\\
        \vdots &\vdots &\vdots &\vdots &\vdots &\vdots &\vdots &\vdots\\
        \lambda_n^{n-1}g_n &\lambda_n^{n-2}f_n &\lambda_n^{n-3}g_n &\lambda_n^{n-4}f_n &\cdots &\lambda_n^2f_n &\lambda_ng_n &f_n\\
  \end{matrix}\right),\\
   \end{aligned}
  \end{equation}
  \begin{equation}\nonumber
  \begin{aligned}
  {M_n}&=\left(\begin{matrix}
        \lambda_1^{n}f_1 &\lambda_1^{n-2}f_1 &\lambda_1^{n-3}g_1 &\lambda_1^{n-4}f_1 &\cdots &\lambda_1^2f_1 &\lambda_1g_1 &f_1\\
        \lambda_2^{n}f_2 &\lambda_2^{n-2}f_2 &\lambda_2^{n-3}g_2 &\lambda_2^{n-4}f_2 &\cdots &\lambda_2^2f_2 &\lambda_2g_2 &f_2\\
        \vdots &\vdots &\vdots &\vdots &\vdots &\vdots &\vdots &\vdots\\
        \lambda_n^{n}f_n &\lambda_n^{n-2}f_n &\lambda_n^{n-3}g_n &\lambda_n^{n-4}f_n &\cdots &\lambda_n^2f_n &\lambda_ng_n &f_n\\
  \end{matrix}\right),\\
  \widehat{M_n}&=\left(\begin{matrix}
        \lambda_1^{n}g_1 &\lambda_1^{n-2}g_1 &\lambda_1^{n-3}f_1 &\lambda_1^{n-4}g_1 &\cdots &\lambda_1^2g_1 &\lambda_1f_1 &g_1\\
        \lambda_2^{n}g_2 &\lambda_2^{n-2}g_2 &\lambda_2^{n-3}f_2 &\lambda_2^{n-4}g_2 &\cdots &\lambda_2^2g_2 &\lambda_2f_2 &g_1\\
        \vdots &\vdots &\vdots &\vdots &\vdots &\vdots &\vdots &\vdots\\
        \lambda_n^{n}g_n &\lambda_n^{n-2}g_n &\lambda_n^{n-3}f_n &\lambda_n^{n-4}g_n &\cdots &\lambda_n^2g_n &\lambda_nf_n &g_n\\
  \end{matrix}\right).
  \end{aligned}
\end{equation}
\end{itemize}
In particular
\begin{equation}
  \nonumber
  W_2=\left(\begin{matrix}
    \lambda_1f_1 &g_1\\
    \lambda_2f_2 &g_2
  \end{matrix}\right),\quad
  \widehat{W_2}=\left(\begin{matrix}
    \lambda_1g_1 &f_1\\
    \lambda_2g_2 &f_2
  \end{matrix}\right),\quad
  M_2=\left(\begin{matrix}
    \lambda_1^2f_1 &f_1\\
    \lambda_2^2f_2 &f_2
  \end{matrix}\right),\quad
  \widehat{M_2}=\left(\begin{matrix}
    \lambda_1^2g_1 &g_1\\
    \lambda_2^2g_2 &g_2
  \end{matrix}\right).
\end{equation}
\begin{lemma}
  With the parameters defined above,
  \begin{equation}\label{P0}
    P_{0}=
   \begin{cases}
   \left(\begin{matrix}
    \lambda_1\lambda_2\lambda_3\ldots\lambda_n\sqrt{\frac{|W_n|}{|\widehat{W_n}|}}& \\
    &\lambda_1\lambda_2\lambda_3\ldots\lambda_n\sqrt{\frac{|\widehat{W_n|}}{|W_n|}}
   \end{matrix}\right)&\mbox{if $n$ is even},\\\\
   \left(\begin{matrix}
     \sqrt{H} &\\
     &\frac{1}{\sqrt{H}}
   \end{matrix}\right)\left(\begin{matrix}
     &-\lambda_1\lambda_2\lambda_3\ldots\lambda_n\sqrt{\frac{|W_n|}{|\widehat{W_n}|}} \\
    -\lambda_1\lambda_2\lambda_3\ldots\lambda_n\sqrt{\frac{|\widehat{W_n|}}{|W_n|}}&
   \end{matrix}\right) &\mbox{if $n$ is odd}.
   \end{cases}
  \end{equation}
\end{lemma}
\begin{proof}
  Basic step: For $n=1,\,2$, according to theorems 1 and 2, we have
  \begin{equation}\nonumber
    P_0=\left(\begin{matrix}
     \sqrt{H} &\\
     &\frac{1}{\sqrt{H}}
   \end{matrix}\right)\left(\begin{matrix}
     &-\lambda_1\sqrt{\frac{f_1}{g_1}} \\
    -\lambda_1\sqrt{\frac{g_1}{f_1}}&
   \end{matrix}\right)=\left(\begin{matrix}
     \sqrt{H} &\\
     &\frac{1}{\sqrt{H}}
   \end{matrix}\right)\left(\begin{matrix}
     &-\lambda_1\sqrt{\frac{|W_1|}{|\widehat{W_1}|}} \\
    -\lambda_1\sqrt{\frac{|\widehat{W_1|}}{|W_1|}}&
   \end{matrix}\right),\quad \mbox{n=1}.
  \end{equation}
and
  \begin{equation}\nonumber
  \begin{aligned}
    P_0&=\left(\begin{matrix}
    \lambda_1\lambda_2\sqrt{\frac{h_2^{[1]}}{H}}& \\ &\lambda_1\lambda_2\sqrt{\frac{H}{h_2^{[1]}}}
  \end{matrix}\right)\\&=\left(\begin{matrix}\lambda_1\lambda_2\sqrt{\frac{\begin{vmatrix}\lambda_1f_1 &g_1\\    \lambda_2f_2 &g_2  \end{vmatrix}}{\begin{vmatrix}\lambda_1g_1 &f_1\\ \lambda_2g_2 &f_2\end{vmatrix}}}&\\
  &\lambda_1\lambda_2\sqrt{\frac{\begin{vmatrix}\lambda_1g_1 &f_1\\ \lambda_2g_2 &f_2\end{vmatrix}}{\begin{vmatrix}\lambda_1f_1 &g_1\\    \lambda_2f_2 &g_2  \end{vmatrix}}}\end{matrix}\right)=\left(\begin{matrix}
     \lambda_1\lambda_2\sqrt{\frac{|W_2|}{\widehat{|W_2|}}} &\\
    &\lambda_1\lambda_2\sqrt{\frac{|\widehat{W_2|}}{|W_2|}}
   \end{matrix}\right),\quad \mbox{n=2}.
  \end{aligned}
  \end{equation}
So eq. \eqref{P0} is true when $n=1,\,2$.

Inductive Step: Now we assume \eqref{P0} is true when $n=k\,(k\in\mathbb{Z^+})$, that is,
\begin{equation}\label{induction}
    P_{0}=
   \begin{cases}
   \left(\begin{matrix}
    \lambda_1\lambda_2\lambda_3\ldots\lambda_k\sqrt{\frac{|W_k|}{\widehat{|W_k|}}}& \\
    &\lambda_1\lambda_2\lambda_3\ldots\lambda_k\sqrt{\frac{|\widehat{W_k|}}{|W_k|}}
   \end{matrix}\right)&\mbox{if $k$ is even},\\\\
   \left(\begin{matrix}
     \sqrt{H} &\\
     &\frac{1}{\sqrt{H}}
   \end{matrix}\right)\left(\begin{matrix}
     &-\lambda_1\lambda_2\lambda_3\ldots\lambda_k\sqrt{\frac{|W_k|}{\widehat{|W_k|}}} \\
    -\lambda_1\lambda_2\lambda_3\ldots\lambda_k\sqrt{\frac{|\widehat{W_k|}}{|W_k|}}&
   \end{matrix}\right) &\mbox{if $k$ is odd}.
   \end{cases}
  \end{equation}
  From this assumption we want to deduce
  \begin{equation} \label{Pok+1}
  P_{0}\hspace{-0.2cm}=\hspace{-0.2cm}
   \begin{cases}\hspace{-0.2cm}
   \left(\begin{matrix}
    \lambda_1\lambda_2\lambda_3\ldots\lambda_{k+1}\sqrt{\frac{|W_{k+1}|}{\widehat{|W_{k+1}|}}}& \\
    &\lambda_1\lambda_2\lambda_3\ldots\lambda_{k+1}\sqrt{\frac{|\widehat{W_{k+1}|}}{|W_{k+1}|}}
   \end{matrix}\right),&\mbox{\hspace{-0.3cm}if $k+1$ is even},\\\\
   \hspace{-0.2cm}\left(\begin{matrix}
     \sqrt{H} &\\
     &\frac{1}{\sqrt{H}}
   \end{matrix}\right)\left(\begin{matrix}
     &-\lambda_1\lambda_2\lambda_3\ldots\lambda_k\sqrt{\frac{|W_{k+1}|}{\widehat{|W_{k+1}|}}} \\
    -\lambda_1\lambda_2\lambda_3\ldots\lambda_k\sqrt{\frac{|\widehat{W_{k+1}|}}{|W_{k+1}|}}&
   \end{matrix}\right),&\mbox{\hspace{-0.3cm}if $k+1$ is odd}.
   \end{cases}
  \end{equation}

  Note that the kernel of $T_k$ consists of  $\Phi_j(j=1,2,\cdots,k)$, i.e.,
$T_k(\lambda;\lambda_1,\lambda_2,...,\lambda_k)\Phi_j|_{\lambda=\lambda_j}
 =0$. Substituting the above induction hypothesis \eqref{induction} into these algebraic equations,
we obtain the $k$-fold   DT by the Cramer's rule as
  \begin{equation*}
  T_k=T_k(\lambda;\lambda_1,\lambda_2,...,\lambda_k)=
   \begin{cases}
   \frac{1}{\sqrt{|W_k||\widehat{W_k}|}}
      \left(\begin{matrix}
      \widehat{(T_k)}_{11} &\widehat{(T_k)}_{12}\\
      \widehat{(T_k)}_{21} &\widehat{(T_k)}_{22}
      \end{matrix}\right) &\mbox{if $k$ is even},\\\\
   \frac{1}{\sqrt{|W_k||\widehat{W_k}|}}
      \left(\begin{matrix}
     \sqrt{H} &\\
     &\frac{1}{\sqrt{H}}
   \end{matrix}\right)\left(\begin{matrix}
      \widehat{(T_k)}_{11} &\widehat{(T_k)}_{12}\\
      \widehat{(T_k)}_{21} &\widehat{(T_k)}_{22}
      \end{matrix}\right) &\mbox{if $k$ is odd},
   \end{cases}
  \end{equation*}
where
    \begin{equation*}
    \begin{aligned}
    \widehat{(T_k)}_{11}=\begin{vmatrix}
      \lambda^k &\widehat{\xi_k}\\
      \widehat{\eta_k} &\widehat{W_k}
    \end{vmatrix},\quad
    \widehat{(T_k)}_{12}=\begin{vmatrix}
      0 &\xi_k\\
      \widehat{\eta_k} &\widehat{W_k}
    \end{vmatrix},\quad
    \widehat{(T_k)}_{21}=\begin{vmatrix}
      0 &\xi_k\\
     {\eta_n} &{W_k}
    \end{vmatrix},\quad
    \widehat{(T_k)}_{22}=\begin{vmatrix}
      \lambda^k &\widehat{\xi_k}\\
      {\eta_k} &{W_k}
    \end{vmatrix},
    \end{aligned}
  \end{equation*}
  \begin{equation*}
    \widehat{\eta_k}=\left(\begin{matrix}
      \lambda_1^kf_1 &\lambda_2^kf_2  &\lambda_3^kf_3 &\ldots &\lambda_k^kf_k
    \end{matrix}\right)^T,\quad
    {\eta_k}=\left(\begin{matrix}
      \lambda_1^kg_1  &\lambda_2^kg_2 &\lambda_3^kg_3 &\ldots &\lambda_k^kg_k
    \end{matrix}\right)^T,
  \end{equation*}
  and
  \begin{itemize}
    \item If $k$ is even,
    \begin{equation}
      \nonumber
      \widehat{\xi_k}=\left(\begin{matrix}
        0 &\lambda^{k-2} &0 &\lambda^{k-4} &\ldots &0 &1
      \end{matrix}\right),\quad
      {\xi_k}=\left(\begin{matrix}
        \lambda^{k-1} &0 &\lambda^{k-3} &0 &\ldots &\lambda &0
      \end{matrix}\right),
    \end{equation}
    \item If $k$ is odd,
    \begin{equation}
      \nonumber
      \widehat{\xi_k}=\left(\begin{matrix}
        0 &\lambda^{k-2} &0 &\lambda^{k-4} &\ldots &\lambda &0
      \end{matrix}\right),\quad
      {\xi_k}=\left(\begin{matrix}
        \lambda^{k-1} &0 &\lambda^{k-3} &0 &\ldots &0 &1
      \end{matrix}\right).
    \end{equation}
  \end{itemize}
After the action of the $k$-fold   DT,
the $k$-fold   eigenfunction $\Phi_{k+1}^{[k]}=\left(\begin{matrix}
      f_{k+1}^{[k]}\\
      g_{k+1}^{[k]}
    \end{matrix}\right)$ corresponding to $\lambda_{k+1}$ is
    \begin{equation}\nonumber
    \Phi_{k+1}^{[k]}=\left(\begin{matrix}
      f_{k+1}^{[k]}\\
      g_{k+1}^{[k]}
    \end{matrix}\right)=
   \begin{cases}
   \frac{1}{\sqrt{|W_k||\widehat{W_k}|}}
      \left(\begin{matrix}
     |W_{k+1}|\\
     |\widehat{W_{k+1}}|
   \end{matrix}\right) &\mbox{if $k$ is even},\\\\
   \frac{1}{\sqrt{|W_k||\widehat{W_k}|}}
      \left(\begin{matrix}
     \sqrt{H} &\\
     &\frac{1}{\sqrt{H}}
   \end{matrix}\right)\left(\begin{matrix}
     |W_{k+1}|\\
     |\widehat{W_{k+1}}|
   \end{matrix}\right) &\mbox{if $k$ is odd}.
   \end{cases}
  \end{equation}
  By a direct calculation using the above expression of $\Phi_{k+1}^{[k]}$,
  \begin{equation}
    \label{quotient}
    h_{k+1}^{[k]}=\frac{f_{k+1}^{[k]}}{g_{k+1}^{[k]}}=\begin{cases}
   \dfrac{|W_{k+1}|}{\widehat{|W_{k+1}}|} &\mbox{if $k$ is even},\\\\
   \dfrac{|W_{k+1}|}{\widehat{|W_{k+1}}|}H &\mbox{if $k$ is odd}.
   \end{cases}
  \end{equation}
 By substituting the first equation of \eqref{quotient} into the second equation of
 \eqref{P0withhn}, and substituting the second  equation of \eqref{quotient} into the first equation of
 \eqref{P0withhn}, we obtain \eqref{Pok+1}. This means the expression
 of $P_0$ in \eqref{P0} holds for $n=k+1(k\geq 2)$.
Therefore the lemma is true for arbitrary $n(n\geq 1)$ according to the
principle of mathematical induction. \end{proof}
\noindent This lemma shows all unfavorable integrals $H^{[k]}
(k=1,2,\cdots,n-1) $ in $T_n$ are eliminated
except $H$.   \\

   We are now in a position to get the determinant of $T_n$ by solving $T_n\Phi_j|_{\lambda=
\lambda_j}=0$ ($j=1,2,\cdots,n$) with $P_0$ in \eqref{P0} and  $T_n$ in
\eqref{Tngeneral}.

\begin{theorem}\label{thm_nDT}
  With the Vandermonde-type matrices defined above, the $n$-fold
   DT of the coupled DNLSII equations  is  expressed by
  \begin{equation}
  \begin{aligned}
   T_n=T_n(\lambda;\lambda_1,\lambda_2,...,\lambda_n)=
   \begin{cases}
   \frac{1}{\sqrt{|W_n||\widehat{W_n}|}}
      \left(\begin{matrix}
      \widehat{(T_n)}_{11} &\widehat{(T_n)}_{12}\\
      \widehat{(T_n)}_{21} &\widehat{(T_n)}_{22}
      \end{matrix}\right) &\mbox{if $n$ is even},\\\\
   \frac{1}{\sqrt{|W_n||\widehat{W_n}|}}
     \left(\begin{matrix}
     \sqrt{H} &\\
     &\frac{1}{\sqrt{H}}
   \end{matrix}\right) \left(\begin{matrix}
      \widehat{(T_n)}_{11} &\widehat{(T_n)}_{12}\\
      \widehat{(T_n)}_{21} &\widehat{(T_n)}_{22}
      \end{matrix}\right) &\mbox{if $n$ is odd},
   \end{cases}
  \end{aligned}
  \end{equation}
  where
    \begin{equation*}
    \begin{aligned}
    \widehat{(T_n)}_{11}=\begin{vmatrix}
      \lambda^n &\widehat{\xi_n}\\
      \widehat{\eta_n} &\widehat{W_n}
    \end{vmatrix},\quad
    \widehat{(T_n)}_{12}=\begin{vmatrix}
      0 &\xi_n\\
      \widehat{\eta_n} &\widehat{W_n}
    \end{vmatrix},\quad
    \widehat{(T_n)}_{21}=\begin{vmatrix}
      0 &\xi_n\\
     {\eta_n} &{W_n}
    \end{vmatrix},\quad
    \widehat{(T_n)}_{22}=\begin{vmatrix}
      \lambda^n &\widehat{\xi_n}\\
      {\eta_n} &{W_n}
    \end{vmatrix}.
    \end{aligned}
  \end{equation*}
  \begin{equation*}
    \widehat{\eta_n}=\left(\begin{matrix}
      \lambda_1^nf_1 &\lambda_2^nf_2  &\lambda_3^nf_3 &\ldots &\lambda_n^nf_n
    \end{matrix}\right)^T,\quad
    {\eta_n}=\left(\begin{matrix}
      \lambda_1^ng_1  &\lambda_2^ng_2 &\lambda_3^ng_3 &\ldots &\lambda_n^ng_n
    \end{matrix}\right)^T.
  \end{equation*}
  \begin{itemize}
    \item if $n=2k\ (k>0)$
    \begin{equation}
      \nonumber
      \widehat{\xi_n}=\left(\begin{matrix}
        0 &\lambda^{n-2} &0 &\lambda^{n-4} &\ldots &0 &1
      \end{matrix}\right),\quad
      {\xi_n}=\left(\begin{matrix}
        \lambda^{n-1} &0 &\lambda^{n-3} &0 &\ldots &\lambda &0
      \end{matrix}\right).
    \end{equation}
    \item if $n=2k+1\ (k\geq0)$
    \begin{equation}
      \nonumber
      \widehat{\xi_n}=\left(\begin{matrix}
        0 &\lambda^{n-2} &0 &\lambda^{n-4} &\ldots &\lambda &0
      \end{matrix}\right),\quad
      {\xi_n}=\left(\begin{matrix}
        \lambda^{n-1} &0 &\lambda^{n-3} &0 &\ldots &0 &1
      \end{matrix}\right).
    \end{equation}
  \end{itemize}
\end{theorem}
\begin{corollary}\label{thm_qn}
  The $n$-th order solutions $(q^{[n]},\ r^{[n]})$, generated by $T_n$ from the seed solution
$(q,r)$, of the coupled DNLSII equations are
  \begin{equation}\label{qn}
    q^{[n]}=\begin{cases}
    \frac{|\widehat{W_n}|}{|W_n|}q-2{\rm i}\frac{|M_n|}{|W_n|} &\mbox{if $n$ is even},\\\\
    H\left(\frac{\widehat{|W_n|}}{|W_n|}q-2{\rm i}\frac{|M_n|}{|W_n|}\right) &\mbox{if $n$ is odd}.
    \end{cases}\qquad
  r^{[n]}=\begin{cases}
    \frac{|W_n|}{|\widehat{W_n}|}r-2{\rm i}\frac{|\widehat{M_n}|}{|\widehat{W_n}|} &\mbox{if $n$ is even},\\\\
    \frac 1H\left(\frac{|W_n|}{|\widehat{W_n}|}r-2{\rm i}\frac{|\widehat{M_n}|}{|\widehat{W_n}|}\right) &\mbox{if $n$ is odd}.
  \end{cases}
  \end{equation}
\end{corollary}
\begin{proof}
Substituting the expression of $T_n$ in theorem \ref{thm_nDT} into
    \begin{equation*}
    T_{nx}+T_nU=U^{[n]}T_n,
    \end{equation*}
    and comparing the coefficient of $\lambda^{n+1}$ on both sides, we obtain the expression of $q^{[n]}$
    and $r^{[n]}$ in \eqref{qn}. Here $U^{[n]}=U(q\rightarrow q^{[n]},r\rightarrow r^{[n]})$.
\end{proof}

Furthermore, we get the transformed eigenfunctions
$\Phi_{j}^{[n]}=T_n\Phi_j|_{\lambda=\lambda_j}$ $(j\geq n+1)$ using determinant representation
of $T_n$ in theorem \ref{thm_nDT}.
Note that $\Phi_{j}^{[n]}=0$ $(j=1,2,\cdots,n)$.
\begin{corollary}\label{thm eigenfunofTn}
If $j\geq n+1$. The $n$-th order
 eigenfunction  $\Phi_{j}^{[n]}=\left(\begin{matrix}
      f_{j}^{[n]}\\
      g_{j}^{[n]}
    \end{matrix}\right)$ associated with
    ($\lambda_{j}, q^{[n]}, r^{[n]}$) is
    \begin{equation}
    \Phi_{j}^{[n]}=\left(\begin{matrix}
      f_{j}^{[n]}\\
      g_{j}^{[n]}
    \end{matrix}\right)=
   \begin{cases}
   \frac{1}{\sqrt{|W_n||\widehat{W_n}|}}
      \left(\begin{matrix}
     |W_{n+1}|_{n+1\rightarrow j}\\
     |\widehat{W_{n+1}}|_{n+1\rightarrow j}
   \end{matrix}\right) &\mbox{if $n$ is even},\\\\
   \frac{1}{\sqrt{|W_n||\widehat{W_n}|}}
      \left(\begin{matrix}
     \sqrt{H} &\\
     &\frac{1}{\sqrt{H}}
   \end{matrix}\right)\left(\begin{matrix}
     |W_{n+1}|_{n+1\rightarrow j}\\
     |\widehat{W_{n+1}}|_{n+1\rightarrow j}
   \end{matrix}\right) &\mbox{if $n$ is odd}.
   \end{cases}
  \end{equation}
 where the subscript $n+1\rightarrow j$ means $\lambda_{n+1},\,f_{n+1}\,
 \mbox{ and }\,g_{n+1}$ are replaced by $\lambda_{j},\,f_{j}\,
 \mbox{ and }\,g_{j}$ respectively.
\end{corollary}

  Obviously, we have removed the unfavorable integrals $H^{[k]}$
except the  simplest one $H$ in above three theorems for the coupled DNLSII equations.
For the widely used two seeds, the vacuum  seed and the periodic seed, $H$ is
presented in remark \ref{remark1} under a reduction condition $q=r^*$,
so explicit forms of $q^{[n]}$ and $r^{[n]}$ can be
calculated by determinant representation of $T_n$. Furthermore, in next subsection, we shall
consider how to  realize the reduction $q^{[n]}= (r^{[n]})^*$ in order to get the
 solution of the DNLSII equation \eqref{cll}.
\\
\\
{\noindent\bf{2.4 The reduction of Darboux transformation}}\\

 In this subsection, we will discuss the DTs for the coupled DNLSII equations \eqref{clleq}
under the reduction condition $r=q^*$, which results in the DTs for
the DNLSII equation \eqref{cll}. To this end, the following properties of the
eigenfunctions are necessary.
\begin{lemma} \label{relation}
  Under the reduction condition $r=q^*$, the eigenfunction  $\Phi_k=\left(\begin{matrix}
  f_k \\
  g_k
\end{matrix}\right)$ associated with eigenvalue $\lambda_k$ possesses the following properties:
\begin{enumerate}
  \item $f_k^*=g_k$, $\lambda_k=-\lambda_k^*$;
  \item $f_k^*=g_l$, $g_k^*=f_l$, $\lambda_k^*=-\lambda_l$, $(k\neq l)$.
\end{enumerate}
\end{lemma}
\begin{theorem}\label{thm reductionqn}
Let
\begin{equation}\label{choicesoflambdas}
  \begin{cases}
     \lambda_k=-\lambda_k^*\,(k\in\mathbb{N^*})  &\mbox{$n$ is odd},\\
    \lambda_{2k}=-\lambda_{2k-1}^*\,(k\in\mathbb{N^*}) &\mbox{$n$ is even}.
  \end{cases}
\end{equation}
 Then the $n$-th order solution $(q^{[n]},\ r^{[n]})$ satisfies
 the reduction condition, i.e., $r^{[n]}={q^{[n]}}^*$,
 which implies $q^{[n]}$ is a solution of the DNLSII equation. Moreover,
 $T_n$ is the $n$-fold  DT for the same equation.
\end{theorem}

\begin{proof}
  For $n=1$, let $\lambda_1=-\lambda_1^*$, then $\Phi_1=\left(\begin{matrix}f_1\\ f_1^*\end{matrix}\right)$. According to \eqref{newpotentials} we have
  \begin{equation*}
  \begin{aligned}
      {q^{[1]}}^*&=\frac{H^*}{h_1^*}(q^*+2{\rm i}\lambda_1^*h_1^*)=\frac{f_1}{f_1^*}(r-2{\rm i}\lambda_1\frac{f_1^*}{f_1})H^*,\\
      r^{[1]}&=\frac{h_1}{H}(r-2{\rm i}{\lambda_1}{h_1})=\frac{f_1}{Hf_1^*}(r-2{\rm i}\lambda_1\frac{f_1^*}{f_1}).
  \end{aligned}
  \end{equation*}
  Besides, we get $H^*=\frac{1}{H}$ from \eqref{HH}. So $r^{[1]}={q^{[1]}}^*$.

   For $n=2$, let $\lambda_2=-\lambda_1^*$, then $\Phi_2=\left(\begin{matrix} g_1^*\\f_1^*\end{matrix}\right)$. According to  \eqref{q2_1} we have
  \begin{equation*}
    \begin{aligned}
      q^{[2]}=\frac{\begin{vmatrix} \lambda_1g_1 &f_1\\  \lambda_2g_2  &f_2\end{vmatrix}}{\begin{vmatrix} \lambda_1f_1  &g_1\\ \lambda_2f_2  &g_2\end{vmatrix}}q-2{\rm i}\frac{\begin{vmatrix}\lambda_1^2f_1  &f_1\\  \lambda_2^2f_2  &f_2\end{vmatrix}}{\begin{vmatrix} \lambda_1f_1 &g_1\\  \lambda_2f_2 &g_2\end{vmatrix}}=\frac{\begin{vmatrix} \lambda_1g_1 &f_1\\  -\lambda_1^*f_1^*  &g_1^*\end{vmatrix}}{\begin{vmatrix} \lambda_1f_1  &g_1\\ -\lambda_1^*g_1^*  &f_1^*\end{vmatrix}}q-2{\rm i}\frac{\begin{vmatrix}\lambda_1^2f_1  &f_1\\  {\lambda_1^*}^2g_1^*  &g_1^*\end{vmatrix}}{\begin{vmatrix} \lambda_1f_1 &g_1\\  -\lambda_1^*g_1^* &f_1^*\end{vmatrix}},\\
      r^{[2]}=\frac{\begin{vmatrix} \lambda_1f_1 &g_1\\  \lambda_2f_2  &g_2\end{vmatrix}}{\begin{vmatrix} \lambda_1g_1  &f_1\\ \lambda_2g_2  &f_2\end{vmatrix}}r-2{\rm i}\frac{\begin{vmatrix}\lambda_1^2g_1  &g_1\\  \lambda_2^2g_2  &g_2\end{vmatrix}}{\begin{vmatrix} \lambda_1g_1 &f_1\\  \lambda_2g_2 &f_2\end{vmatrix}}=\frac{\begin{vmatrix} \lambda_1f_1  &g_1\\ -\lambda_1^*g_1^*  &f_1^*\end{vmatrix}}{\begin{vmatrix} \lambda_1g_1 &f_1\\  -\lambda_1^*f_1^*  &g_1^*\end{vmatrix}}r-2{\rm i}\frac{\begin{vmatrix}\lambda_1^2g_1  &g_1\\  {\lambda_1^*}^2f_1^*  &f_1^*\end{vmatrix}}{\begin{vmatrix} \lambda_1g_1 &f_1\\  -\lambda_1^*f_1^*  &g_1^*\end{vmatrix}},
    \end{aligned}
  \end{equation*}
  and
  \begin{equation*}
  {q^{[2]}}^*=\frac{\begin{vmatrix}\lambda_1^*g_1^* &f_1^*\\ -\lambda_1f_1 &g_1\end{vmatrix}}{\begin{vmatrix}\lambda_1^*f_1^* &g_1^*\\ -\lambda_1g_1 &f_1\end{vmatrix}}r+2{\rm i}\frac{\begin{vmatrix}{\lambda_1^*}^2f_1^* &f_1^*\\ \lambda_1^2g_1 &g_1\end{vmatrix}}{\begin{vmatrix}\lambda_1^*f_1^* &g_1^*\\ -\lambda_1g_1 &f_1\end{vmatrix}}=\frac{\begin{vmatrix} \lambda_1f_1  &g_1\\ -\lambda_1^*g_1^*  &f_1^*\end{vmatrix}}{\begin{vmatrix} \lambda_1g_1 &f_1\\  -\lambda_1^*f_1^*  &g_1^*\end{vmatrix}}r-2{\rm i}\frac{\begin{vmatrix}\lambda_1^2g_1  &g_1\\  {\lambda_1^*}^2f_1^*  &f_1^*\end{vmatrix}}{\begin{vmatrix} \lambda_1g_1 &f_1\\  -\lambda_1^*f_1^*  &g_1^*\end{vmatrix}}.
  \end{equation*}
  That is, ${q^{[2]}}^*=r^{[2]}$.

In the case $n>2$, ${q^{[n]}}^*=r^{[n]}$ can be verified by iteration of
$T_1$ for $n=2k+1$ and  $T_2$ for $n=2k$. Note that
 each step of iteration preserves the reduction condition.
\end{proof}

\section{Solutions of the DNLSII equation}


In this section, we shall construct several explicit solutions of the
DNLSII equation \eqref{cll} by applying the determinant representation
of $T_n$ for the DNLSII equation. We always choose eigenvalues as eq.(\ref{choicesoflambdas}).  \\

\noindent{\bf 3.1 Solutions from a vacuum seed by Darboux transformation}\\

  Set seed solution $q=r=0$,  then let $H=1$ without loss of generality
because $H$ is a constant according to remark \ref{remark1}. By solving the Lax pair
\eqref{laxpair}, we have an eigenfunction
  \begin{equation}
    \Phi_k=\left(\begin{matrix}
      f_k\\g_k
    \end{matrix}\right),
    \quad f_k={{\rm exp}({-{\rm i}{\lambda_k}^{2}x-2\,
    {\rm i}{\lambda_k}^{4}t})},\quad g_k={{\rm exp}({{\rm i}
    {\lambda_k}^{2}x+2\,{\rm i}{\lambda_k}^{4}t})}
  \end{equation}
for the  eigenvalue $\lambda_k$.\\


{\bf Case A:} Let $n=1$ and $\lambda_1={\rm i}\beta_1$, then we get
    \begin{equation}
     q^{[1]}=(\frac{f_1^*}{f_1}q-2{\rm i}\lambda_1)
     =-2{\rm i}\lambda_1=2\beta_1
    \end{equation}
according to eq.(\ref{newpotentials}). Unfortunately, this is a
trivial solution of the DNLSII equation.
\\


{\bf Case B:} Set $n=2$,  $\lambda_1=\alpha_1+{\rm i}\beta$ and
 $\lambda_2=-\lambda_1^*$.  Eq. \eqref{q2_1} provides  an usual
 single-soliton solution of the DNLSII equation:
    \begin{equation}\label{1soliton}
      q^{[2]}_1=\frac{\begin{vmatrix} \lambda_1g_1 &f_1\\  -\lambda_1^*f_1^*  &g_1^*\end{vmatrix}}{\begin{vmatrix} \lambda_1f_1  &g_1\\ -\lambda_1^*g_1^*  &f_1^*\end{vmatrix}}q-2{\rm i}\frac{\begin{vmatrix}\lambda_1^2f_1  &f_1\\  {-\lambda_1^*}^2g_1^*  &g_1^*\end{vmatrix}}{\begin{vmatrix} \lambda_1f_1 &g_1\\  -\lambda_1^*g_1^* &f_1^*\end{vmatrix}}=\dfrac{-4\alpha_1\beta_1{\rm exp}(-{\rm i}F_2)}{\alpha_1\cosh(F_1)+{\rm i}\beta_1\sinh(F_1)},
    \end{equation}
    with
    \begin{equation*}
      \begin{aligned}
      {F_1}&=4\,{\alpha_1}\,{\beta_1}\, \left( 4\,t{{\alpha_1}}^{2}-
               4\,t{{\beta_1}}^{2}+x \right),\\
      {F_2}&=4\,{{\alpha_1}}^{4}t-24\,{{\alpha_1}}^{2}t{{\beta_1}}^{
         2}+2\,{{\alpha_1}}^{2}x+4\,{{\beta_1}}^{4}t-2\,{{\beta_1}}^{2}x.
      \end{aligned}
    \end{equation*}
Set $\alpha_1\rightarrow 0$ in \eqref{1soliton}, then $q^{[2]}_1$ becomes a rational
single-soliton solution
    \begin{equation}\label{1rational}
q^{[2]}_2={\frac {{4\beta_1}\,{{\rm exp}({-2\,{\rm i}{{\beta_1}}^{2} ( 2\,t{
      {\beta_1}}^{2}-x ) })}}{1+4\,{\rm i}{{\beta_1}}^{2}x-16\,{\rm i}{{\beta_1}}^{4}t}}.
    \end{equation}
The trajectory of \eqref{1soliton} and \eqref{1rational}
are two lines  defined by
$x=-4(\alpha_1^2-\beta_1^2)t$ and $x=4\beta_1^2t$ respectively.
The profiles of them are depicted in Fig. \ref{fig_soliton}.
From Fig. \ref{fig_soliton} (c), we find the rational
single-soliton solution is narrower than the usual single-soliton
solution with the same parameters.

Comparing with the solutions of the DNLSI equation generated by DT
\cite{JPA44305203},  these two solutions and their modulses are able
to be obtained from the corresponding solutions of the DNLSI
equation by inverse gauge transformation \eqref{gt}, because
integral in \eqref{gt} containing square of modulus of solutions are
indeed calculable explicitly. This is consistent with the known
assertion mentioned in the introduction. So we omit the discussion
about higher order soliton solutions of the DNLSII equation.
\\

\noindent{\bf 3.2 Breathers from a periodic seed by Darboux transformation}\\

  Here, we shall discuss the solutions of the DNLSII equation generated from a
periodic solution. In general, set a periodic seed
  \begin{equation}
    q=c\,{\rm e}^{({\rm i}\rho)},\quad \rho=ax+bt,\quad b=-a^2-c^2a,\quad a,\ c\in\mathbb{R}.
  \end{equation}
According to remark \ref{remark1},
$H={{\rm exp}({-\frac{1}{4}\,{\rm i}{c}^{2}\left( 4\,ta+{c}^{2}t-2\,x\right) })}$
in this case. Solving the Lax pair \eqref{laxpair} associated with this seed solution,
and separating variables,  we get two different solutions
corresponding to an eigenfunction $\lambda$ as
  \begin{equation}\label{originaleigenfunforp}
  \begin{aligned}
    \Psi_1&=\left(\begin{matrix}
      \psi_{11}(x,t,\lambda)\\
      \psi_{12}(x,t,\lambda)
    \end{matrix}\right)=\left(\begin{matrix}
      {\frac {2\,{\rm i}a+4\,{\rm i}{\lambda}^{2}-{\rm i}{c}^{2}+{\rm i}s}{4\lambda}}{\rm exp}\left({\rm i}\left(\frac{\rho}{2}-\frac{s}{4}x+\frac{\left(2a+c^2-4\lambda^2\right)s}{8}t\right)\right)\\
      c\,{\rm exp}\left({\rm i}\left(-\frac{\rho}{2}-\frac{s}{4}x+\frac{\left(2a+c^2-4\lambda^2\right)s}{8}t\right)\right)
    \end{matrix}\right),\\
    \Psi_2&=\left(\begin{matrix}
      \psi_{21}(x,t,\lambda)\\
      \psi_{22}(x,t,\lambda)
    \end{matrix}\right)=\left(\begin{matrix}
      {\frac {2\,{\rm i}a+4\,{\rm i}{\lambda}^{2}-{\rm i}{c}^{2}-{\rm i}s}{4\lambda}}{\rm exp}\left({\rm i}\left(\frac{\rho}{2}+\frac{s}{4}x-\frac{\left(2a+c^2-4\lambda^2\right)s}{8}t\right)\right)\\
      c\,{\rm exp}\left({\rm i}\left(-\frac{\rho}{2}+\frac{s}{4}x-\frac{\left(2a+c^2-4\lambda^2\right)s}{8}t\right)\right)
    \end{matrix}\right)
  \end{aligned}
  \end{equation}
  with
  \begin{equation}\label{phasefactor}
    s=\sqrt{4\,{a}^{2}+16\,a{\lambda}^{2}-4\,{c}^{2}a+16\,{\lambda}^{4}+8\,{\lambda}^{2}{c}^{2}+{c}^{4}}.
  \end{equation}
  In condition of $r=q^*$, it is obvious that $\left(\begin{matrix}
     \psi^*_{12}(x,t,-\lambda^*)\\ \psi^*_{11}(x,t,-\lambda^*)
   \end{matrix}\right)$ and $\left(\begin{matrix}
     \psi^*_{22}(x,t,-\lambda^*)\\ \psi^*_{21}(x,t,-\lambda^*)
   \end{matrix}\right)$ are also solutions of the Lax pair \eqref{laxpair} associated with
    $\lambda$ corresponding to lemma \ref{relation}. Using  the principle of superposition, we select the first
solution in \eqref{originaleigenfunforp} to get  a new solution of
the Lax pair \eqref{laxpair} with $\lambda=\lambda_j$
  \begin{equation}\label{eigenfunnozero}
    \Phi_j=\left(\begin{matrix}
      f_j(x,t,\lambda_j)\\
      g_j(x,t,\lambda_j)
    \end{matrix}\right)=\left(\begin{matrix}
     \psi_{11}(x,t,\lambda_j)+\psi^*_{12}(x,t,-\lambda_j^*)\\
     \psi_{12}(x,t,\lambda_j)+\psi^*_{11}(x,t,-\lambda_j^*)\\
    \end{matrix}\right).
  \end{equation}
  Next, we shall  use this eigenfunction to construct new solutions of the DNLSII equation
by means of theorems \ref{thm_DT1} and \ref{thm_DT2} under the
choices of eigenvalues in eq.(\ref{choicesoflambdas}), or equivalently through
theorem \ref{thm reductionqn}.\\
 {\bf Case A:} $n=1$.  Let $\lambda_1={\rm i}{\beta}$, the first order solution is obtained
  through theorem \ref{thm_DT1}  and eq.(\ref{choicesoflambdas}).
        \begin{equation}\label{solitonfromnonzero}
        q^{[1]}=\left(2\beta+\frac{2ac-4\beta^2c-c^3+cs_1+4\beta{c^2}\Gamma_1}
        {(2a-4\beta^2-c^2+s_1)\Gamma_1+4\beta{c}}\right)
        {{\rm exp}({-\frac{1}{4}\,{\rm i}{c}^{2}\left( t{c}^{2}-2\,x+4\,ta \right) })}
         \end{equation}
        with\\[-20pt]
         \begin{equation*}
         \begin{aligned}
           \Gamma_1&=\exp\left(2{\rm i}s_1\left(-\frac{x}{4}+\frac{(2a+c^2+4\beta^2)t}{8}\right)\right),\\
           s_1&=\sqrt{4\,{a}^{2}-16\,a{\beta}^{2}-4\,{c}^{2}a+16\,{\beta}^{4}-8\,{\beta}^{2}
                  {c}^{2}+{c}^{4}}.
                \end{aligned}
         \end{equation*}

     In the case $4\,{a}^{2}-16\,a{\beta}^{2}-4\,{c}^{2}a+16\,{\beta}^{4}-8\,{\beta}^{2}
    {c}^{2}+{c}^{4}<0$,
    $q^{[1]}$ gives a soliton solution, which reaches its amplitude $|q_A^{[1]}|=|2\beta+c|$ at
    $x=\frac{2a+c^2+4\beta^2}{2}t$. When $x\rightarrow\infty,\,t\rightarrow\infty$,
    $|q_B^{[1]}|=|2\beta+{\frac {2\,a-4\,{\beta}^{2}-{c}^{2}+s_1}{4\beta}}|$. If $|q_B^{[1]}|> |q_A^{[1]}|$,
    it generates a dark soliton. Otherwise,
    it gives a bright soliton with a non-vanishing boundary.
    Moreover, $q^{[1]}$ generates a periodic solution while $4\,{a}^{2}-16\,a{\beta}^{2}-4\,{c}^{2}a+16\,
    {\beta}^{4}-8\,{\beta}^{2}
    {c}^{2}+{c}^{4}>0$. The dark soliton and bright soliton are showed in
    Fig. \ref{fig.1_solitonfromnonzero}.   \\
   {\bf Case B:}  $n=2$. To preserve the reduction condition
    $r^{[2]}={q^{[2]}}^*$,  set $\lambda_1=\alpha_1+{\rm i}\beta_1\,\mbox{and } \lambda_2=-\lambda_1^*=-\alpha_1
    +{\rm i}\beta_1\,(\alpha_1\neq0)$, then $f_2=g_1^*,\ g_2=f_1^*$
    according to lemma \ref{relation} and theorem \ref{thm reductionqn}. In order to compare with the breather solution of the DNLSI equation in \cite{JPA44305203},
    let $a=-2\,{\alpha_{{1}}}^{2}+2\,{\beta_{{1}}}^{2}-\frac{1}{2}\,{c}^{2}$ such that the imaginary part of
    $ {4\,{a}^{2}+16\,a{\lambda_1}^{2}-4\,{c}^{2}a+16\,{\lambda_1}^{4}+8\,{\lambda_1}^{2}{c}^{2}+{c}^{4}}$
     equals $0$.   \eqref{q2_1} in theorem 2 gives a breather solution
     of the DNLSII equation as follows:
        \begin{equation}\label{breather}
          q_1^{[2]}=\frac{L_1L_3}{L_2}
        \end{equation}
        with

          \begin{align*}
         K=&2\,\sqrt {{c}^{4}+4\,{\alpha_{{1}}}^{2}{c}^{2}-4\,{\beta_{{1}}}^{2}{c}
            ^{2}-16\,{\alpha_{{1}}}^{2}{\beta_{{1}}}^{2}},\\
         L_1=&A_1\sin\Theta_1+A_2\cos\Theta_1+A_3\sinh\Theta_2
         +A_4\cosh\Theta_2,\\
         L_2=&B_1\sin\Theta_1+B_2\cos\Theta_1+B_3\sinh\Theta_2
         +B_4\cosh\Theta_2,\\
         L_3=&\,{\rm exp}({\rm i}(\tilde{a}x+\tilde{b}t)), \quad
         \Theta_1=\frac{1}{2}K\left(x+4\left(\alpha_1^2-\beta_1^2\right)t\right),\quad
         \Theta_2=2K\alpha_1\beta_1t,\\
         A_1=&\,{\rm i}  \alpha_{{1}}\beta_{{1}}{c}\left(8\,{\beta_{{1}}}^{2}-2\,{c}^{2}+K \right),
        \ \  A_2=\alpha_{{1}}\beta_{{1}}{c}\left( 8\,{\alpha_{{1}}}^{2}+2\,{c}^{2}
              -K \right),\\
         A_3=&\frac{1}{2}{\rm i}\beta_{{1}}\left(8\,{\beta_{{1}}}^{2}-2\,{c}^{2}+K \right)  \left( 8\,{\alpha_{{1}}}^{2}+{c}^{2} \right),
\ A_4=\frac{1}{2}\,\alpha_{{1}}\left( 8\,{\alpha_{{1}}}^{2}+2\,{c}^{2}-K \right)\left( {c}^{2}-8\,{\beta_{{1}}}^{2} \right),\\
B_1=&\,{\rm i}\alpha_{{1}}
             \beta_{{1}}\left(8\,{\beta_{{1}}}^{2}-2\,{c}^{2}+K \right),
        \ B_2=-\alpha_{{1}}\beta_{{1}} \left(8\,{\alpha_{{1}}}^{2}+2\,{c}^{2}-K\right),\\
B_3=&-\frac{1}{2}{\rm i}{c}\beta_{{1}}\left( 8\,{\beta_{{1}}}^{2}-2\,{c}^{2}+K \right),
\ B_4=\frac{1}{2}\,{c}\alpha_{{1}} \left(8\,{\alpha_{{1}}}^{2}+2\,{c}^{2}-K\right),\\
\tilde{a}=&( -2\,{\alpha_{{1}}}^{2}+2\,{\beta_{{1}}}^{2}-\frac{1}{2}\,{c}^{2}
        ),
\ \tilde{b}=( -4\,{\alpha_{{1}}}^{4}+8\,{\alpha_{{1}}}^{2}{\beta_
        {{1}}}^{2}-4\,{\beta_{{1}}}^{4}+\frac{1}{4}\,{c}^{4} ).\\
         \end{align*}


        If $K^2<0$, it reaches its amplitude $|c+4\beta_1|$ at $x=4(\,\beta_1^2-\,\alpha_1^2)t$, i.e., it propagates along a line in the direction $x=4(\,\beta_1^2-\,\alpha_1^2)t$ on ($x,t$)-plane. Especially, when $|\alpha_1|=|\beta_1|$, the temporal periodic breather solution (Kuznetsov-Ma breather \cite{Kuznetsov1977,SAM6053}) is obtained. If $K^2>0$, it reaches its amplitude $|c+4\beta_1|$ at $t=0$, i.e., the
        spacial periodic breather solution (Akhmediev breather \cite{Akhmediev_1985_894}).
        Three kinds of breather solutions are shown in
        Fig. \ref{fig.breather} to confirm their periodicity. Furthermore, in order to compare with the
        result \cite{JPA44305203} of the DNLSI equation, we present
        an expression of the square of the modulus of a breather solution given by
         \eqref{breather} as
        \begin{equation}\label{breathermodulus}
          {|q_1^{[2]}|}^{2}=\frac{\widetilde{L_1}}{\widetilde{L_2}},
        \end{equation}
        with

          \begin{align*}
            \qquad\widetilde{L_1}=&\widetilde{A_1}\sin\Theta_1\sinh\Theta_2+\widetilde{A_2}\cos\Theta_1\cosh
            \Theta_2+\widetilde{A_3}\cos(2\Theta_1)+\widetilde{A_4}\cosh(2\Theta_2)+\widetilde{A_5},\\
            \qquad\widetilde{L_2}=&\widetilde{B_1}\sin\Theta_1\sinh\Theta_2+\widetilde{B_2}\cos\Theta_1\cosh
            \Theta_2+\widetilde{B_3}\cos(2\Theta_1)+\widetilde{B_4}\cosh(2\Theta_2)+\widetilde{B_5},\\
 \widetilde{A_1}=&(\alpha_{{1}}{\beta_{{1}}}^{2}{c}^{3}+8\,{\alpha_{{1}}}^{3}{
                             \beta_{{1}}}^{2}c) {k_{{1}}}^{2},
            \widetilde{A_2}=( {\alpha_{{1}}}^{2}\beta_{{1}}{c}^{3}-8\,{\alpha_{{1}}}^{2}{
                             \beta_{{1}}}^{3}c) {k_{{2}}}^{2},
            \widetilde{A_3}=\frac{1}{2}\,{\alpha_{{1}}}^{2}{\beta_{{1}}}^{2}{c}^{2} \left( {k_{{2}}}^{2}-{k_{{1}}}^{2} \right),\\
            \widetilde{A_4}=&\frac{1}{8}\,{c}^{4} \left( {\beta_{{1}}}^{2}{k_{{1}}}^{2}+{\alpha_{{1}}}^{2}{
                            k_{{2}}}^{2} \right) +8\,{\alpha_{{1}}}^{2}{\beta_{{1}}}^{2} \left( {
                            \alpha_{{1}}}^{2}{k_{{1}}}^{2}+{\beta_{{1}}}^{2}{k_{{2}}}^{2} \right)
                           +2\,{\alpha_{{1}}}^{2}{\beta_{{1}}}^{2}{c}^{2} \left( {k_{{1}}}^{2}+{k
                           _{{2}}}^{2} \right),\\
            \widetilde{A_5}=&\frac{1}{8}\,{c}^{4} \left( {\alpha_{{1}}}^{2}{k_{{2}}}^{2}-{\beta_{{1}}}^{2}{
                            k_{{1}}}^{2} \right) +8\,{\alpha_{{1}}}^{2}{\beta_{{1}}}^{2} \left( {
                            \beta_{{1}}}^{2}{k_{{2}}}^{2}-{\alpha_{{1}}}^{2}{k_{{1}}}^{2} \right)
                            -\frac{3}{2}\,{\alpha_{{1}}}^{2}{\beta_{{1}}}^{2}{c}^{2} \left( {k_{{1}}}^{2}+{k_{{2}}}^{2} \right),\\
            \widetilde{B_1}=&-c\alpha_{{1}}{\beta_{{1}}}^{2}{k_{{1}}}^{2},
           \  \widetilde{B_2}=-c{\alpha_{{1}}}^{2}\beta_{{1}}{k_{{2}}}^{2},
            \ \widetilde{B_3}=\frac{1}{2}\,{\alpha_{{1}}}^{2}{\beta_{{1}}}^{2} \left( {k_{{2}}}^{2}-{k_{{1}}}^{2} \right),\\
            \widetilde{B_4}=&\frac{1}{8}\,{c}^{2} \left( {\beta_{{1}}}^{2}{k_{{1}}}^{2}+{\alpha_{{1}}}^{2}{
                            k_{{2}}}^{2} \right),
           \  \widetilde{B_5}=\frac{1}{8}\,{c}^{2} \left( {\alpha_{{1}}}^{2}{k_{{2}}}^{2}-{\beta_{{1}}}^{2}{
                            k_{{1}}}^{2} \right) +\frac{1}{2}\,{\alpha_{{1}}}^{2}{\beta_{{1}}}^{2} \left(
                            {k_{{1}}}^{2}+{k_{{2}}}^{2} \right),\\
            k_{1}=&8\,{\beta_{{1}}}^{2}-2\,{c}^{2}+K, \ k_{2}=8\,{\alpha_{{1}}}^{2}+2\,{c}^{2}-K.
          \end{align*}

        It is obvious that ${|q_1^{[2]}|}^{2}$ is different from the
        square of modulus of a breather (see eq.(53) of
        ref.\cite{JPA44305203}) for the  DNLSI equation.
        So $q_1^{[2]}$ can not be
        obtained from the above mentioned breather of the DNLSI equation by the
        gauge transformation \eqref{gt} because this transformation keeps
        the modulus of two solutions invariant.\\

Note that the period of the breather solution
\eqref{breather} is proportional to $\frac{1}{K}$.
Let $K$ go to zero, then the period of the breather goes to infinity
 and leave only one  peak located around the origin
 of the $(x,\,t)$-plane.  Since $c=2\beta_1$ is a common zero of
$K$ and  the eigenfunction $\Phi_1$ in \eqref{eigenfunnozero},
$q_1^{[2]}$ in \eqref{breather} becomes an indeterminate form
 $\frac{0}{0}$. By adopting L'Hospital's rule
 for this breather solution at $c=2\beta_1$, we obtain the first order
 rogue wave solution $q^{[2]}_2$ of the DNLSII equation:
        \begin{equation}\label{1rw}
          q^{[2]}_2=-\frac{L_1}{L_2}2\beta_1{\rm exp}\left({-2\,{\rm i}{\alpha_{{1}}}^{2} \left( x+2\,t{\alpha_{{1}}}^{2}-4\,t{\beta_{{1}}}^{2} \right) }\right),
        \end{equation}
        with

          \begin{align*}
            L_1=&16\,{\rm i}{\beta_{{1}}}^{4}{x}^{2}+16\,{\rm i}{\alpha_{{1}}}^{2}{\beta_{{1}}}^{2}
                 {x}^{2}-128\,{\rm i}{\beta_{{1}}}^{6}xt+128\,{\rm i}{\beta_{{1}}}^{2}t{\alpha_{{1}
                 }}^{4}x+8\,{\beta_{{1}}}^{2}x\\
                 &+96\,t{\alpha_{{1}}}^{2}{\beta_{{1}}}^{2}
                 +256\,{\rm i}{\beta_{{1}}}^{2}{t}^{2}{\alpha_{{1}}}^{6}-3\,{\rm i}+256\,{\rm i}{\beta_{{
                 1}}}^{8}{t}^{2},\\
            L_2=&16\,{\rm i}{\alpha_{{1}}}^{2}{\beta_{{1}}}^{2}{x}^{2}+16\,{\rm i}{\beta_{{1}}}^{4}
                  {x}^{2}+128\,{\rm i}{\beta_{{1}}}^{2}t{\alpha_{{1}}}^{4}x-128\,{\rm i}{\beta_{{1}}
                  }^{6}xt-8\,{\beta_{{1}}}^{2}x\\
                  &+256\,{\rm i}{\beta_{{1}}}^{2}{t}^{2}{\alpha_{{
                  1}}}^{6}+256\,{\rm i}{\beta_{{1}}}^{8}{t}^{2}-32\,t{\alpha_{{1}}}^{2}{\beta_
                  {{1}}}^{2}+{\rm i}+64\,t{\beta_{{1}}}^{4}.
          \end{align*}

       While $x$ (or $t$) goes to infinity, $|q_2^{[2]}|^2$ goes to $4\beta_1^2$. The maximum amplitude of $|q_2^{[2]}|^2$ equals to $36\beta_1^2$ locating at $(0,\,0)$, and the minimum amplitude of $|q_2^{[2]}|^2$ is $0$ locating at $(\,{\frac {3\sqrt {3}{\alpha_{{1}}}^{2}}{4\sqrt {4\,{\alpha_{{1}}}^{2}
              +{\beta_{{1}}}^{2}} \left( {\alpha_{{1}}}^{2}+{\beta_{{1}}}^{2}
              \right) \beta_{{1}}}}$,\\
       $-{\frac {\sqrt {3}}{16\sqrt {4\,{\alpha_{{1}}}^{2}+{\beta_{{1}}}^{
        2}} \left( {\alpha_{{1}}}^{2}+{\beta_{{1}}}^{2} \right) \beta_{{1}}}}
       )$ and $(-\,{\frac {3\sqrt {3}{\alpha_{{1}}}^{2}}{4\sqrt {4\,{\alpha_{{1}}}^{2}
              +{\beta_{{1}}}^{2}} \left( {\alpha_{{1}}}^{2}+{\beta_{{1}}}^{2}
              \right) \beta_{{1}}}}
       ,\,{\frac {\sqrt {3}}{16\sqrt {4\,{\alpha_{{1}}}^{2}+{\beta_{{1}}}^{
        2}} \left( {\alpha_{{1}}}^{2}+{\beta_{{1}}}^{2} \right) \beta_{{1}}}}
       )$. Two figures of the rogue
       wave solution \eqref{1rw} with special parameters are
       displayed in Fig. \ref{fig.1_rogue wave} in order to show its
       typical features: localized profile on whole plane and
       a remarkable peak over asymptotical plane.\\

  Note that the first order rogue wave solution $q^{[2]}_2$ of
  the DNLSII equation is different from the rogue wave of the
   DNLSI equation (see eq.(56) of ref.\cite{JPA44305203}). The degree
   of the polynomial in denominator for the latter is four,
   which is double of former. But the transformation \eqref{gt} can not change this degree of
  two related functions.  Nevertheless, the modulus
  of above mentioned two first order rogue waves are the same.
 That is, they possess the same asymptotical plane and amplitude.  \\

\noindent{\bf 3.3 Rogue waves from breathers by higher order Taylor expansion}\\


In this subsection, we shall construct  higher order rogue waves of the
DNLSII equation according to theorems \ref{thm_nDT} and
\ref{thm reductionqn} under double degeneration
\cite{PRE85026601} including eigenvalue degeneration $\lambda_j
\rightarrow \lambda_0$ $(j=1,3,5,\dots, 2k-1)$
and eigenfunction degeneration  $\Phi_j(\lambda_0)=0$ in $T_n$ with $n=2k$.

According to the principle of linear superposition and lemma \ref{relation},
we construct a new general eigenfunction associated with $\lambda_j$ as
  \begin{equation}\label{eigenfunnozero_1}
    \Phi_j=\left(\begin{matrix}
      f_j(x,t,\lambda_j)\\
      g_j(x,t,\lambda_j)
    \end{matrix}\right)=\left(\begin{matrix}
     C_1\psi_{11}(x,t,\lambda_j)+C_2{{\psi}_{12}}^*(x,t,-{\lambda_j}^*)\\
     C_1\psi_{12}(x,t,\lambda_j)+C_2{{\psi}_{11}}^*(x,t,-{\lambda_j}^*)\\
    \end{matrix}\right),
  \end{equation}
from $\Psi_1(\lambda)$ in \eqref{originaleigenfunforp}. Here two
superposition coefficients are $C_1={\rm exp}\left({\rm
i}s\sum_0^{k-1}\epsilon^{2m}l_{m}\right),\ C_2={\rm
exp}\left({\rm -i}s\sum_0^{k-1}\epsilon^{2m}l_{m}\right), \
m=0,1,2,3\cdots k-1, l_m\in\mathbb{C}$. If
$l_m=0$ $(m=0,1,2,\cdots,k-1)$, $\Phi_j$ in \eqref{eigenfunnozero_1} is
the same as that in \eqref{eigenfunnozero}.  Note that $\Phi_j$ $(j=1,3,5,\cdots,2k-1)$ have the same expression derived from \eqref{originaleigenfunforp} with different eigenvalues $\lambda_j$.
This fact is important to guarantee the degeneration of $T_n$ when eigenvalues
are the same as a special one $\lambda_0$.

It is not difficult to check that $\lambda_0=\sqrt{\frac{-a}{2}}+{\rm i}c$ $(a<0)$ is only one
common zero of eigenfunctions $\Phi_j$ in \eqref{eigenfunnozero_1} and $s$ in
\eqref{phasefactor}. When $\Phi_j$ be given by \eqref{eigenfunnozero_1}, let $\lambda_j\to\lambda_0$, then theorem 4, theorem 5 and corollary 1 imply that double degeneration occurs in $T_n$, and $q^{[n]}$ is of indeterminate form $\frac 00$. By higher order Taylor
expansion as  in \cite{arxivhe,PRE85026607,hexukp2012, liwuwanghe2013,wangkphe2013,
xuhechengkp2012,Z,G}, $q^{[n]}$ generates the $k$-th order rogue wave of the DNLSII equation.
Here $n=2k$.
\begin{theorem}\label{thm nrogue}
Set $n=2k$, $j=1,3,5,\dots,k$, $a<0$, $\lambda_j=\sqrt{\frac{-a}{2}}+{\rm i}c+\epsilon^2$. Let $\Phi_j$ be given by \eqref{eigenfunnozero_1}.  The  $k$-th order
 rogue wave solution for the DNLSII equation is expressed by
    \begin{equation}\label{qnrogue}
      q^{[n]}=\frac{|\widehat{W_n}'|}{|W_n'|}q-2{\rm i}\frac{|M_n'|}{|W_n'|}
    \end{equation}
    with
    \begin{equation}
    \begin{aligned}
      \widehat{W_n}'=&\left(\left.\frac{\partial^{n_i}}{\partial\epsilon^{n_i}}\right|_{\epsilon=0}
      (\widehat{W_n})_{ij}(\lambda_0+\epsilon^2)\right)_{n\times n},\\
       {W_n}'=&\left(\left.\frac{\partial^{n_i}}{\partial\epsilon^{n_i}}\right|_{\epsilon=0}
      \left({W_n}\right)_{ij}(\lambda_0+\epsilon^2)\right)_{n\times n},\\
      {M_n}'=&\left(\left.\frac{\partial^{n_i}}{\partial\epsilon^{n_i}}\right|_{\epsilon=0}
      \left({M_n}\right)_{ij}(\lambda_0+\epsilon^2)\right)_{n\times n}.\\
    \end{aligned}
    \end{equation}
Here $n_i=i$ if $i$ is odd and $n_i=i-1$ if $i$ is even. This solution has
 free parameters $a,c,l_i$ $(i=0,1,2,\cdots,k-1)$.
  \end{theorem}

For convenience, let $a=-1$ and $c=1$ in following all rogue waves.
Let $n=4$, $l_0=0$ in theorem \ref{thm nrogue}.
The second order rogue wave of the DNLSII equation  is given by
\begin{equation}
    \label{2rw}
  q^{[4]}=\frac{L_1}{L_2}{\rm exp}(-{\rm i}x),
  \end{equation}
  where

  \begin{align*}
    L_1=&-15-9\,{x}^{6}-243\,{t}^{6}+81\,{x}^{2}+207\,{t}^{2}+1323\,{t}^{4}+39
        \,{x}^{4}+1296\,{\rm i}{x}^{2}{t}^{3}\\
        &+1458\,{\rm i}x{t}^{4}+180\,{\rm i}{x}^{4}t+612\,{\rm i}{
        t}^{2}{x}^{3}-54\,{x}^{5}t-567\,{t}^{4}{x}^{2}-486\,{t}^{5}x\\
        &-396\,{t}^{3}{x}^{3}-189\,{x}^{4}{t}^{2}-54\,xt
         +1242\,{x}^{2}{t}^{2}+2052\,x{t}^       {3}+276\,{x}^{3}t-30\,{\rm i}x\\
         &-216\,{\rm i}{t}^{3}-28\,{\rm i}{x}^{3}-228\,{\rm i}t-288\,{\rm i}{x}^
         {2}t-972\,{\rm i}x{t}^{2}+18\,{\rm i}{x}^{5}+972\,{\rm i}{t}^{5}\\
         &+(324\,\sqrt {2}{
           x}^{2}t-216\,{\rm i}xt\sqrt {2}+324\,\sqrt {2}x{t}^{2}-216\,{\rm i}\sqrt {2}{x}^{2}
           +432\,{\rm i}{t}^{2}\sqrt {2}\\
           &+36\,\sqrt {2}t-108\,\sqrt {2}{t}^{3}+36\,
       \sqrt {2}{x}^{3}-96\,{\rm i}\sqrt {2}-36\,\sqrt {2}x) l_{{1}}-216\,{l
            _{{1}}}^{2},\\
    L_2=&-3-9\,{x}^{6}-243\,{t}^{6}-39\,{x}^{2}-441\,{t}^{2}-297\,{t}^{4}+3\,{x
         }^{4}+108\,{\rm i}x{t}^{2}-54\,{x}^{5}t\\
         &-567\,{t}^{4}{x}^{2}-486\,{t}^{5}x-
          396\,{t}^{3}{x}^{3}-189\,{x}^{4}{t}^{2}-54\,xt+162\,{x}^{2}{t}^{2}+324
         \,x{t}^{3}\\
         &-12\,{x}^{3}t-18\,{\rm i}{x}^{5}-18\,{\rm i}x-20\,{\rm i}{x}^{3}-162\,{\rm i}x{t}^{4
         }-216\,{\rm i}{x}^{2}{t}^{3}-72\,{\rm i}{x}^{2}t\\
         &-180\,{\rm i}{t}^{2}{x}^{3}-72\,{\rm i}{x}^{4}
         t+432\,{\rm i}{t}^{3}+48\,{\rm i}t-216\,{l_{{1}}}^{2}+( 216\,{\rm i}xt\sqrt {2}+36\\
          &\ +324\,\sqrt {2}x{t}^{2}+24\,{\rm i}\sqrt {2}-396\,\sqrt {2
         }t+324\,\sqrt {2}{x}^{2}t+216\,{\rm i}{t}^{2}\sqrt {2}-36\,\sqrt {2}x\\
         &\ -108\,
        \sqrt {2}{t}^{3}) l_{{1}}.
  \end{align*}
  There exists only one free parameter $l_1$ in these expressions. Here we have set $l_0=0$
without loss of generality, because it only changes
the location of the rogue wave with the same profile. Other parameters $l_j$ for $j\geq 2$
disappear automatically in $q^{[2]}$ because of the limit $\epsilon \rightarrow 0$.
While $l_1=0$, we get the fundamental pattern of the second order rogue wave,  whose maximum amplitude is
equal to $5$ and locates at (0, 0).  While $l_1=100$, the second order rogue
wave splits up into three first order rogue waves rather than two, which locate
at approximately $(5.77,\, 1.74)$, $(3.63,\,-5.29)$, and $(-9.62,\,3.66)$. The profiles
of two second-order rogue wave are plotted in Fig. \ref{fig.2_rw}.

When $n$ increases, more patterns of $k$-th
rogue wave solution can be obtained from theorem \ref{thm nrogue}.
For instance, when $k=3$, the third order rogue wave with three
parameters, i.e. $l_0,l_1,l_2$, possesses a circular pattern apart from  the fundamental pattern and the triangular pattern. Furthermore,
set $l_j=0$ $(j=0,1,2)$, the third order fundamental pattern of the DNLSII equation
is given by
 \begin{equation}\label{3rw}
  q^{[6]}=\frac{L_1}{L_2}\exp(-{\rm i}x)
\end{equation}
with
  \begin{align*}
L_1=&-81\,{x}^{12}-972\,{x}^{11}t-6318\,{x}^{10}{t}^{2}-27540\,{x}^{9}{t}^{3}
-88695\,{x}^{8}{t}^{4}-219672\,{x}^{7}{t}^{5}-428004\,{x}^{6
}{t}^{6}\\
&-659016\,{x}^{5}{t}^{7}-798255\,{x}^{4}{t}^{8}-743580\,{x}^{3}
{t}^{9}-511758\,{x}^{2}{t}^{10}-236196\,x{t}^{11}-59049\,{t}^{12}\\
&+324\,{\rm i}{x}^{11}+5184\,{\rm i}{x}^{10}t+37260\,{\rm i}{x}^{9}{t}^{2}+171720\,{\rm i}{x}^{8}{t
}^{3}+560520\,{\rm i}{x}^{7}{t}^{4}+1371168\,{\rm i}{x}^{6}{t}^{5}\\
&+2560248\,{\rm i}{x}^{5}{t}^{6}+3674160\,{\rm i}{x}^{4}{t}^{7}+3980340\,{\rm i}{x}^{3}{t}^{8}+3149280\,{\rm i}
{x}^{2}{t}^{9}+1653372\,{\rm i}x{t}^{10}+472392\,{\rm i}{t}^{11}\\
&+1026\,{x}^{10}+15660\,{x}^{9}t+127170\,{x}^{8}{t}^{2}+619920\,{x}^{7}{t}^{3}+2057940
\,{x}^{6}{t}^{4}+4865832\,{x}^{5}{t}^{5}\\
&+8393220\,{x}^{4}{t}^{6}+
10478160\,{x}^{3}{t}^{7}+9163530\,{x}^{2}{t}^{8}+5117580\,x{t}^{9}+
1430298\,{t}^{10}-1620\,{\rm i}{x}^{9}\\
&-29160\,{\rm i}{x}^{8}t-252720\,{\rm i}{x}^{7}{t}^
{2}-1296000\,{\rm i}{x}^{6}{t}^{3}-4247640\,{\rm i}{x}^{5}{t}^{4}-9311760\,{\rm i}{x}^{4
}{t}^{5}-13860720\,{\rm i}{x}^{3}{t}^{6}\\
&-13296960\,{\rm i}{x}^{2}{t}^{7}-7479540\,
{\rm i}x{t}^{8}-1312200\,{\rm i}{t}^{9}+2745\,{x}^{8}-2520\,{x}^{7}t-194220\,{x}^{
6}{t}^{2}\\
&-1344600\,{x}^{5}{t}^{3}-4167450\,{x}^{4}{t}^{4}-7597800\,{x}
^{3}{t}^{5}-9083340\,{x}^{2}{t}^{6}-6561000\,x{t}^{7}-127575\,{t}^{8}\\
&-8760\,{\rm i}{x}^{7}-85920\,{\rm i}{x}^{6}t-212760\,{\rm i}{x}^{5}{t}^{2}-97200\,{\rm i}{x}^{4
}{t}^{3}+113400\,{\rm i}{x}^{3}{t}^{4}+3732480\,{\rm i}{x}^{2}{t}^{5}\\
&+8135640\,{\rm i}x{t}^{6}+4024080\,{\rm i}{t}^{7}+21020\,{x}^{6}+12600\,{x}^{5}t-56700\,{x}^{4}
{t}^{2}-874800\,{x}^{3}{t}^{3}-413100\,{x}^{2}{t}^{4}\\
&+4791960\,x{t}^{5
}+9977580\,{t}^{6}-33480\,{\rm i}{x}^{5}-313200\,{\rm i}{x}^{4}t-630000\,{\rm i}{x}^{3}{
t}^{2}-1425600\,{\rm i}{x}^{2}{t}^{3}\\
&-3450600\,{\rm i}x{t}^{4}-9545040\,{\rm i}{t}^{5}-
41775\,{x}^{4}-318300\,{x}^{3}t-2031750\,{x}^{2}{t}^{2}-3410100\,x{t}^
{3}\\
&-3983175\,{t}^{4}+16500\,{\rm i}{x}^{3}+237600\,{\rm i}{x}^{2}t+1147500\,{\rm i}x{t}^
{2}-1171800\,{\rm i}{t}^{3}-28350\,{x}^{2}+121500\,xt\\
&-166950\,{t}^{2}+6300\,
{\rm i}x+63000\,{\rm i}t+1575,\\
L_2=&81\,{x}^{12}+972\,{x}^{11}t+6318\,{x}^{10}{t}^{2}+27540\,{x}^{9}{t}^{3
}+88695\,{x}^{8}{t}^{4}+219672\,{x}^{7}{t}^{5}+428004\,{x}^{6}{t}^{6}\\
&+659016\,{x}^{5}{t}^{7}+798255\,{x}^{4}{t}^{8}+743580\,{x}^{3}{t}^{9}+
511758\,{x}^{2}{t}^{10}+236196\,x{t}^{11}+59049\,{t}^{12}\\
&+324\,{\rm i}{x}^{11}+3240\,{\rm i}{x}^{10}t+17820\,{\rm i}{x}^{9}{t}^{2}+64800\,{\rm i}{x}^{8}{t}^{3}+
171720\,{\rm i}{x}^{7}{t}^{4}+340848\,{\rm i}{x}^{6}{t}^{5}\\
&+515160\,{\rm i}{x}^{5}{t}^{6}+583200\,{\rm i}{x}^{4}{t}^{7}+481140\,{\rm i}{x}^{3}{t}^{8}+262440\,{\rm i}{x}^{2}{t}^
{9}+78732\,{\rm i}x{t}^{10}-378\,{x}^{10}\\
&-2700\,{x}^{9}t-20250\,{x}^{8}{t}^{2}-101520\,{x}^{7}{t}^{3}
-340740\,{x}^{6}{t}^{4}-775656\,{x}^{5}{t}^{5}-1180980\,{x}^{4}{t}^{6}\\
&-1146960\,{x}^{3}{t}^{7}-503010\,{x}^{2}{t}^{8}+131220\,x{t}^{9}
+301806\,{t}^{10}+540\,{\rm i}{x}^{9}+6480\,{\rm i}{x}^{8}t+
6480\,{\rm i}{x}^{7}{t}^{2}\\
&-73440\,{\rm i}{x}^{6}{t}^{3}-372600\,{\rm i}{x}^{5}{t}^{4}-
933120\,{\rm i}{x}^{4}{t}^{5}-1419120\,{\rm i}{x}^{3}{t}^{6}-1516320\,{\rm i}{x}^{2}{t}^
{7}-1006020\,{\rm i}x{t}^{8}\\
&-524880\,{\rm i}{t}^{9}+495\,{x}^{8}+2520\,{x}^{7}t+
51660\,{x}^{6}{t}^{2}+255960\,{x}^{5}{t}^{3}+506250\,{x}^{4}{t}^{4}+
443880\,{x}^{3}{t}^{5}\\
&+2201580\,{x}^{2}{t}^{6}+4228200\,x{t}^{7}+
3189375\,{t}^{8}+5640\,{\rm i}{x}^{7}+35760\,{\rm i}{x}^{6}t+145800\,{\rm i}{x}^{5}{t}^{
2}\\
&+324000\,{\rm i}{x}^{4}{t}^{3}-145800\,{\rm i}{x}^{3}{t}^{4}-3013200\,{\rm i}{x}^{2}{t
}^{5}-3411720\,{\rm i}x{t}^{6}-2099520\,{\rm i}{t}^{7}+12340\,{x}^{6}\\
&+70920\,{x}^{
5}t+369900\,{x}^{4}{t}^{2}+874800\,{x}^{3}{t}^{3}+4009500\,{x}^{2}{t}^
{4}+7727400\,x{t}^{5}+4544100\,{t}^{6}\\
&+23640\,{\rm i}{x}^{5}+43200\,{\rm i}{x}^{4}
t+61200\,{\rm i}{x}^{3}{t}^{2}-842400\,{\rm i}{x}^{2}{t}^{3}+631800\,{\rm i}x{t}^{4}+
583200\,{\rm i}{t}^{5}-7425\,{x}^{4}\\
&+30300\,{x}^{3}t-322650\,{x}^{2}{t}^{2}+
40500\,x{t}^{3}+2330775\,{t}^{4}+11700\,{\rm i}{x}^{3}+23400\,{\rm i}{x}^{2}t-
245700\,{\rm i}x{t}^{2}\\
&-993600\,{\rm i}{t}^{3}+8550\,{x}^{2}+900\,xt+291150\,{t}^{
2}+2700\,{\rm i}x-18000\,{\rm i}t+225.
  \end{align*}

All patterns of the third order rogue are
displayed in Fig. \ref{fig.3_rw} by using formula (\ref{qnrogue}).

Set $k=4$ in theorem \ref{thm nrogue}, then $q^{[8]}$ in \eqref{qnrogue}
gives the fourth order rogue wave of the DNLSII equation. This solution
possesses as least four patterns: a fundamental pattern,
a triangular pattern, a circular-fundamental pattern and
a circular-triangular pattern. These four patterns of rogue wave
are displayed in Fig. \ref{fig.4_rw} by using formula (\ref{qnrogue}).

All the above-mentioned rogue waves are plotted by using the analytical
formulas according to theorem \ref{thm nrogue}. Their validity has been
confirmed analytically by a simple symbolic computation. In addition,  we have got explicit expressions of
other higher order rogue waves and their figures. However, they are too long to
write out here. Based on these figures and analytical formulas of
the rogue waves, we can provide following conjectures for the DNLSII
 equation.\\[-20pt]
\begin{remark}
\begin{itemize}
\mbox{\hspace{-1cm}}
\item 1) Under the fundamental pattern, the height of above-mentioned four rogue waves is 3, 5, 7, 9,
respectively. We conjecture that the maximum amplitude  of the
$k$-th order rogue wave of the DNLSII equation is $(2k+1)c$. Here $c$ is the
asymptotical height of rogue wave.
\item 2) The degree of the polynomials in above-mentioned four rogue waves
is 2, 6, 12, 20 respectively. We conjecture that the degree of the
polynomials in the $k$-th order rogue wave of the DNLSII equation is $k(k+1)$.
\item 3) The maximum number of peaks in above-mentioned four rogue waves
is 1, 3, 6, 10, respectively. We conjecture that the
$k$-th order rogue wave of the DNLSII equation can be decomposed thoroughly
into $\dfrac{k(k+1)}{2}$ first order rogue waves.
\item 4) From Figs.\ref{fig.3_rw}(c),\ref{fig.4_rw}(c) and \ref{fig.4_rw}(d), we
conjecture that the $k$-th order rogue wave of the DNLSII equation can be decomposed into a
pattern of a $(k-2)$-th order rogue wave plus an outer ring on which there are
$2k-1$ first-order rogues. Note that the central profile of these three figures
is a first-order rogue wave, a fundamental pattern and a triangular pattern of the
second-order rogue wave, respectively.\\[-20pt]
\end{itemize}
\end{remark}
These conjectures can be regarded  as natural extensions of the
counterparts \cite{Ankiewicz_2010_122002,Kedziora_2011_56611,matveev2011nlsKPI,Gaillard_2011_435204} for the NLS equation.
It is worth mentioning again that the item 2 is not true for the DNLSI
equation. This fact shows again that rogue waves of this equation can not
be obtained by transformation (\ref{gt}) from rogue waves of the DNLSI equation,
because transformation (\ref{gt}) does not change the degree in
two related solutions.
\section{Conclusion and discussion}
In this paper, the DNLSII equation, which is an important model to
describe the propagation of light pulse involving self-steepening without
concomitant self-phase-modulation \cite{PRA76021802}, has been
studied from the point of view of the analytical theory for integrable
systems. The main results of the paper are:\\
1) Present the determinant representation $T_n$ of the coupled DNLSII equations (theorem \ref{thm_nDT});\\
2) Under choice \eqref{choicesoflambdas} on the eigenvalues and
eigenfunctions, $T_n$ is reduced to the  $n$-fold   DT of the DNLSII
equation and $q^{[n]}$ generated by this reduced $T_n$ is a
solution of it (theorem \ref{thm reductionqn});\\
3) Obtain an analytical formula of the $k$-th order   rogue waves for
the DNLSII equation (theorem \ref{thm nrogue}).\\
In addition, several breathers and lower order rogue waves are given
through analytical expressions and their profiles are plotted
analytically in figures. As $n$-times iteration of the $T_1$, these
$T_n$ of the DNLSII equation include complicated integrals
 like $\int q^{[j-1]}r^{[j-1]}\mathrm{d}x$ $(j=1,2,\cdots,n)$.
To our best knowledge, this kind of DT is unusual by comparing
with known DT for soliton equations. These integrals are not
calculable explicitly in general case, and thus result in a
remarkable difficulty to construct $T_n$ in terms of determinants
and to get explicit solutions of the DNLSII equation. We have
introduced functions $H^{[i]}$ $(i=0,1,2, \cdots,n-1)$ and proved that
$\dfrac{g^{[k]}_i}{f^{[k]}_i}H^{[k+1]}$ is a constant for $i\geq
k+1$ in theorem \ref{thm_simplify}. Finally, we have eliminated
these unfavorable integrals in $T_n$ by using this crucial fact.

It is a well-accepted idea \cite{Matveev,Gu2} in the research
community of the mathematical physics that an integrable model can
be solved by means of DT. From this sense, the DNLSII equation,
which was introduced as an integrable model at 1979 \cite{PS20490},
is a rare exceptional example because its DT has not been
established over the past 30 years. Thus, it is a long-standing
problem to construct the DT of this equation. We have solved this
problem by establishment of the determinant representation of $T_n$
in theorems \ref{thm_nDT} and \ref{thm reductionqn} for the DNLSII
equation. Moreover, breather in \eqref{breathermodulus} and rogue
waves in eqs.(\ref{1rw}),(\ref{2rw}) and (\ref{3rw}), can not be
obtained from counterparts of the DNLSI equation by gauge
transformation (\ref{gt}). This observation shows that it is
necessary to get explicit solution $q$ of the DNLSII equation, and
then clarifies a well-known doubt of the necessity for solving
DNLSII equation. This doubt originates from gauge transformation
(\ref{gt}) \cite{ JMP253433,JPSJ52394} between the DNLSII equation
and the DNLSI equation.

The solutions presented in this paper have wide relevance of physics
because the DNLSII equation is a physical model of light pulses
\cite{PRA76021802}. On the one side, two typical features, i.e. high
amplitude (or high energy equivalently in optics) and localized
property of the rogue wave pulse, may result in high possibility
to observe the self-steepening effect in light pulses. It is
difficult to realize this in normal state for nonlinear propagation
of light \cite{PRA76021802}. On the other side, with the solutions in this
paper we have a new opportunity to observe optical rogue wave by
using the balance between the dispersion and self-steepening instead
of the self-phase-modulation in optical fiber. This will be helpful
to use and control rogue wave in the future. For example, as  in the
case of the NLS equation \cite{arxivhe}, different patterns of the
second order rogue wave here  can also be controlled by phase
modulation at the interaction area of two first order breathers.
This can be realized by choosing different values of $l_0$ in
eigenfunction (\ref{eigenfunnozero_1}) with $l_i=0$ $(i>0)$.
More specifically, set $n = 4$, $a = -1$, $c = 1$, and let $\Phi_j$ $ (j = 1,3)$ be given by \eqref{eigenfunnozero_1}, then by theorem \ref{thm reductionqn}, $q^{[4]}$  in \eqref{qn}
generates a second-order breather $q_{2b}$
of the DNLSII equation with parameters $a,c,\lambda_1,\lambda_3$ and
$l_0$.  Firstly, if $l_0=0$ in $\Phi_1$ and $\Phi_3$, the two
breathers in $q_{2b}$ have no phase difference. Set
$\lambda_1=0.7+0.5{\rm i}$ and $\lambda_3=0.75+0.5{\rm i}$ in
$q_{2b}$, we get a ``fundamental pattern" at interaction area of the
second-order breather, which is plotted in fig.
\ref{fig.2_breather_1}. It is obvious that the central profile is
similar to a fundamental pattern in fig. \ref{fig.2_rw} (a).
Secondly, let $l_0=6.25$ and $\lambda_1=0.68+0.5{\rm i}$ in
$\Phi_1$, and $l_0=-8{\rm i}$ and $\lambda_3=0.8+0.5{\rm i}$ in
$\Phi_3$, then $q_{2b}$ displays a ``triangular pattern" at
interaction area  which is plotted in fig. \ref{fig.2_breather_2}.
It is obvious that the central profile is similar to the triangular
pattern in fig. \ref{fig.2_rw}(b). Of course, we can get more
interesting patterns by tuning the phase of breathers at interaction
area for higher order
case.  More detailed studies  will be done in the near future.

Further more, based on the DT we constructed, we can study other interesting phenomena in such systems from unstable plane wave background, such as the so-called super-regular solitonic solutions \cite{zakharov2013}, and the integrable soliton turbulence \cite{Shrira2010}. The study of the evolution of rational rogue waves in DNLSII in presence of an additional noise by a numerical algorithm will be realized in near future \cite{calini2014}.
\\
{\bf Acknowledgments}  This work is supported by the NSF of China under Grant No.11271210 and No.11171073,
the K. C. Wong Magna Fund in Ningbo University and the Natural
Science Foundation of Ningbo under Grant No. 2011A610179. J.S. He thanks
sincerely Prof. A.S. Fokas for arranging the visit to Cambridge
University in 2012-2014 and for many useful discussions. We thank very much two referee for helpful and detailed suggestions
on our initial submission.

\newpage
\begin{figure}[!htbp]
\centering
\raisebox{20 ex}{$|q_1^{[2]}|\,$}\subfigure[]{\includegraphics[height=5cm,width=5cm]{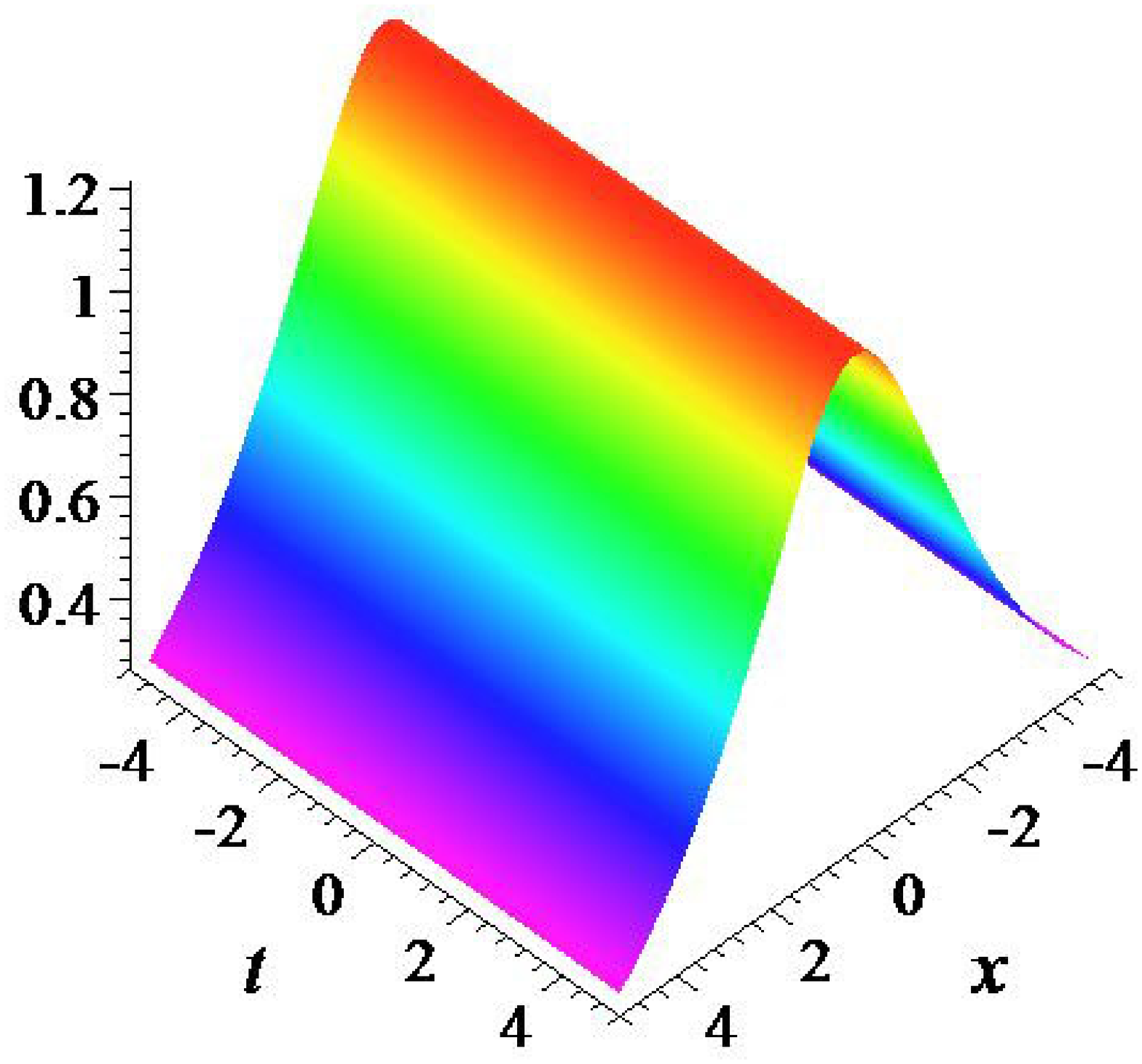}}
\qquad
\raisebox{20 ex}{$|q_2^{[2]}|\,$}\subfigure[]{\includegraphics[height=5cm,width=5cm]{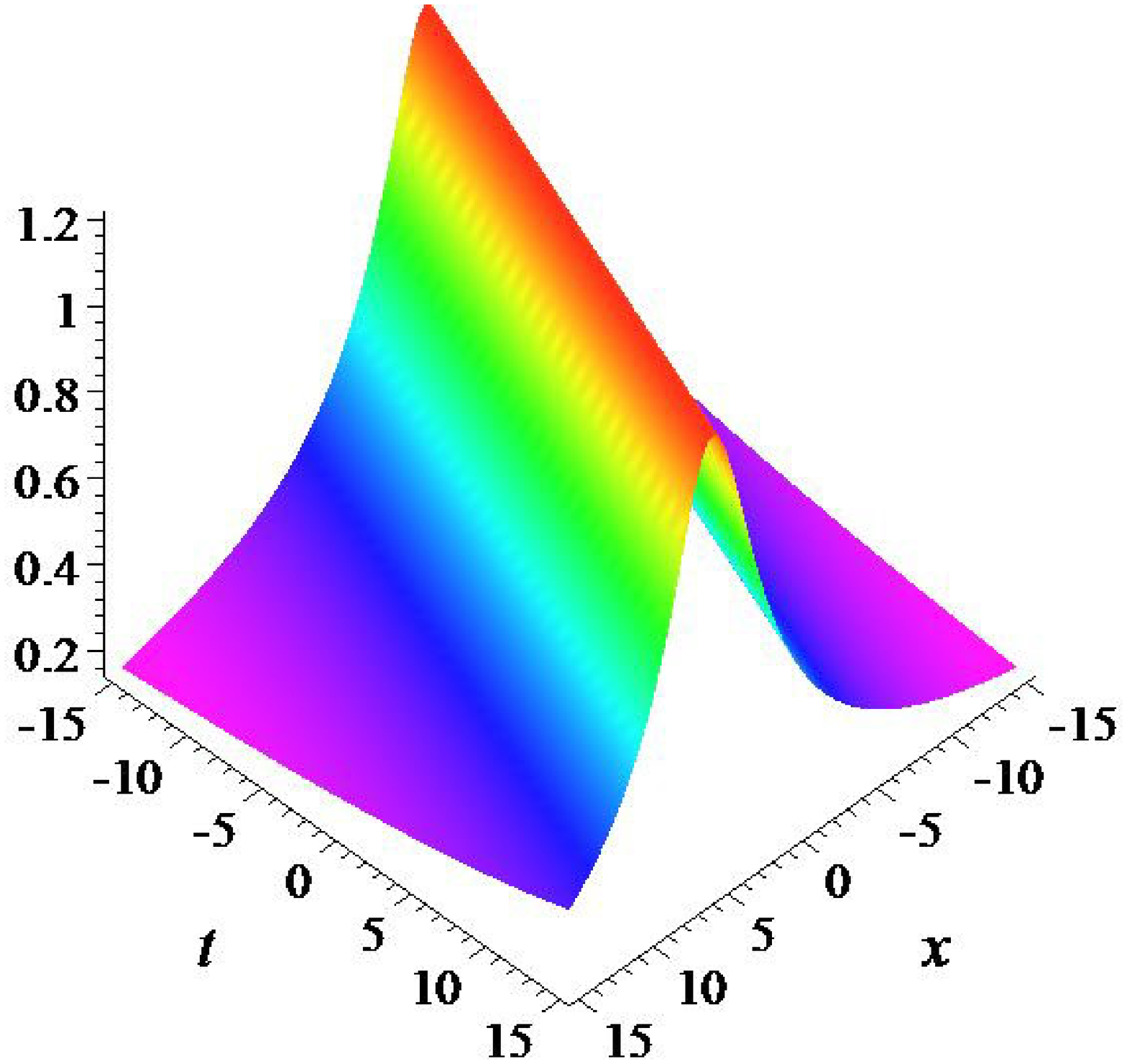}}
\\
\raisebox{20 ex}{$|q_1^{[2]}|\&|q_2^{[2]}|\,$}\subfigure[]{\includegraphics[height=5cm,width=5cm]{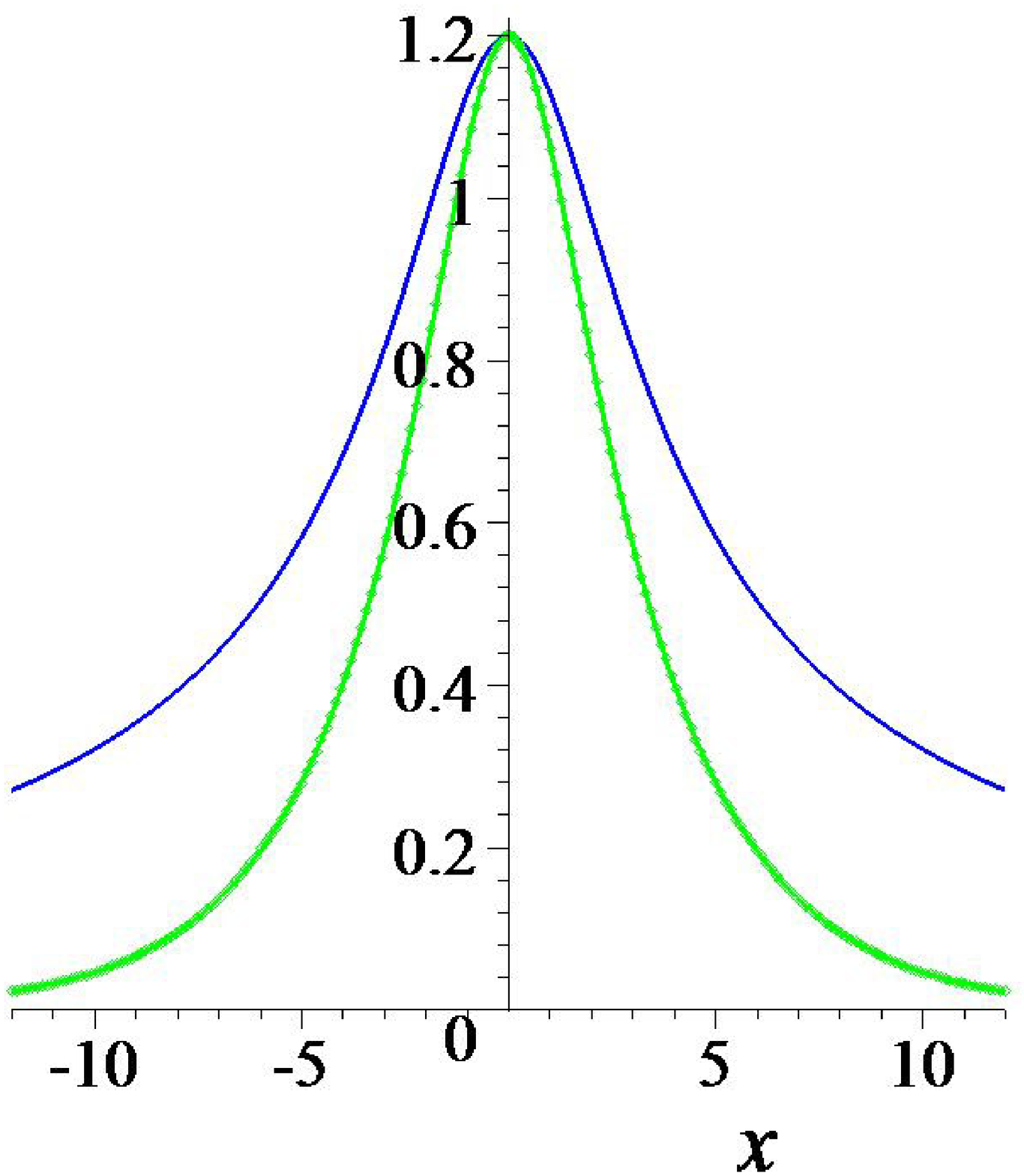}}
\qquad
\raisebox{20 ex}{$|q_2^{[2]}|\,$}\subfigure[]
{\includegraphics[height=5cm,width=5cm]{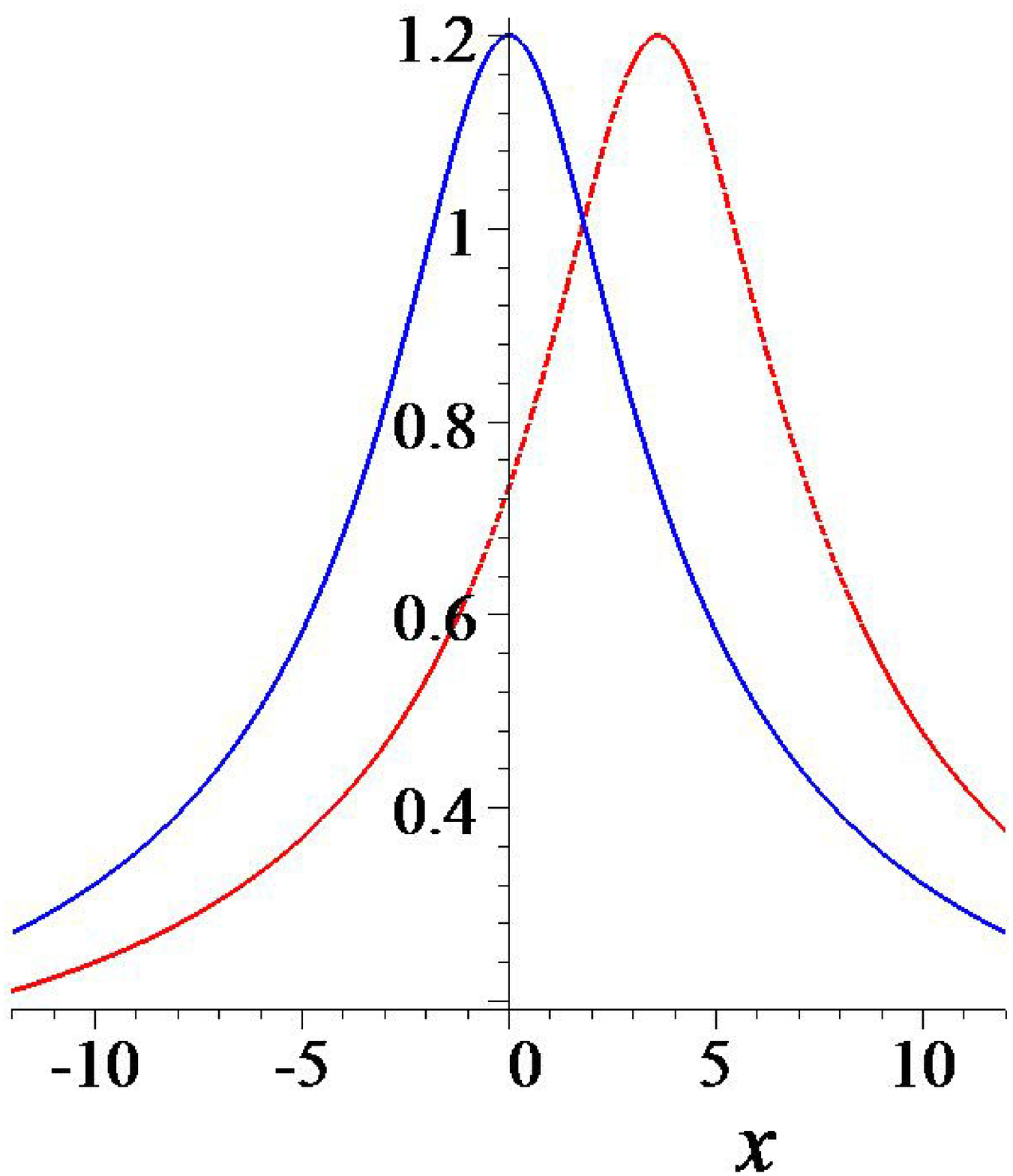}}
\caption{ Profiles of the singles-solitons generated from a vacuum solution.
(a) A usual single-soliton
with $\alpha_1=0.3,\,\beta_1=0.3$ locating near $x=0$.
(b) A rational single-soliton with $\beta_1=0.3$
locating near $x=0.36t$.
(c) A usual single-soliton (the green dot line) and
a rational single-soliton (the blue solid line) at $t=0$.
 They possess the same amplitude, but the rational
 single-soliton is steeper than  the usual one.
 (d) The rational single-soliton at $t=0$
 (the blue solid line) and $t=10$ (the red dash line).
 The blue line reaches its peak at $x=0$,
 the red one at $x=3.6$.}\label{fig_soliton}
\end{figure}


\begin{figure}[!htbp]
\centering
\raisebox{20 ex}{${|q^{[1]}|}^{2}\,$}\subfigure[]{\includegraphics[height=4cm,width=4cm]{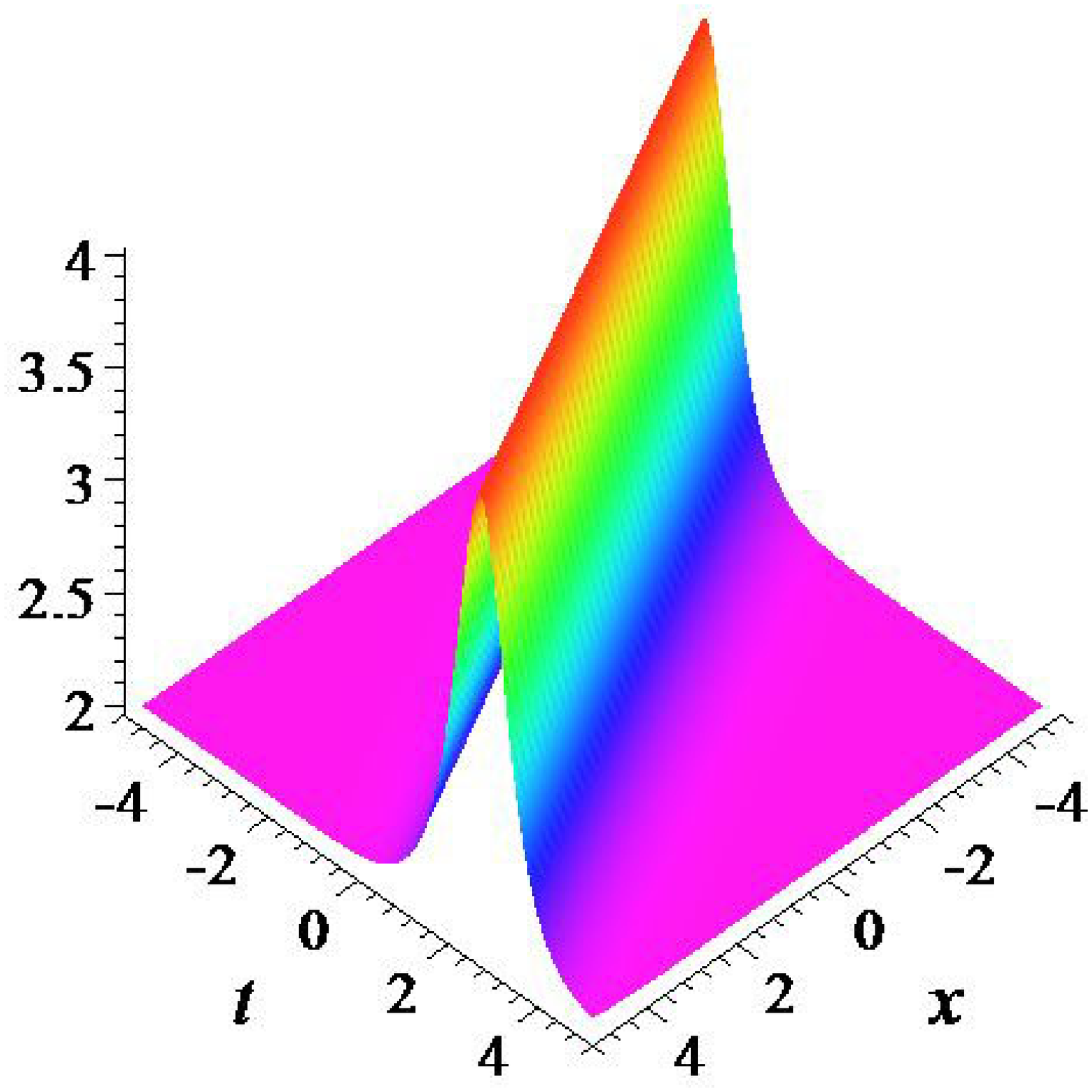}}
\quad
\raisebox{20 ex}{${|q^{[1]}|}^{2}\,$}\subfigure[]{\includegraphics[height=4cm,width=4cm]{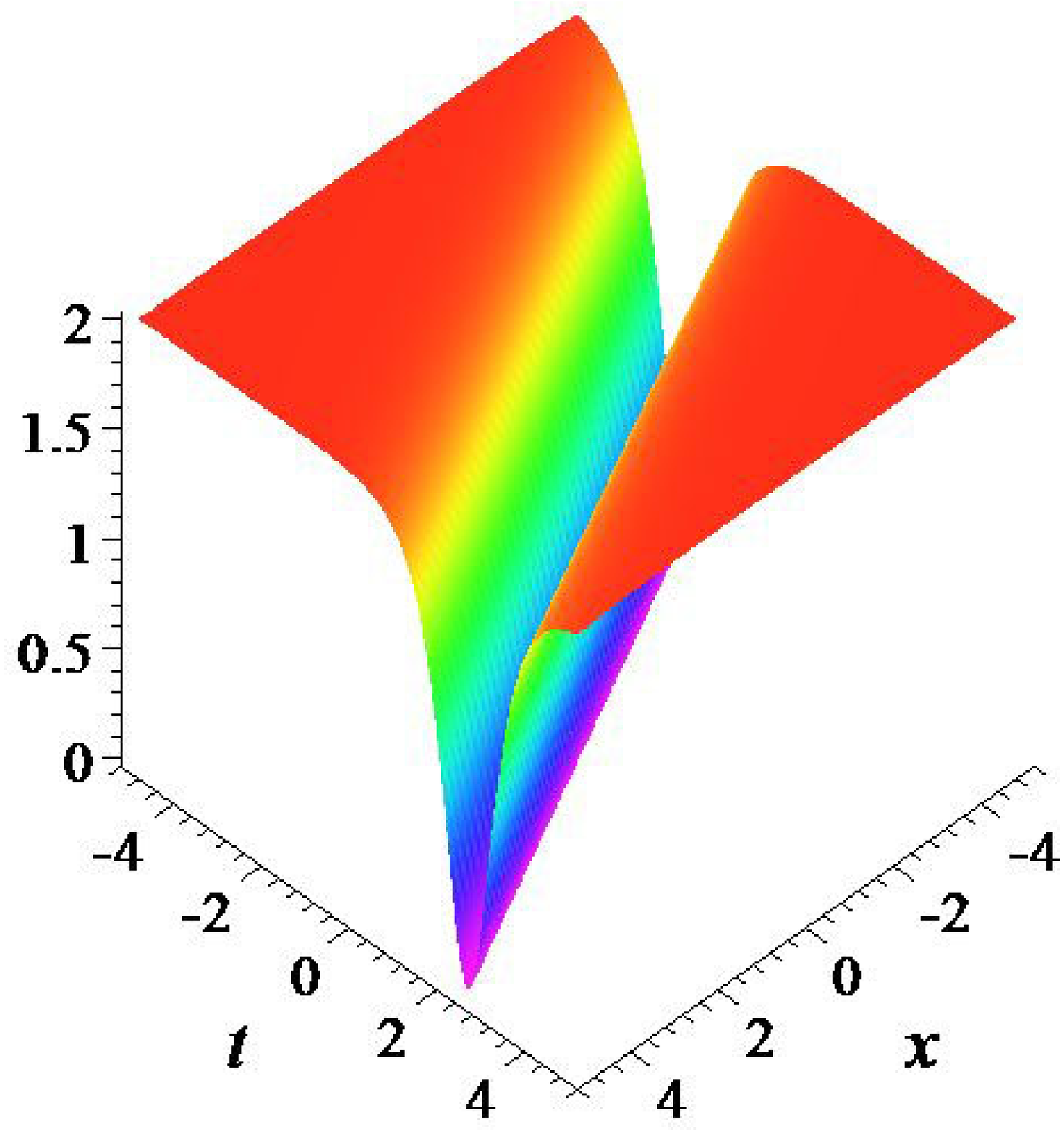}}
\quad
\raisebox{20 ex}{${|q^{[1]}|}^{2}\,$}\subfigure[]{\includegraphics[height=4cm,width=4cm]{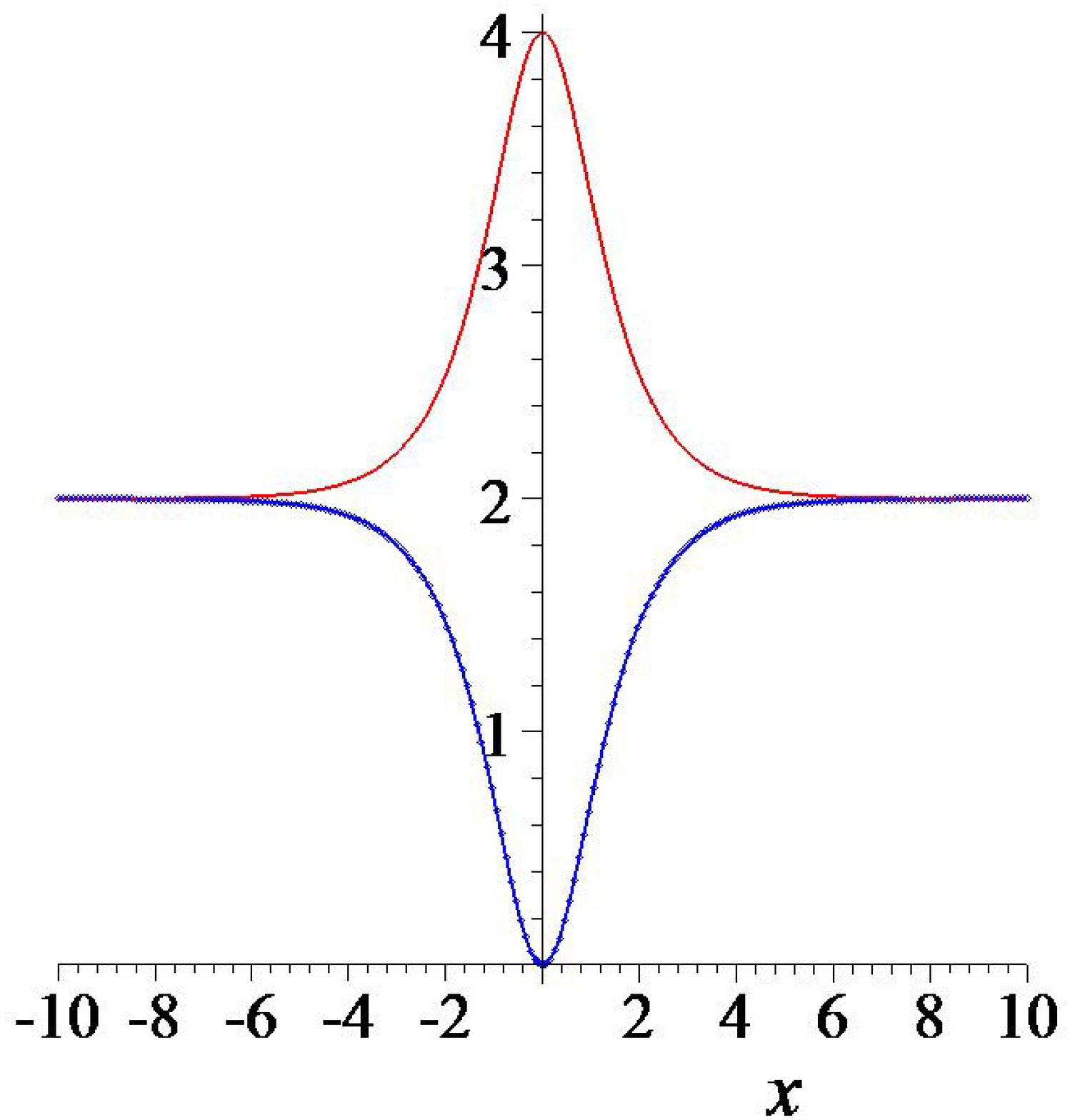}}
\caption{ Profiles of single-solitons generated from a periodic seed.
(a) A bright soliton with  $a=1,\,c=1$ and $\beta=0.5$. (b) A dark soliton
with $a=1,\,c=1$ and $\beta=-0.5$.
(c) The graphical superposition of  a bright soliton and a
 dark soliton at $t=0$.
 Both the bright soliton and the dark soliton achieve
 their peaks at the same time, and share the same non-vanishing boundary
 condition.}\label{fig.1_solitonfromnonzero}
\end{figure}


\begin{figure}[!htbp]
\centering
\raisebox{20 ex}{${|q_1^{[2]}|}\,$}\subfigure[]{\includegraphics[height=4cm,width=4cm]{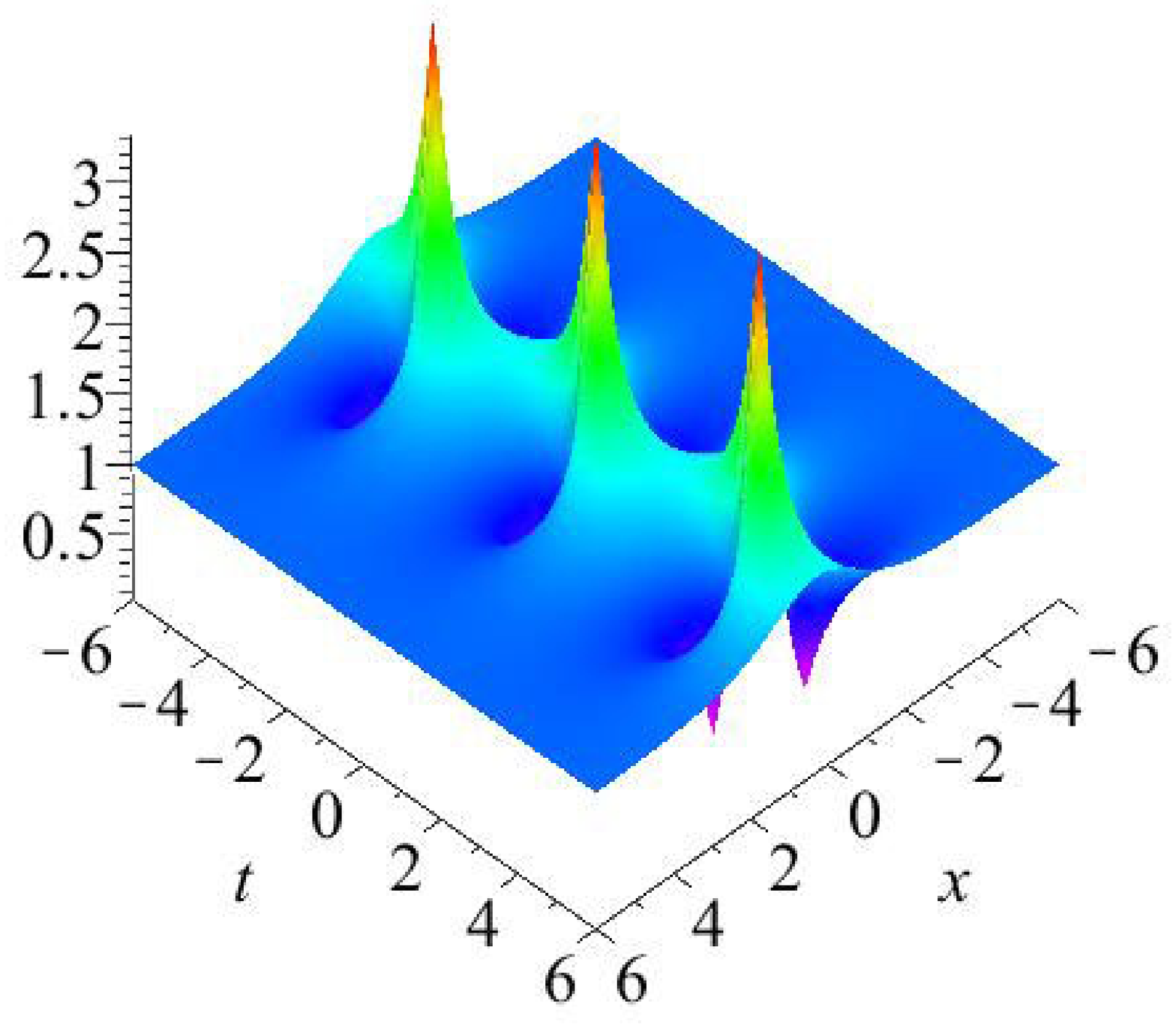}}\qquad
\raisebox{20 ex}{${|q_1^{[2]}|}\,$}\subfigure[]{\includegraphics[height=4cm,width=4cm]{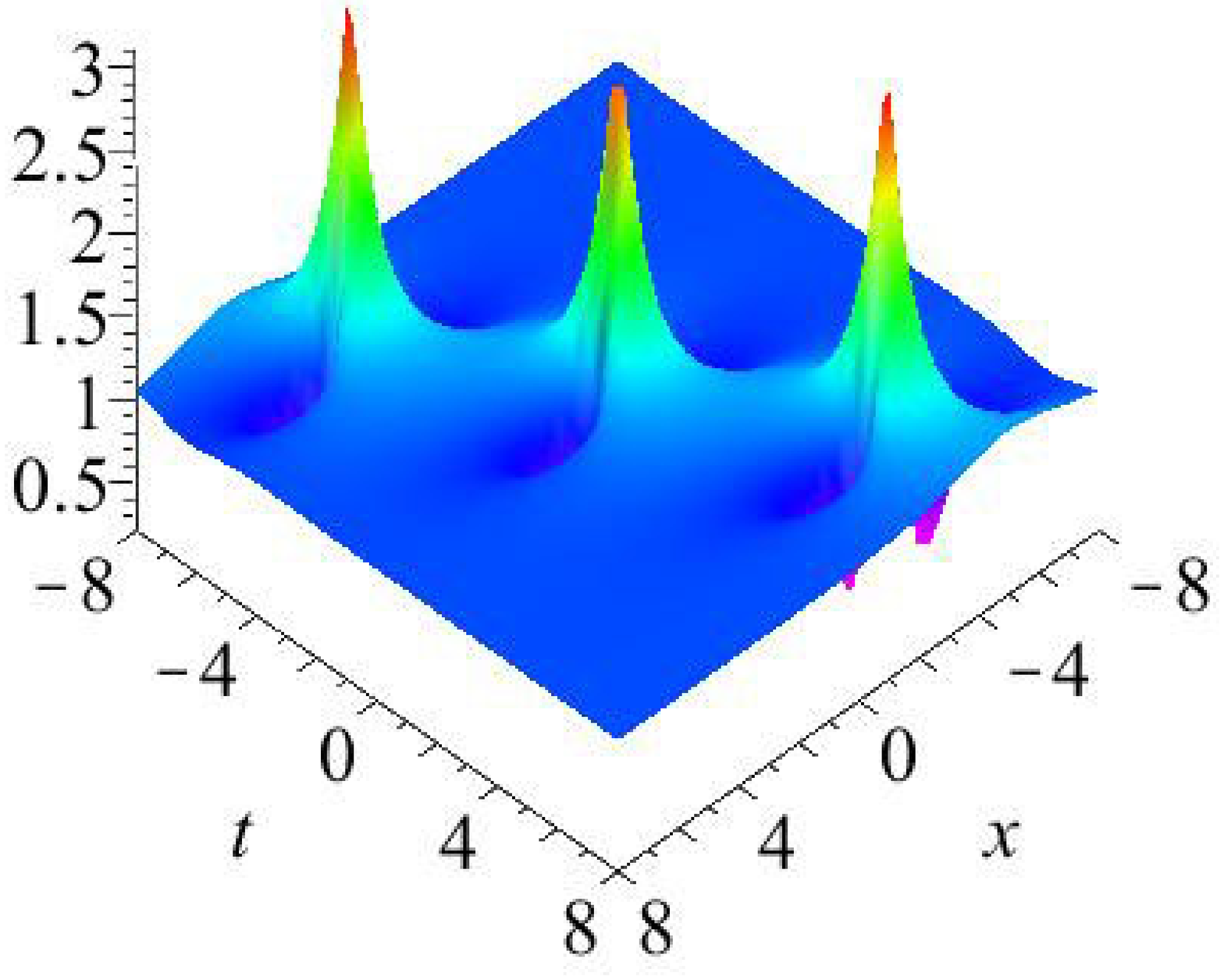}}\qquad
\raisebox{20 ex}{${|q_1^{[2]}|}\,$}\subfigure[]{\includegraphics[height=4cm,width=4cm]{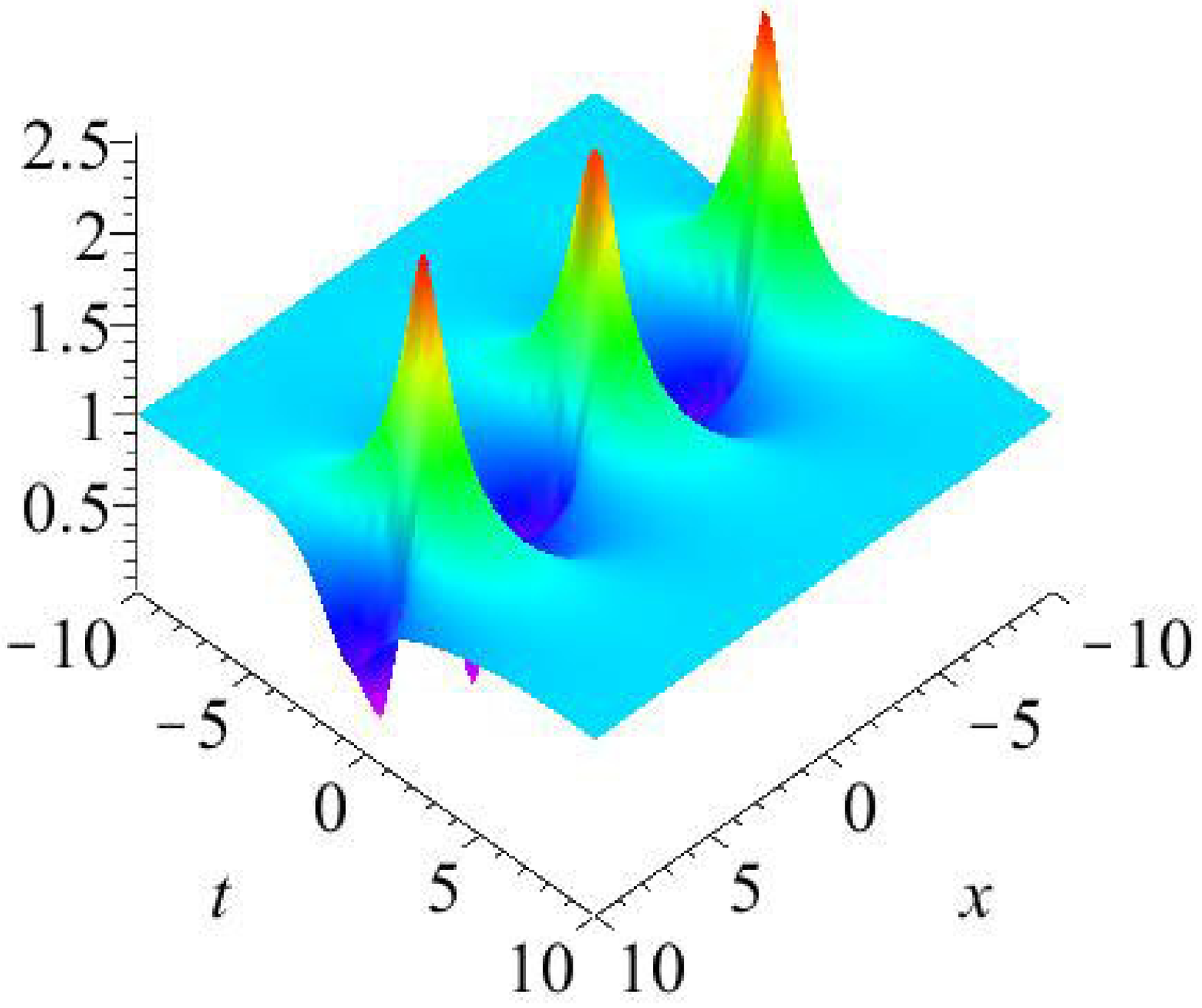}}\quad
\caption{The evolution of three breathers.
(a) A temporal periodic breather with $\alpha_1=0.6,\,\beta_1=0.6$ and $c=1$.
(b) A breather with a certain angle with $x$-axis and $t$-axis under
specific parameters $\alpha_1= 0.68, \beta_1= 0.55$ and $c = 1$.
(c) A spacial periodic breather with $\alpha_1=0.5,\,\beta_1=0.4$ and $c=1$.}\label{fig.breather}
\end{figure}


\begin{figure}[!htbp]
\centering
\raisebox{20 ex}{$|q_2^{[2]}|\,$}\subfigure[]{\includegraphics[height=5cm,width=5cm]{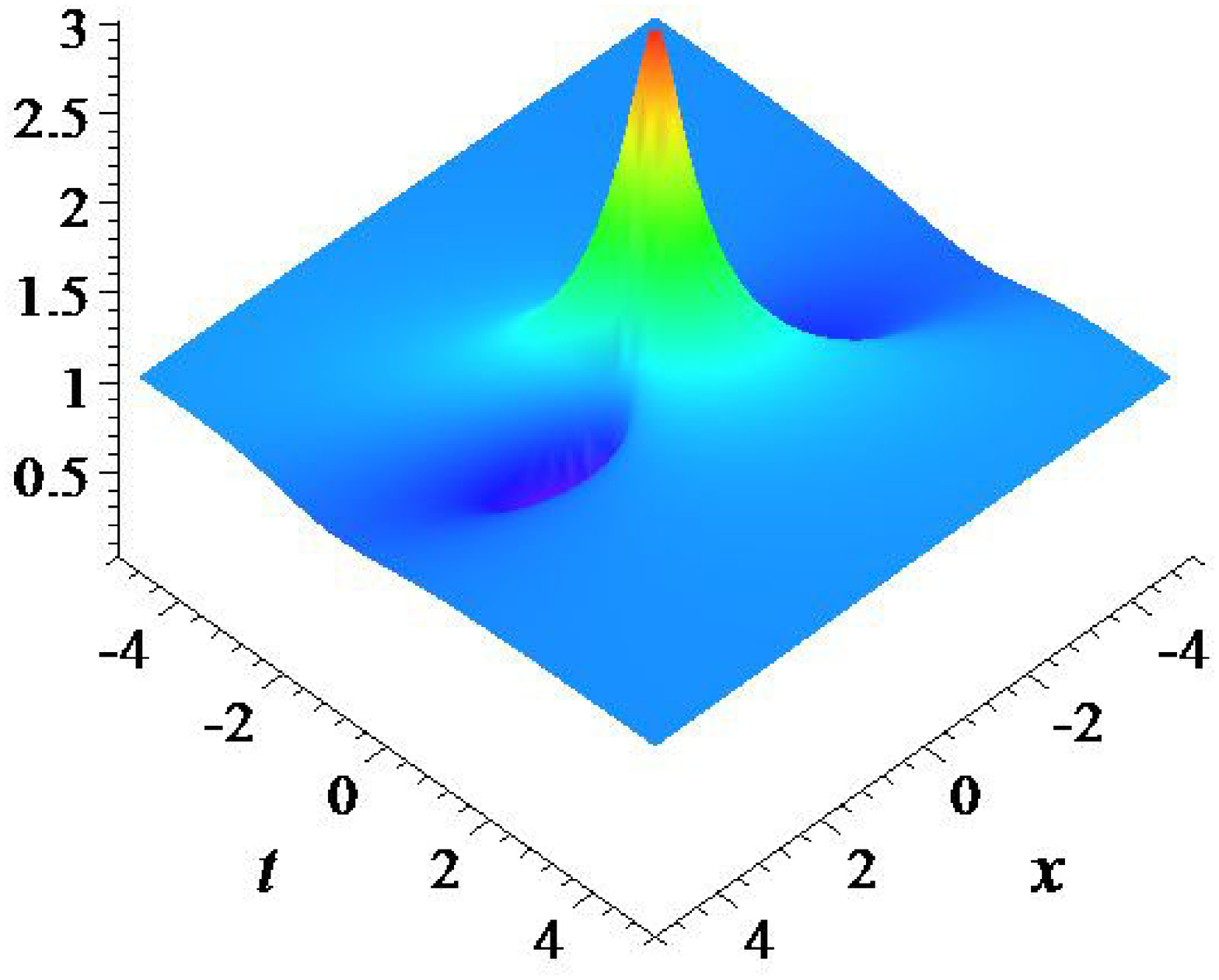}}
\qquad
\raisebox{20 ex}{$|q_2^{[2]}|\,$}\subfigure[]{\includegraphics[height=5cm,width=5cm]{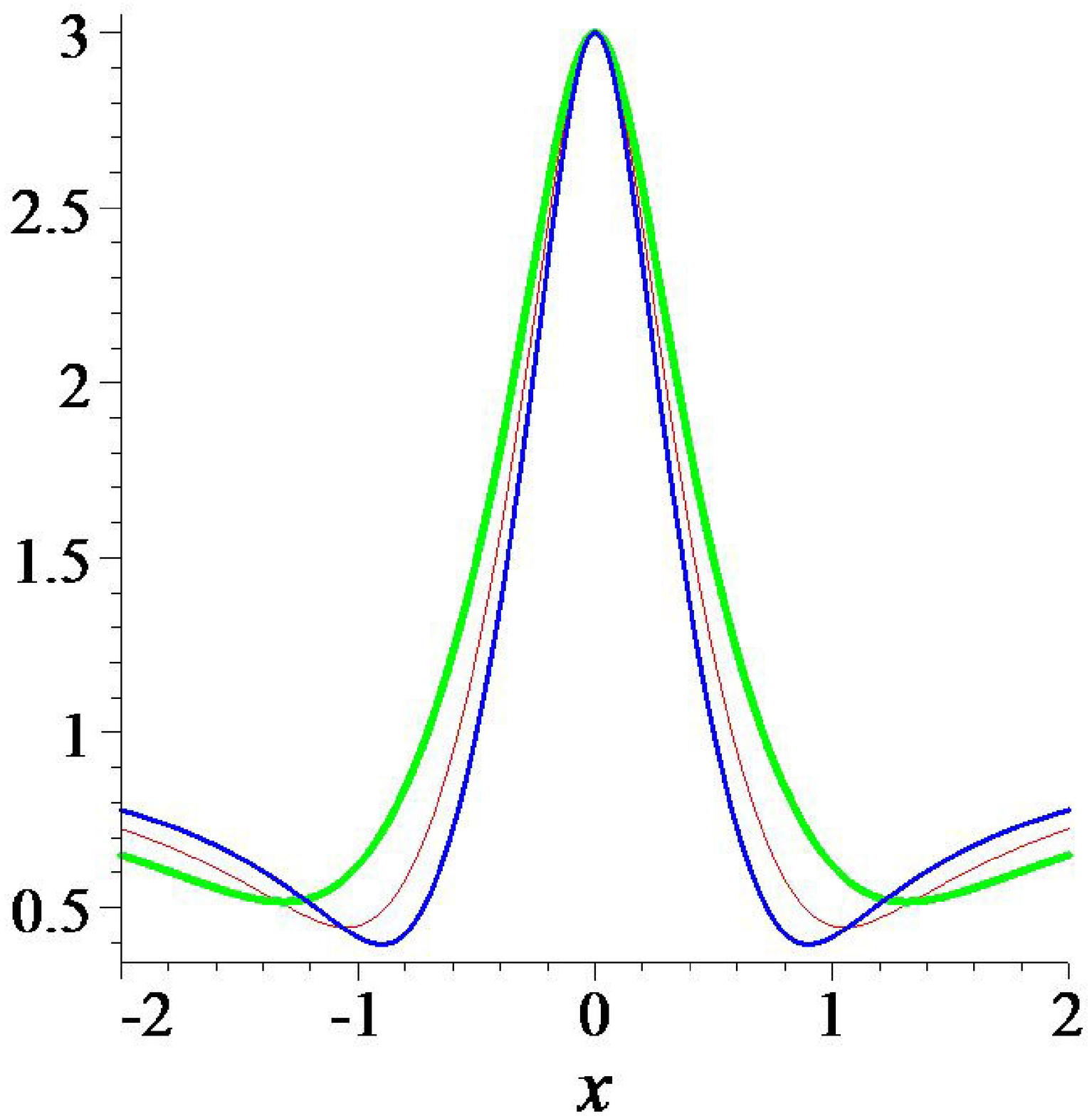}}
\caption{
Profile of the first order rogue wave.
(a) The evolution of rogue wave \eqref{1rw} on ($x,t$)-plane
with $\alpha_1=\frac{\sqrt{2}}{2}$  and $\beta_1=\frac{1}{2}$.
Its asymptotic height  is  $|q_{rw}^{[2]}|=1$,
the maximum amplitude is equal to $3$ and locates
at $(0,\,0)$, and the minimum amplitude occurs at
$(\pm\frac{2\sqrt{3}}{3},\,\mp\frac{\sqrt{3}}{9})$.
 (b) Set $\beta_1=\frac{1}{2}$ in \eqref{1rw}, the profile of the
 rogue wave  becomes steeper when the absolute
 value of $\alpha_1$ become smaller. $\alpha_1=\frac{\sqrt{3
}}{2},\,\frac{\sqrt{2}}{2}$ and $\frac{1}{2}$ for the green thick line, the red thin line, the blue bold line respectively.}\label{fig.1_rogue wave}
\end{figure}

\begin{figure}[!htbp]
\centering
\raisebox{20 ex}{$|q^{[4]}|\,$}\subfigure[]{\includegraphics[height=5cm,width=5cm]{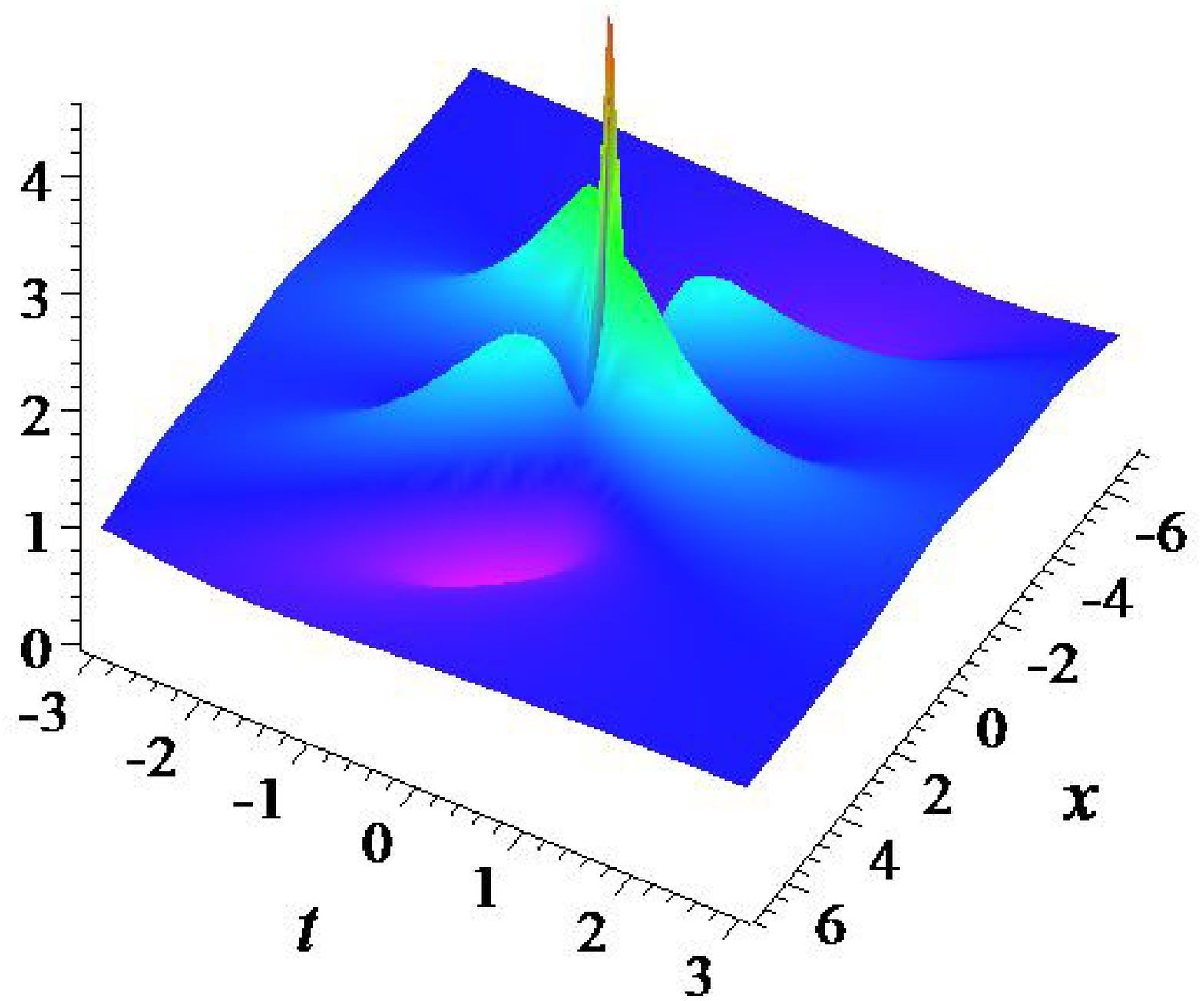}}
\qquad
\raisebox{20 ex}{$|q^{[4]}|\,$}\subfigure[]{\includegraphics[height=5cm,width=5cm]{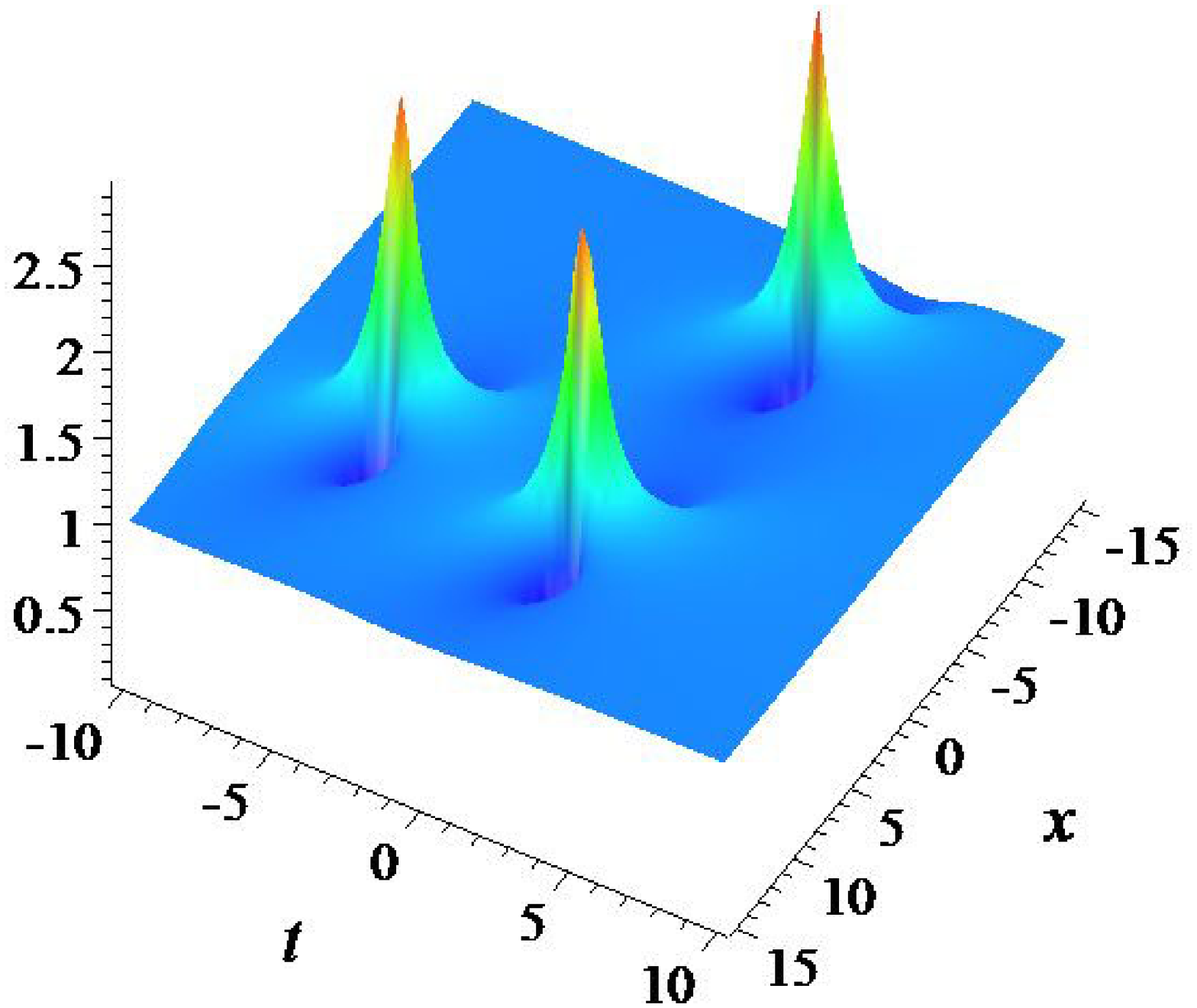}}
\caption{Two second order rogue waves in \eqref{2rw}.
(a) The fundamental pattern with $l_1=0$.
(b) The triangle pattern with $l_1=100$.}\label{fig.2_rw}
\end{figure}


\begin{figure}[!htbp]
\centering
\raisebox{20 ex}{$|q^{[6]}|\,$}\subfigure[]{\includegraphics[height=4cm,width=4cm]{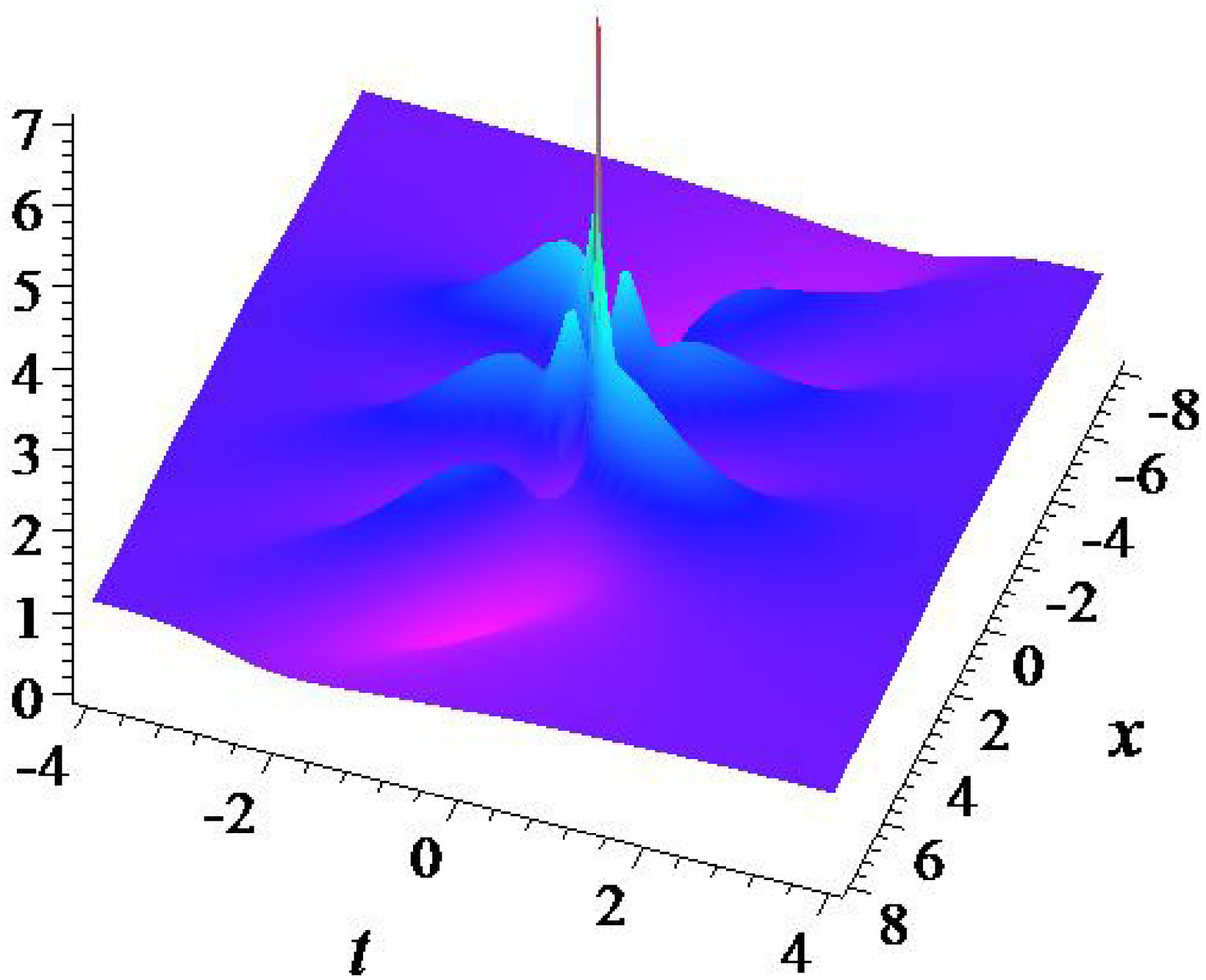}}
\quad
\raisebox{20 ex}{$|q^{[6]}|\,$}\subfigure[]{\includegraphics[height=4cm,width=4cm]{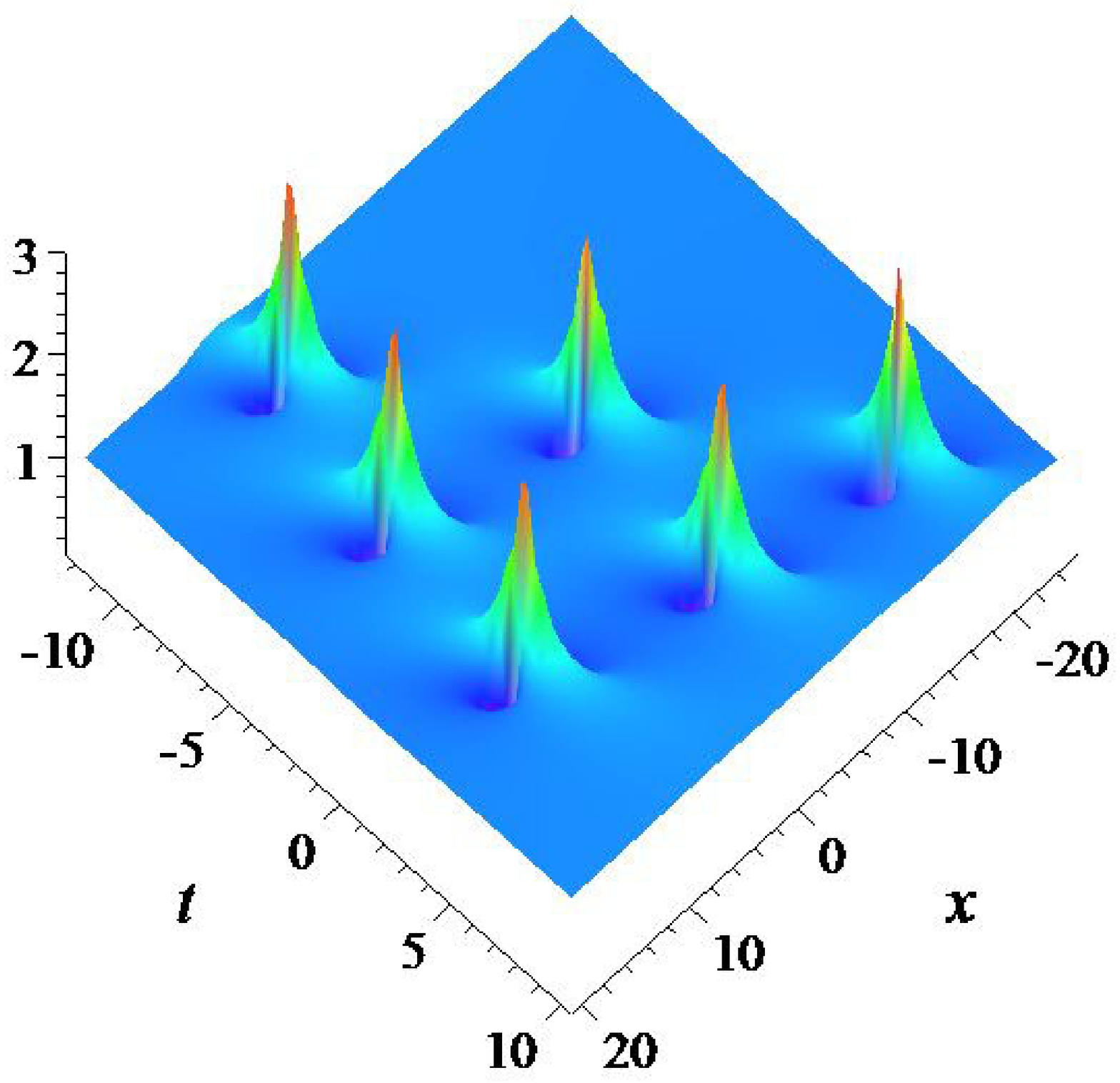}}
\quad
\raisebox{20 ex}{$|q^{[6]}|\,$}\subfigure[]{\includegraphics[height=4cm,width=4cm]{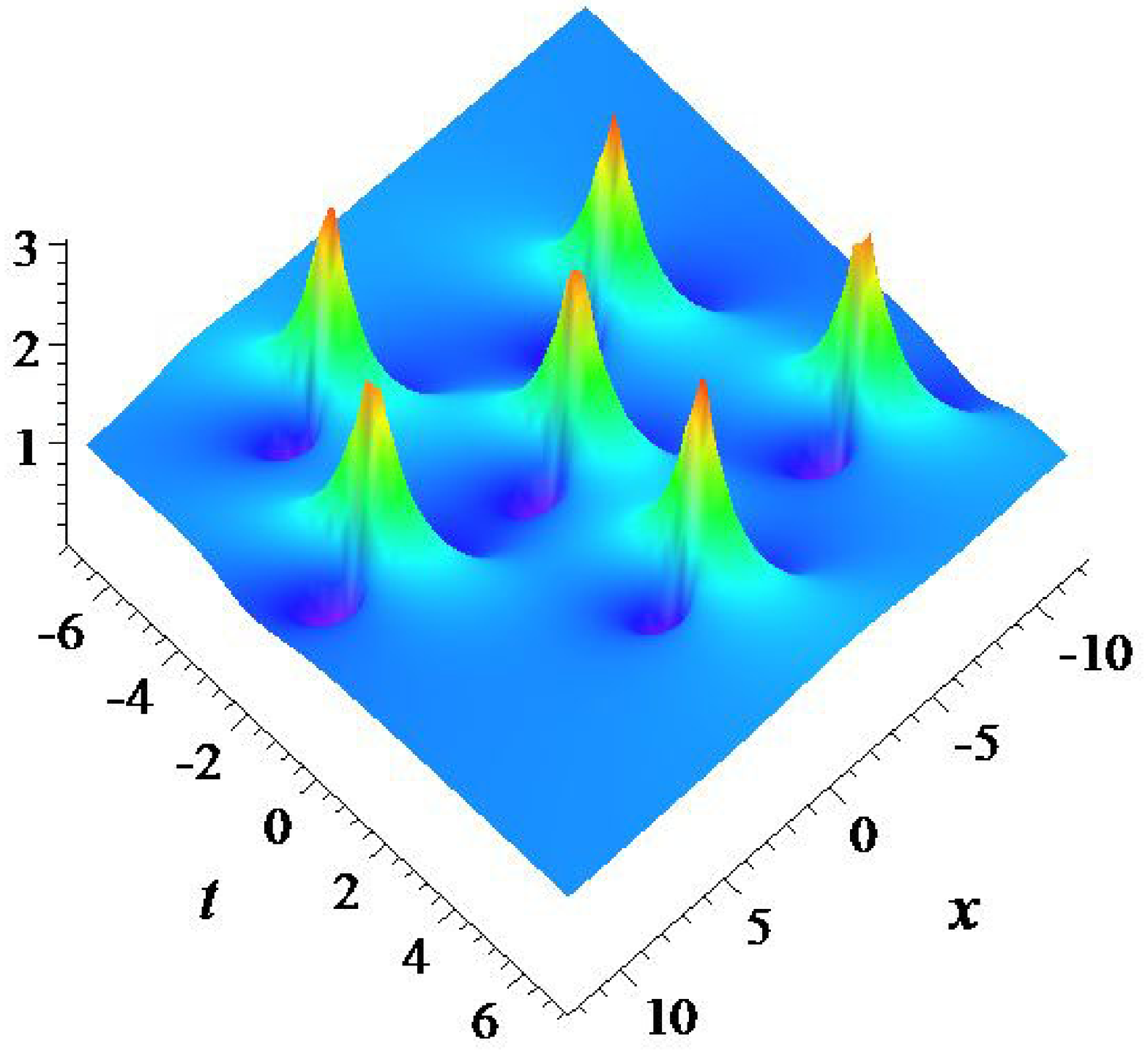}}
\caption{The third order rogue waves.
(a) The fundamental pattern with $l_0=l_1=l_2=0$;
(b) The triangular pattern with $l_1=100,\,l_0=l_2=0$;
(c) The circular pattern with $l_0=l_1=0,\,l_2=500$.}\label{fig.3_rw}
\end{figure}


\begin{figure}[!htbp]
\centering
\raisebox{20 ex}{$|q^{[8]}|\,$}\subfigure[]{\includegraphics[height=5cm,width=5cm]{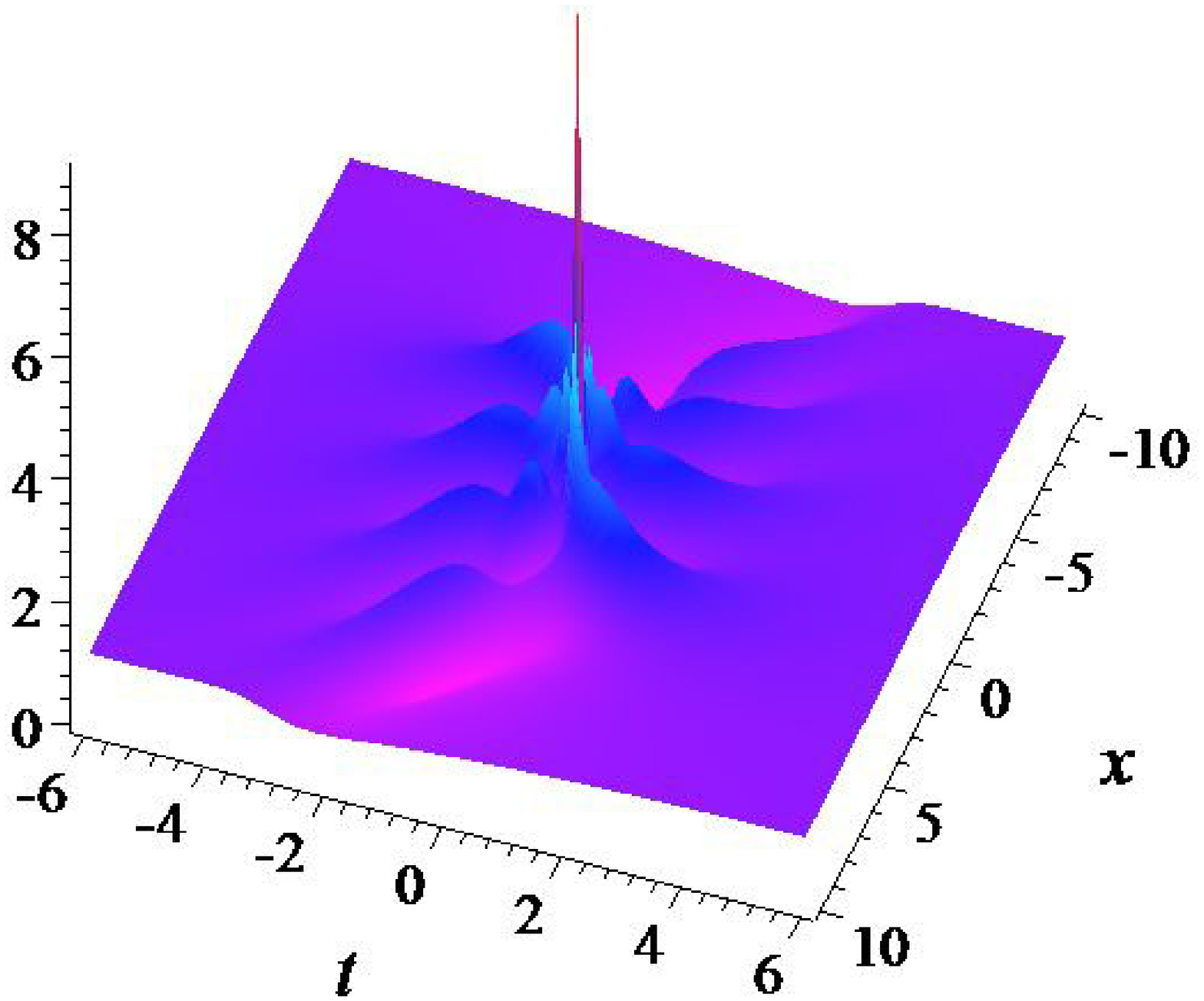}}
\quad
\raisebox{20 ex}{$|q^{[8]}|\,$}\subfigure[]{\includegraphics[height=5cm,width=5cm]{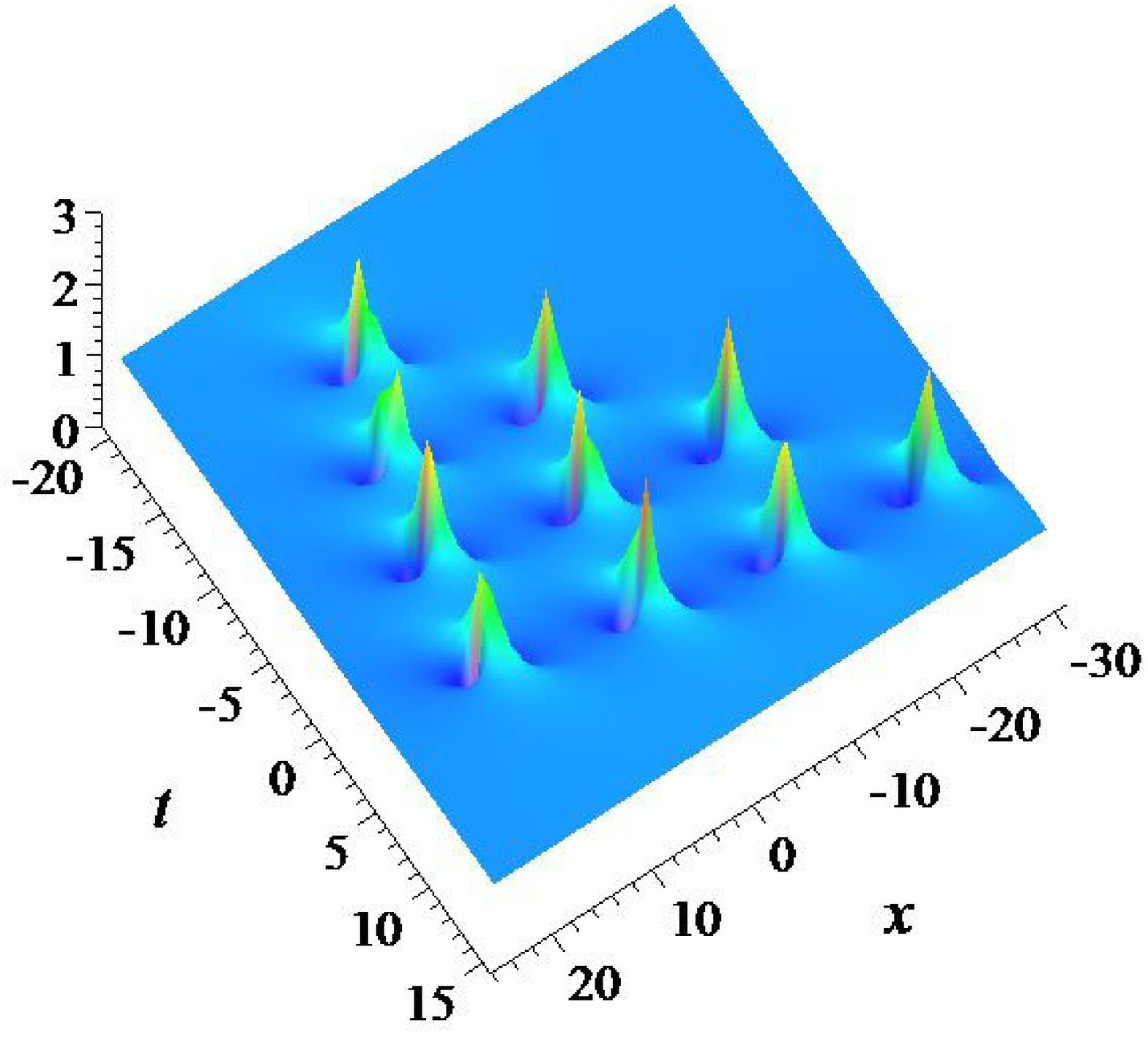}}
\\
\raisebox{20 ex}{$|q^{[8]}|\,$}\subfigure[]{\includegraphics[height=5cm,width=5cm]{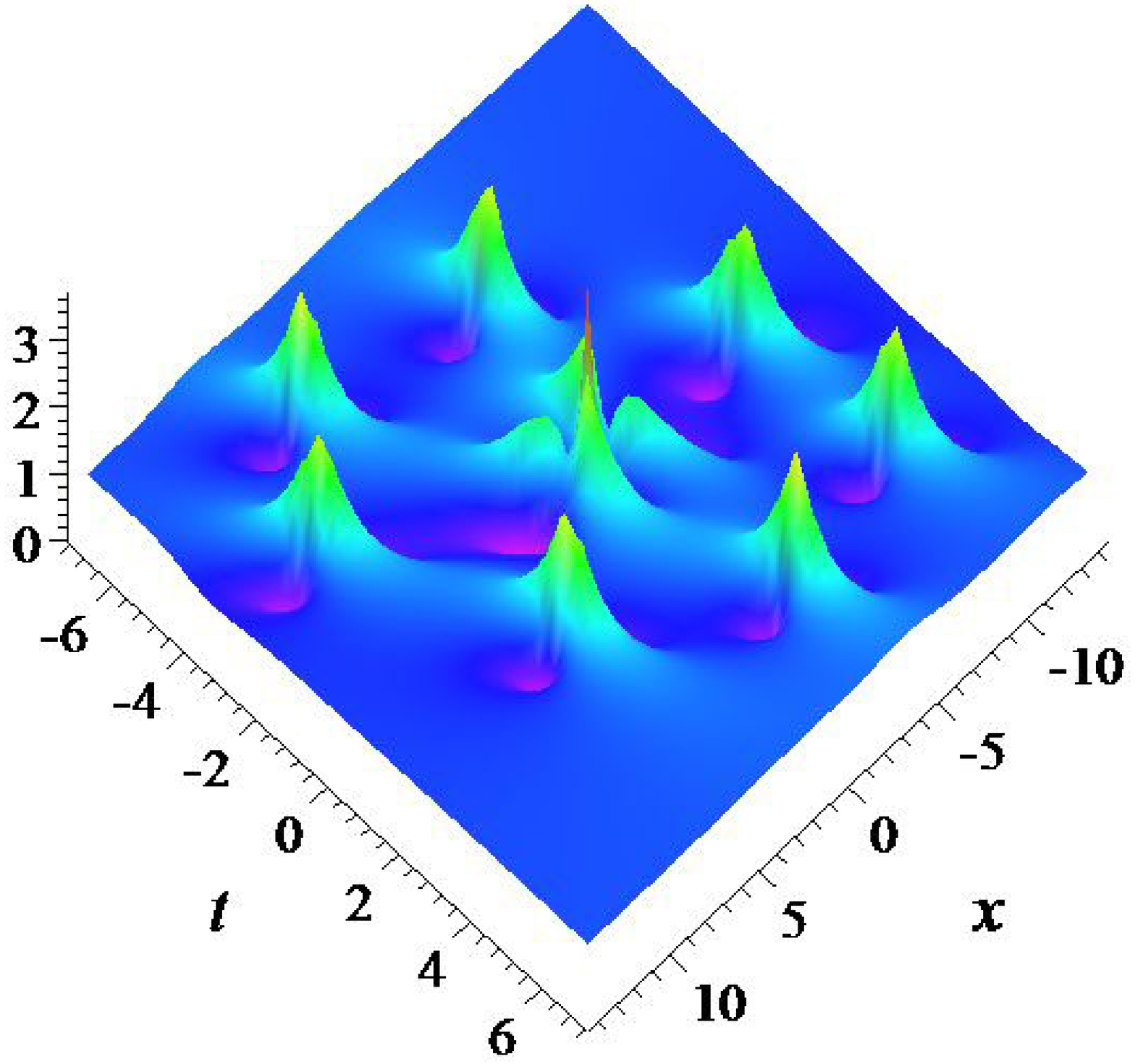}}
\quad
\raisebox{20 ex}{$|q^{[8]}|\,$}\subfigure[]{\includegraphics[height=5cm,width=5cm]{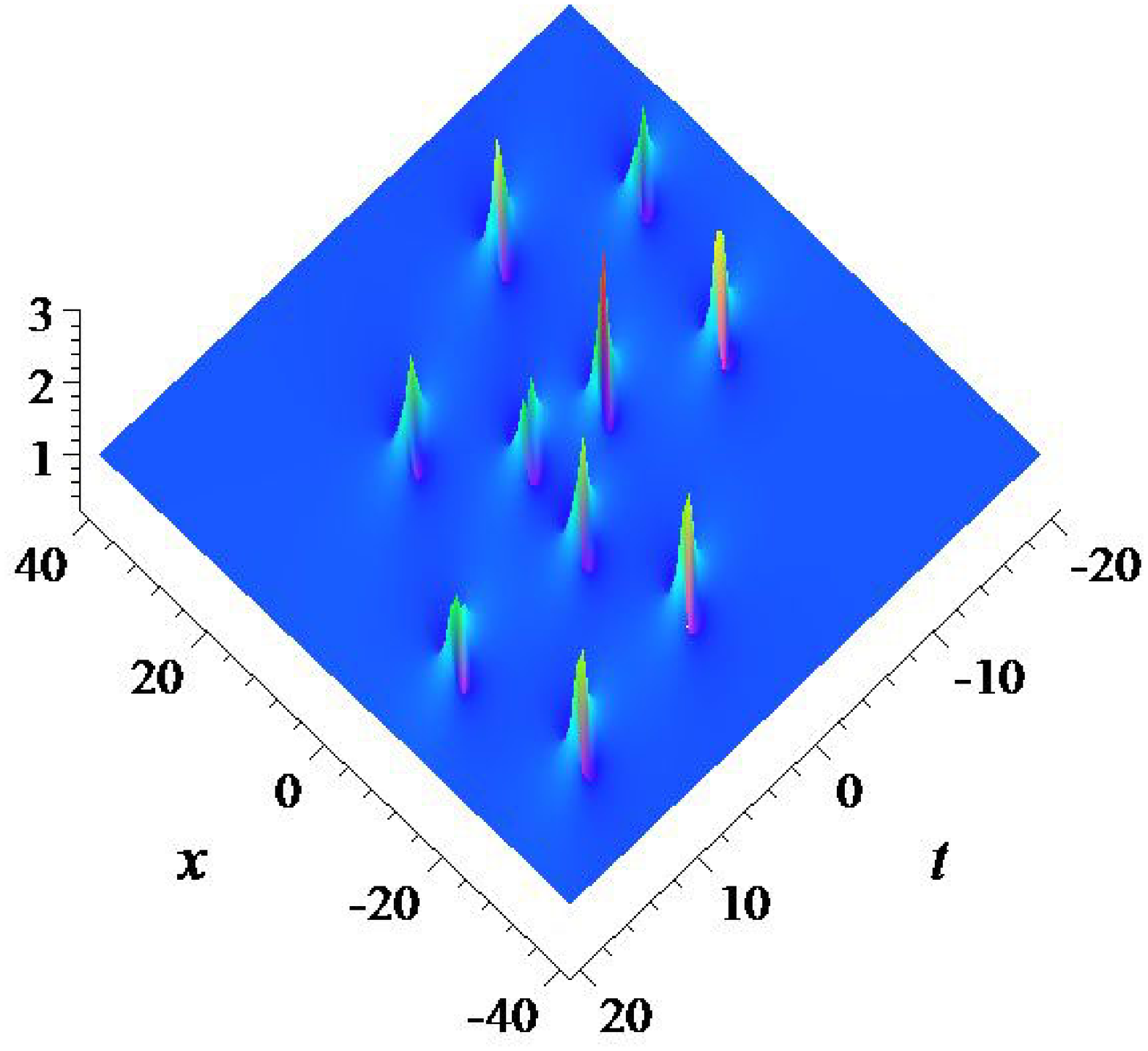}}
\caption{
The fourth order rogue waves.
(a) The fundamental pattern with $l_0=l_1=l_2=l_3=0$;
(b) The triangular pattern with $l_1=100,\,l_0=l_2=l_3=0$;
(c) The circular-fundamental pattern with $l_0=l_1=0=l_3=0,\,l_2=5000$;
(d) The circular-triangular pattern with
$l_0=l_2=0,l_1=50,\,l_3=5000000$.}\label{fig.4_rw}
\end{figure}


\begin{figure}[!htbp]
\centering
\raisebox{20 ex}{$|q^{[4]}|\,$}\subfigure[]{\includegraphics[height=5cm,width=5cm]{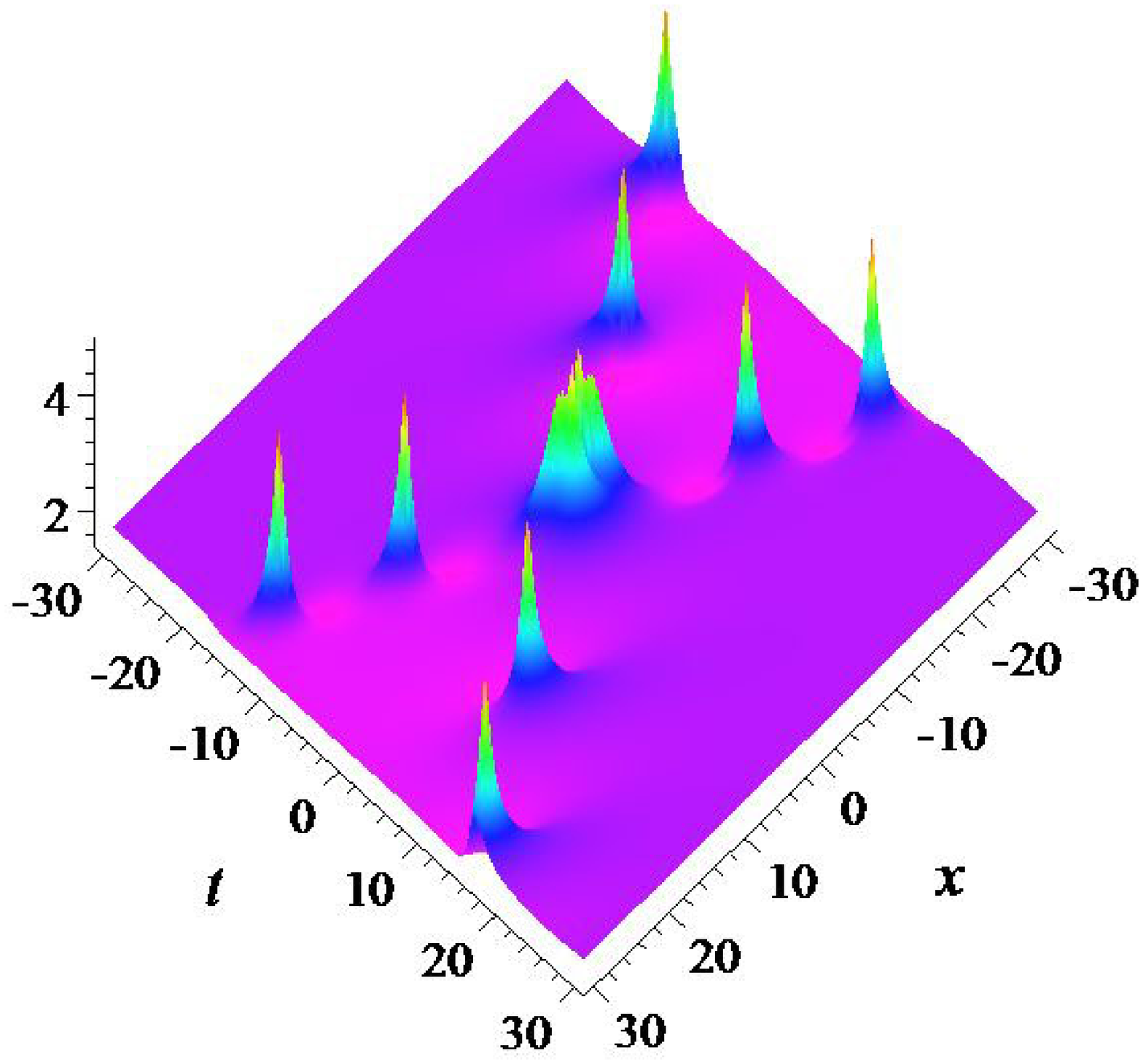}}
\qquad
\raisebox{20 ex}{$|q^{[4]}|\,$}\subfigure[]{\includegraphics[height=5cm,width=5cm]{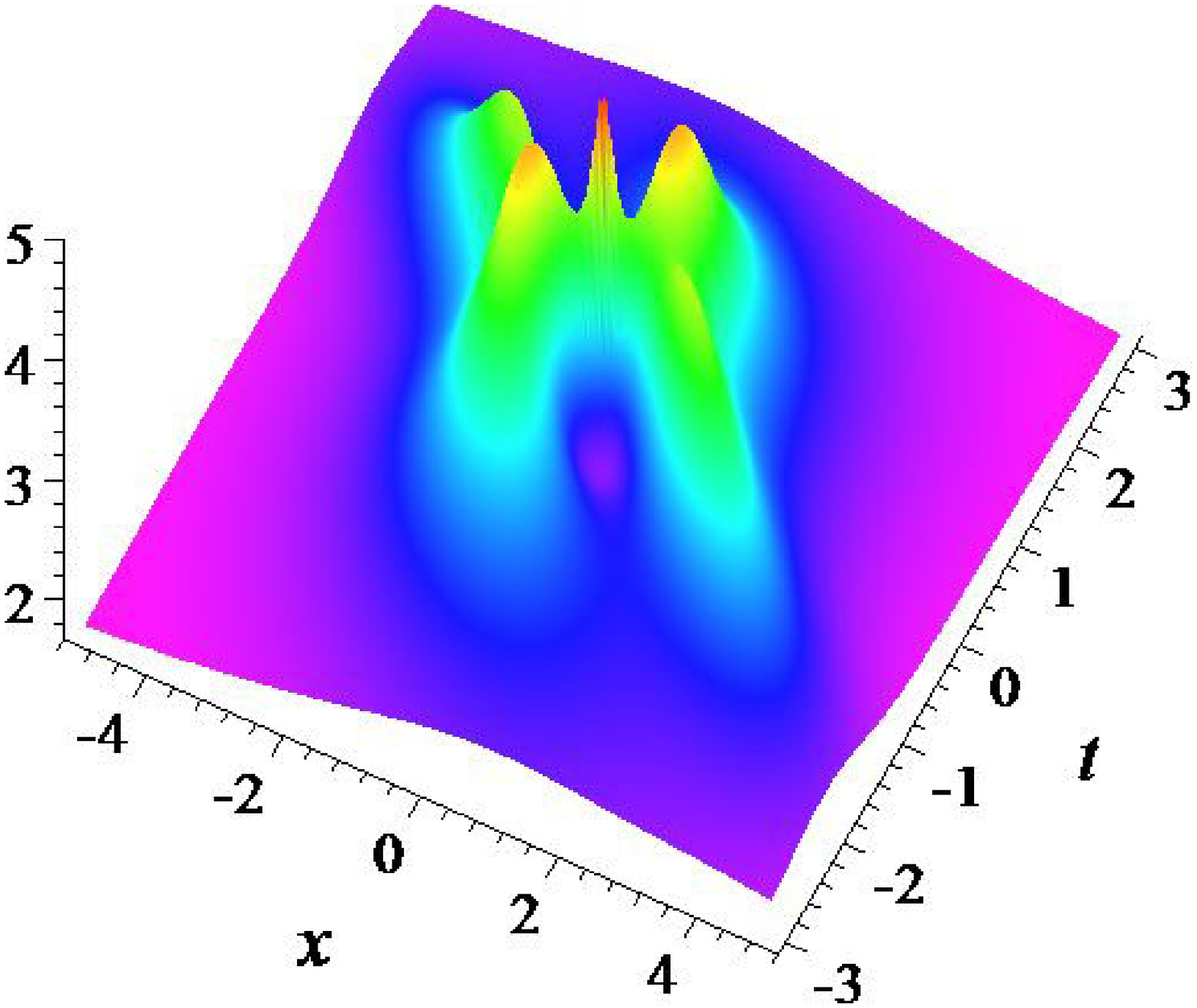}}
\caption{The ``fundamental pattern" of the second order breather.
(a) The  intersection of two breather solutions without phase shift;
(b) The central profile of the left panel, which is very similar
to the fundamental pattern of the second-order rogue wave in fig.
\ref{fig.2_rw}(a). } \label{fig.2_breather_1}
\end{figure}

\begin{figure}[!htbp]
\centering
\raisebox{20 ex}{$|q^{[4]}|\,$}\subfigure[]{\includegraphics[height=5cm,width=5cm]{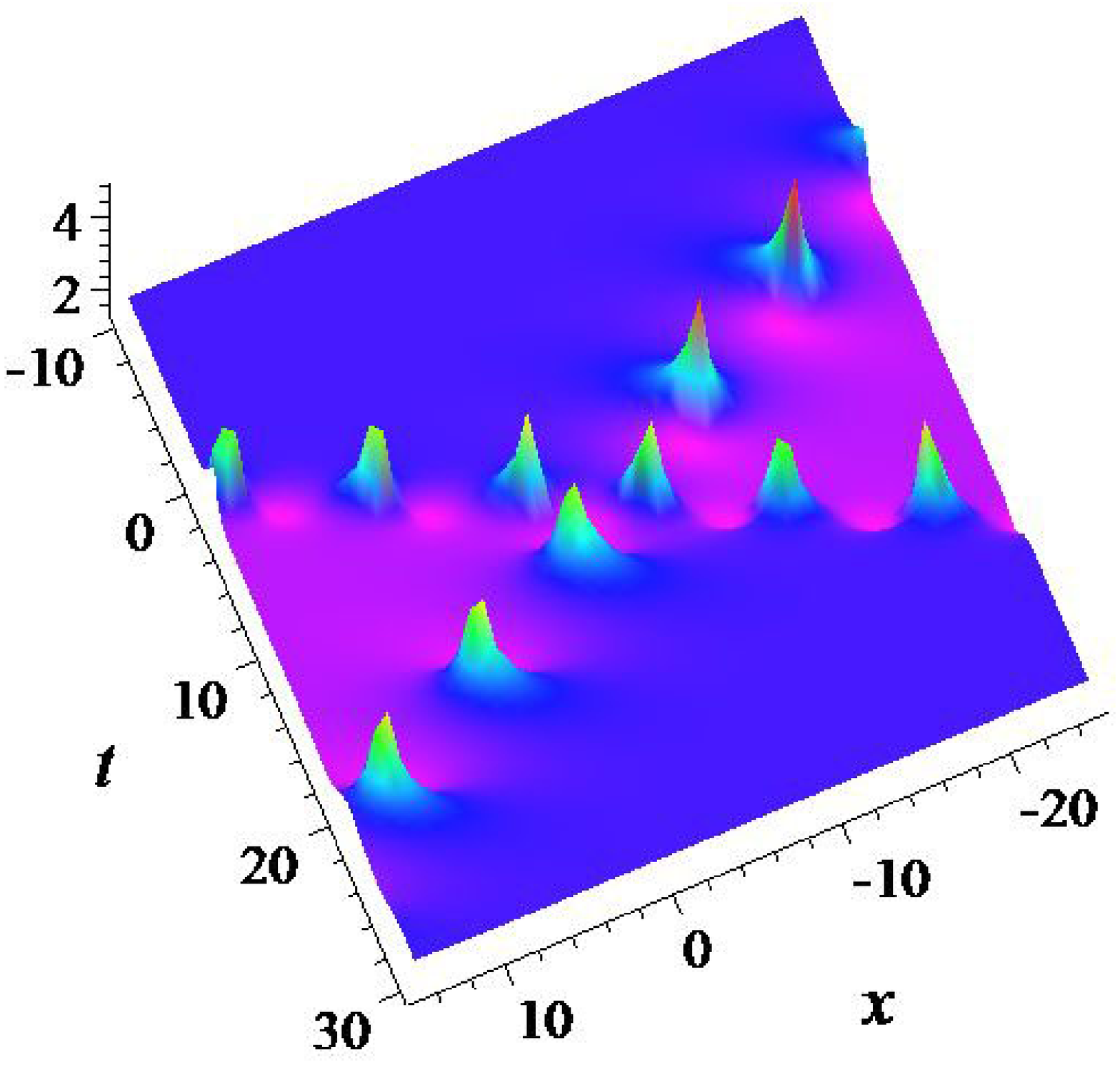}}
\qquad
\raisebox{20 ex}{$|q^{[4]}|\,$}\subfigure[]{\includegraphics[height=5cm,width=5cm]{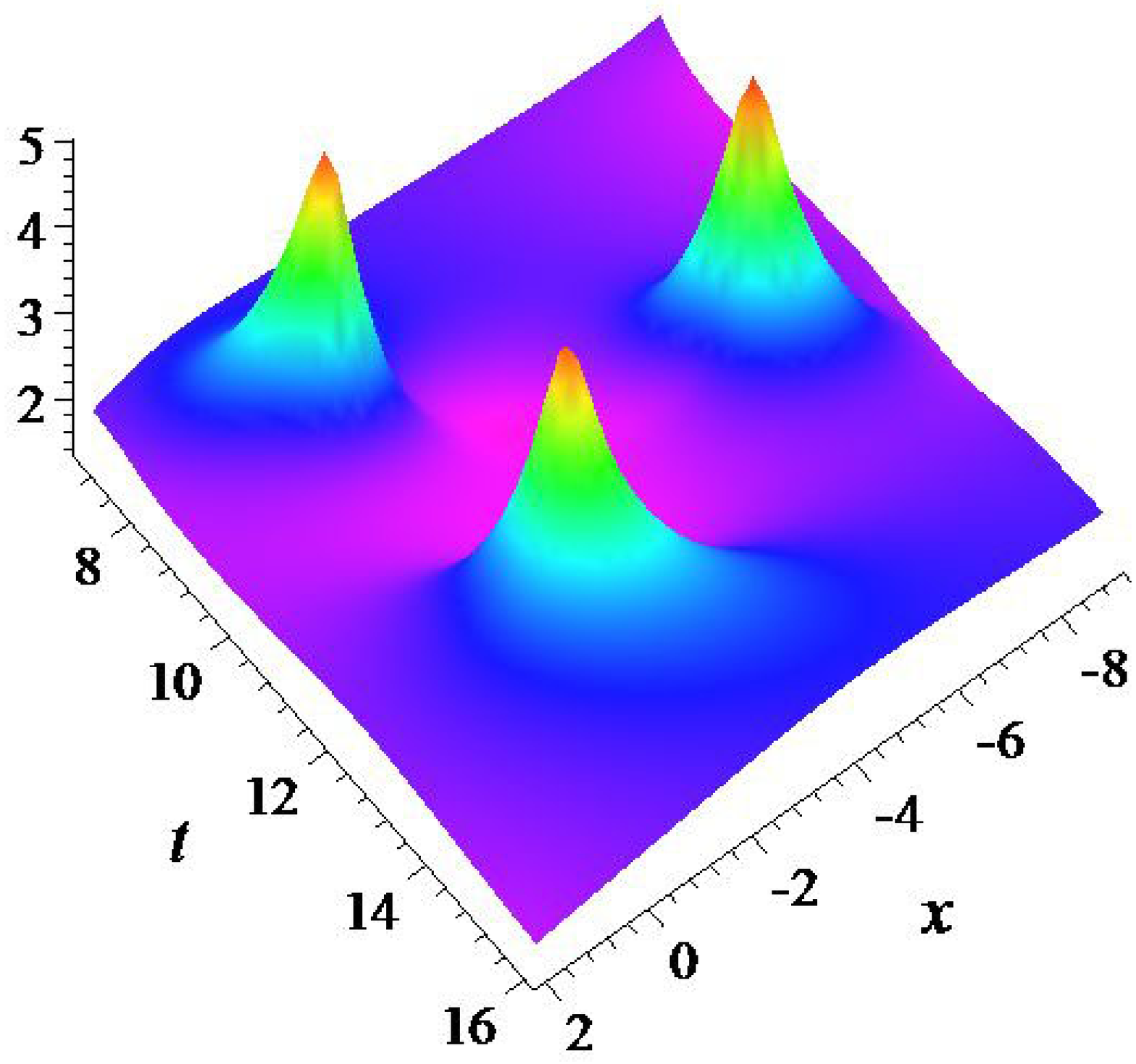}}
\caption{The ``triangle structure" of the second order breather.
(a) The intersection of two breather solutions with phase shift.
(b) The central profile of the left panel, which  is very similar
to the triangular pattern of the second order rogue wave in fig.
 \ref{fig.2_rw}(b).}\label{fig.2_breather_2}
\end{figure}


\end{document}